%% file: main.tex
\documentclass[11pt]{article}
\usepackage[utf8]{inputenc}
\usepackage[T1]{fontenc}
\usepackage[letterpaper,top=1in,bottom=1in,left=1in,right=1in,marginparwidth=1.75cm]{geometry}
\usepackage[nottoc,numbib]{tocbibind}
\usepackage{graphicx}
\usepackage{authblk}
\usepackage{complexity}
\usepackage{bbm}	
\usepackage{amsmath}
\usepackage[ruled,vlined,linesnumbered,noresetcount]{algorithm2e}
\usepackage{amsfonts}
\usepackage{amsmath}
\usepackage[colorlinks=true, linkcolor=red, urlcolor=blue, citecolor=gray]{hyperref}
\usepackage{amssymb}

\usepackage{amsthm}
\usepackage{mathtools}
\usepackage{capt-of}
\theoremstyle{plain}
\usepackage{bbm}
\usepackage{bm}
\usepackage{xcolor}
\usepackage[shortlabels]{enumitem} % To change enumeration characters
\usepackage{comment}
\usepackage{tikz}   
\usepackage{xcolor}
\usepackage{cite}
\usetikzlibrary{decorations.pathreplacing}
\usepackage[title]{appendix}
\usepackage[framemethod=tikz]{mdframed}
\usepackage{titlesec}
\usepackage{todonotes}
\usepackage[misc]{ifsym}
\usepackage{amssymb}
\usepackage{amsmath}
\usepackage{algorithmic}
\usepackage{commath}
\usepackage[capitalise]{cleveref}

\DeclareMathOperator*{\argmin}{arg\,min}
\let\polylog\relax
\DeclareMathOperator*{\polylog}{\mathrm{polylog}}
\DeclareMathOperator*{\tr}{\mathrm{tr}}

\makeatletter
\newtheorem*{rep@theorem}{\rep@title}
\newcommand{\newreptheorem}[2]{%
	\newenvironment{rep#1}[1]{%
		\def\rep@title{#2 \ref{##1}}%
		\begin{rep@theorem}}%
		{\end{rep@theorem}}}

\newcommand{\undersim}[1]{\mathrel{\mathpalette\@undersim{#1}}}
\newcommand{\@undersim}[2]{%
  \vcenter{%
    \ialign{%
      ##\cr
      $\m@th#1#2$\cr
      \noalign{\nointerlineskip\kern.2ex}
      $\m@th#1\sim$\cr
      \noalign{\kern-.4ex}
    }%
  }%
}
\newcommand{\gsim}{\undersim{>}}
\newcommand{\lsim}{\undersim{<}}
\makeatother

\newcommand{\eps}{\epsilon}
\newcommand{\wh}{\widehat}
\newcommand{\wt}{\widetilde}

\newcommand{\Z}{\mathbb{Z}}

\newcommand{\kstodo}[1]{\todo{KS: #1}}

\let\R\relax
\newcommand{\R}{\mathbb{R}}
\DeclareMathOperator{\rank}{rank}

\let\norm\relax
\newcommand{\norm}[1]{\|#1\|}

\SetKw{Break}{break}

\newcommand{\getempen}{\textsc{GetEmpiricalEnergy}}
\newcommand{\getlegsamp}{\textsc{GetLegalSample}}
\newcommand{\locatesig}{\textsc{Locate1Signal}}
\newcommand{\getempenk}{\textsc{GetKEmpiricalEnergy}}
\newcommand{\getlegsampk}{\textsc{GetKLegalSample}}
\newcommand{\locatesigk}{\textsc{LocateKSignal}}
\newcommand{\freqrecovcluster}{\textsc{FrequencyRecovery1Cluster}}
\newcommand{\freqrecovclusterk}{\textsc{FrequencyRecoveryKCluster}}
\newcommand{\onestage}{\textsc{OneStage}}

\newcommand{\hashtobins}{\textsc{HashToBins}}
\newcommand{\onegoodsamp}{\textsc{OneGoodSample}}

\newtheorem{theorem}{Theorem}
\newreptheorem{theorem}{Theorem}
\newtheorem{lemma}{Lemma}[section]
\newreptheorem{lemma}{Lemma}

\newtheorem{claim}[lemma]{Claim}

\newtheorem{definition}[lemma]{Definition}
\newtheorem{problem}{Problem}
\newtheorem{remark}[lemma]{Remark}

\newtheorem*{claim*}{Claim}
\newtheorem*{proposition*}{Proposition}
\newtheorem*{lemma*}{Lemma}
\newtheorem*{problem*}{Problem}

\input{bib_macros}

{\makeatletter
	\gdef\xxxmark{%
		\expandafter\ifx\csname @mpargs\endcsname\relax % in minipage?
		\expandafter\ifx\csname @captype\endcsname\relax % in figure/caption?
		\marginpar{xxx}% not in a caption or minipage, can use marginpar
		\else
		xxx % notice trailing space
		\fi
		\else
		xxx % notice trailing space-
		\fi}
	\gdef\xxx{\@ifnextchar[\xxx@lab\xxx@nolab}
	\long\gdef\xxx@lab[#1]#2{{\bf [\xxxmark #2 ---{\sc #1}]}}
	\long\gdef\xxx@nolab#1{{\bf [\xxxmark #1]}}
}

\title{Sublinear Time Low-Rank Approximation of Toeplitz Matrices}

\author[1]{ Cameron Musco}
\author[2]{Kshiteej Sheth}

\affil[1]{University of Massachusetts Amherst}
\affil[2]{EPFL}

\date{}

\begin{document}

\thispagestyle{empty}

\maketitle

%\Cam{I broke some lines above since it didn't look good to have one line of the title on the second line and have Kshiteej's name split. If you are going to fully write out MIT I think you should fully write out EPFL, rather than use the acryonym.}

%\Cam{All broken out equations need punctuation after them. Section/subsection titles should never have punctuation after them.}

%\Cam{When you put a title in a theorem/lemma it should go $begin\{theorem\}[TITLE HERE] \textbackslash label\{blah\}$. The title should not just be placed in parenthesis.}

%\Cam{Left and right commands should be used to properly size delimiters.}

%\Cam{In bib refs, if the title includes a proper noun like Fourier or Toeplitz, you need to write this as $\{F\}ourier$ to make sure it is capitalized.}

%\Cam{When you reference a numbered section, or theorem, or lemma, etc. it needs to be capitalized. Like Section 1, not section 1.}

%\Cam{When referencing equwations you need to use eqref not ref. So that they have parenthesis around them.}

%\Cam{upto is two words 'up to'}

%\Cam{IF you are using textbf to start a paragraph, you should use a noindent before it and a medskip to give line spacing before it. Dont use manual line breaks.}

\thispagestyle{empty}

%\Cam{Many of your sentences are run-ons. They are very long and contain multiple ideas. I single sentence should contain a single idea. Keep and eye out for run ons that should be broken up. I will flag a few as examples.}

%\Cam{upto is not a single word. It should be up to. I'm still seeing occurences of this.} 

%\Cam{Often when you write 'as per' it should be 'as in'. E.g. 'as per the setting' should be 'as in the setting'. }

%\Cam{cauchy-schwarz is a proper name and should be capitalized.}
\begin{abstract}
%\Cam{I've proofed abstract and I think it is good to go for abstract submission.}
We present a sublinear time algorithm for computing a near optimal low-rank approximation to any positive semidefinite (PSD) Toeplitz matrix $T\in \mathbb{R}^{d\times d}$, given noisy access to its entries. In particular, given entrywise query access to $T+E$ for an arbitrary noise matrix $E\in \mathbb{R}^{d\times d}$, integer rank $k\leq d$, and error parameter $\delta>0$, our algorithm runs in time $\poly(k,\log(d/\delta))$ and outputs (in factored form) a Toeplitz matrix $\wt T \in \R^{d \times d}$ with rank $\poly(k,\log(d/\delta))$ satisfying, for some fixed constant $C$, 
\begin{equation*}
    \|T-\wt T\|_F \leq C \cdot \max\{\|E\|_F,\|T-T_k\|_F\} + \delta  \cdot \|T\|_F.
\end{equation*}

Here $\|\cdot \|_F$ is the Frobenius norm and $T_k$ is the best (not necessarily Toeplitz) rank-$k$ approximation to $T$ in the Frobenius norm, given by projecting $T$ onto its top $k$ eigenvectors. 

Our robust low-rank approximation primitive can be applied in several settings. When $E = 0$, we obtain the first sublinear time near-relative-error low-rank approximation algorithm for  PSD Toeplitz matrices, resolving the main open problem of Kapralov et al. SODA `23, which gave an algorithm with sublinear query complexity but exponential runtime. Our algorithm can also be applied to approximate the unknown Toeplitz covariance matrix of a multivariate Gaussian distribution, given sample access to this distribution. By doing so, we resolve an  open question of Eldar et al. SODA `20, improving the state-of-the-art error bounds and achieving a polynomial rather than exponential (in the sample size) runtime.

Our algorithm is based on applying sparse Fourier transform techniques to recover a low-rank Toeplitz matrix using its Fourier structure. Our key technical contribution is the first polynomial time algorithm for \emph{discrete time off-grid} sparse Fourier recovery, which may be of independent interest. We also contribute a structural heavy-light decomposition result for PSD Toeplitz matrices, which allows us to apply this primitive to low-rank Toeplitz matrix recovery.
\end{abstract}

%\pagebreak
%\tableofcontents

%\Cam{General comment: $\wt O(\cdot)$ should hide log factors in the argument. So like $\wt O(n)$ can hide log factors in $n$. But $\wt O(1)$ isn't generally a good notation, since it is not claer what parameters it is hiding log factors in. We should clean this up throughout. I think often we are hiding $\log d$ factors even when there is no $d$ in the argument. This includes in the main theorem statements and abstract.}

%\Cam{I think we would help our chances a lot if we justified why we need both bicriteria approximation and additive error. Even better if we can show that our tradeoff is tight, but that might be hard. One place to start would be to look at rank-1 approximation, where I think we can probably determine the best Toeplitz low-rank approximation in closed form.} 

\clearpage

\thispagestyle{empty}
 
\tableofcontents

\clearpage

\pagenumbering{arabic}

\input{intro}

\input{preliminaries}
%\input{existence}
%\input{approx_algorithm}
% Polynomial time alg
\input{filter_functions}
\input{single_cluster}
\input{multi_cluster}

%\section{Conclusion.}

\section{Acknowledgements.}
Cameron Musco's work is supported by an Adobe Research grant and NSF Grants No. 2046235 and
No. 1763618. The authors acknowledge useful discussions with Michael Kapralov, Mikhail Makarov, and Hannah Lawrence.

\bibliographystyle{alpha} 
\bibliography{refs}

% \appendix
% \input{filter_functions}
\end{document}

%% file: bib_macros.tex
  \usepackage{nth}
  \usepackage{intcalc}

%% file: intro.tex
% !TEX root = ./main.tex

%\Cam{Captialization of section titles should be made consistent. One some of them we use caps for all words, on some for just the first word.}

\section{Introduction}

%\Cam{We need more citation in this first paragraph. They can be drawn from prior work}
A Toeplitz matrix $T\in \mathbb{R}^{d\times d}$ is constant along each of its diagonals. I.e., $T_{i,j}=T_{k,l}$ for all $i,j,k,l$ with $i-j=k-l$. These highly structured  matrices arise in many fields, including signal processing, scientific computing, control theory, approximation theory, and machine learning -- see \cite{Bunch:1985ti} for a survey. In particular, Toeplitz matrices often arise as the covariance matrices of stationary signals, when the covariance structure is shift invariant. I.e., when the covariance between measurements only depends on their distance in space or time \cite{Gray:2006ta}. A row-reversed Toeplitz matrix is known as a Hankel matrix. Such matrices also find broad applications \cite{Fazel:2013vw,Munkhammar:2017ud,Ghadiri:2023tv}. %Toeplitz matrices also arise in control theory, approximation theory, the study of orthogonal polynomials and partial differential equations \cite{Bunch:1985ti}.

Given their importance, significant work has studied fast algorithms for basic linear algebraic tasks on Toeplitz matrices. A $d \times d$ Toeplitz matrix can be multiplied by a vector in just  $O(d\log d)$ time using the fast Fourier transform. Toeplitz linear systems can be solved in $O(d^2)$ time exactly using Levinson recursion \cite{Golub:2013wp}, and to high-precision in $O(d\cdot \polylog d)$ time using randomization  \cite{XiXiaCauley:2014,XiaXiGu:2012}. A full eigendecomposition of a symmetric Toeplitz matrix can be computed in $O(d^2\cdot \polylog d)$ time \cite{pan1999complexity}. 

\subsection{Sublinear query algorithms for Toeplitz matrices.} 

Recent work has focused on algorithms for Toeplitz matrices with complexity scaling \emph{sublinearly} in the dimension $d$ \cite{Abramovich:1999vs,Chen:2015wz,Qiao:2017tp,Wu:2017ub,Lawrence:2020ut,eldar2020toeplitz,kapralov2022toeplitz}. Kapralov et al. \cite{kapralov2022toeplitz} study low-rank approximation of positive semidefinite (PSD) Toeplitz matrices. They show that by accessing just $\poly(k,\log(d/\delta),1/\epsilon)$ entries of a PSD Toeplitz matrix $T \in \R^{d \times d}$, one can compute (in factored form) a symmetric Toeplitz matrix $\wt T$ with rank $\tilde O(k \log(1/\delta)/\epsilon)$ such that:\footnote{Throughout we use $\tilde O(\cdot)$ to hide polylogarithmic factors in the dimension $d$ and in the argument.}  
\begin{align}\label{eq:nearRelative}
\|T-\wt{T}\|_F \leq (1+\eps)\|T-T_k\|_F + \delta\|T\|_F,
\end{align}
where $\norm{M}_F = \sqrt{\sum_{i=1}^d \sum_{j=1}^d M_{ij}^2}$ is the Frobenius norm and $T_k = \underset{B:\rank(B)\leq k}{\argmin}\|T-B\|_F$ is the best rank-$k$ approximation to $T$ in the Frobenius norm, given by projecting $T_k$ onto its top $k$ eigenvectors. %$\wt T$ gives a near-optimal \emph{structure preserving} low-rank approximation of $T$. %Observe that since the dependence on $\delta$ is logarithmic, the optimality guarantee is near relative error.

Observe that $T_k$ may not itself be Toeplitz and it is not a priori clear that Toeplitz $\wt T$ satisfying \eqref{eq:nearRelative} even exists. The key technical contribution of \cite{kapralov2022toeplitz} is to prove that it does, and further that $\wt T$ can be recovered using sample efficient off-grid sparse Fourier transform techniques. Unfortunately, despite its sample efficiency, the algorithm of \cite{kapralov2022toeplitz} is computationally inefficient, with runtime that is  exponential in $\tilde O(k \polylog(d)/\epsilon)$  --  i.e., at least $d^{\tilde O(k/\epsilon)}$. The main open question left by their work is if this runtime can be improved, giving a sublinear time, not just a sublinear query, algorithm. %Roughly, the columns of $\wt T$ are spanned by $\tilde O(k/\epsilon)$ off grid frequencies, and the algorithm must perform a brute-force-search to identify these frequencies.

Eldar et al. \cite{eldar2020toeplitz} study the related problem of recovering an (approximately) low-rank PSD Toeplitz covariance matrix $T \in \R^{d \times d}$ given independent samples from the $d$-dimensional Gaussian $\mathcal{N}(0,T)$. Low-rank Toeplitz covariance estimation is widely applied in signal processing, including in  direction of arrival (DOA) estimation  \cite{krim1996two}, spectrum sensing for cognitive radio \cite{ma2009signal,cohen2018analog}, and medical and radar image processing \cite{snyder1989use,ruf1988interferometric,fuhrmann1991application,brookes2008optimising,asl2012low,cohen2018sparse}. 
%
%We now turn to the work of \cite{eldar2020toeplitz} who also developed sublinear query algorithms in the context of Toeplitz matrices. The work of \cite{eldar2020toeplitz} is concerned with the following problem - given $s$ i.i.d. samples $x^1,x^2,\ldots, x^s$ from a $d$ dimensional Gaussian $\mathcal{N}(0,T)$ with covariance $T$ that is Toeplitz, estimate $T$ up to some error threshold. This problem statement is not only more realistic since the algorithm only has \emph{noisy} implicit access to $T$ via the i.i.d. samples as opposed to the setup of Theorem \ref{thm:sublinear_lra_soda2023}, but also it is directly relevant to practical applications. Toeplitz covariance structure emerges for example when the random measurement vector corresponds to measurements on a spatial or temporal grid, and the covariance between two coordinates just depends on the difference between them. Fast algorithms for this problem of estimating Toeplitz covariance matrices have many applications ranging from Direction of Arrival (DOA) estimation for signals received from antenna arrays \cite{krim1996two}, spectrum sensing for cognitive radio \cite{ma2009signal,cohen2018analog}, medical and radar image processing \cite{snyder1989use,ruf1988interferometric,fuhrmann1991application,brookes2008optimising,asl2012low,cohen2018sparse}. 
Motivated by these applications, Eldar et al. \cite{eldar2020toeplitz} consider two relevant sample complexity measures: the vector sample complexity (VSC), which is the number of samples (each a vector in $\R^d$) taken from the distribution $\mathcal{N}(0,T)$, and the entry sample complexity (ESC), which is the number of entries read from each vector sample. % and runtime. We encapsulate the definitions of these complexity measures in the following definition, which we will refer to recurringly in the paper. 
%\begin{definition}\label{def:vsc_and_esc}
%Given $s$ i.i.d. samples $s$ of the form $x^1,x^2,\ldots,x^s \sim \mathcal{N}(0,T)$ consider any algorithm which estimates $T$ up to some specified error tolerance using these samples. The Vector sample complexity (VSC) is the number of i.i.d. samples $s$ of the form $x^1,x^2,\ldots,x^s \sim \mathcal{N}(0,T)$ an algorithm uses and the Entry sample complexity (ESC) is the number of entries the algorithm reads from each sample $x^i$ for $i\in [s]$.  
%\end{definition}
They show that with $\poly(k,\log d,1/\epsilon)$ VSC and ESC, it is possible to return $\wt T$ satisfying with high probability:\footnote{Throughout, we let $f(\cdot) \lsim g(\cdot)$ denote that $f(\cdot) \le c \cdot g(\cdot)$ for some fixed constant $c$.}
%The main result of \cite{eldar2020toeplitz} is to achieve a sublinear, i.e. $o(d)$, VSC and ESC, this is formalized below.
%\begin{theorem}[Theorem 2.10 in \cite{eldar2020toeplitz}]\label{thm:elmm_soda20}
%There is an algorithm such given i.i.d samples $x^1,x^2,\ldots, x^s\sim \mathcal{N}(0,T)$ for PSD Toeplitz $T\in \mathbb{R}^{d\times d}$ and parameters $k,\epsilon>0$, it has VSC $s=\wt O(k+1/\epsilon^2)$, ESC $\wt O(k^2\log(1/\epsilon))$ and returns a Toeplitz $\wt T$ satisfying the following with probability $0.9$,
\begin{equation}\label{eq:sampleBound}
    \|T-\wt T\|_2 \lsim \sqrt{\|T-T_k\|_2 \cdot  \tr(T)+ \frac{\|T-T_k\|_F  \cdot \tr(T)}{k}} + \epsilon \|T\|_2,
\end{equation}
where $\tr(T)$ is the trace and $\norm{T}_2$ is the spectral norm.
%\end{theorem}
While the above error bound may seem non-standard,  Eldar et al. show that it implies fairly strong bounds when $T$ has low stable-rank. Observe that when $T$ is exactly rank-$k$, the error bound becomes $\epsilon \norm{T}_2$. As in \cite{kapralov2022toeplitz}, a key drawback is that the algorithm of \cite{eldar2020toeplitz} has runtime scaling exponentially in  $\tilde O(k \polylog d)$. Further,   the output matrix $\wt T$ is not itself guaranteed to be low-rank.

\subsection{Our contributions.}

In this work, we give a \emph{sublinear time} algorithm for computing a near-optimal low-rank approximation of a PSD Toeplitz matrix given noisy access to its entries. Our robust low-rank approximation primitive can be applied in the settings of both \cite{kapralov2022toeplitz} and \cite{eldar2020toeplitz}, yielding the first sublinear time algorithms for standard PSD Toeplitz low-rank approximation and low-rank Toeplitz covariance estimation. Our main result is:
\begin{theorem}[Robust Sublinear Time Toeplitz Low-Rank Approximation]\label{thm:sublinear_time_recovery}
	Let $T\in \mathbb{R}^{d\times d}$ be a PSD Toeplitz matrix, $E\in \mathbb{R}^{d\times d}$ be an arbitrary noise matrix, $\delta>0$ be an error parameter, and $k$ be an integer rank parameter. There exists an algorithm that, given query access to the entries of $T+E$, runs in $\poly(k,\log(d/\delta))$ time and outputs a representation of symmetric Toeplitz matrix $\wt T$ with rank $\poly(k,\log(d/\delta))$ that satisfies, with probability at least $0.9$,
	\begin{equation*}
		\|T- \wt T\|_F \lsim \max\{\|E\|_F,\|T-T_k\|_F\} + \delta \|T\|_F,
	\end{equation*}
	where $T_k = \argmin_{B: \rank(B)\leq k}\|T-B\|_F$ is the best rank-$k$ approximation to $T$.
\end{theorem}
Observe that given our runtime, which depends just poly-logarithmically on the input dimension $d$, it is not possible to output $\wt T \in \R^{d \times d}$ explicitly. Thus, our algorithm outputs a compressed representation of $\wt T$. To do so, we use the well known Vandermonde decomposition theorem, which states that any Toeplitz matrix can be diagonalized by a Fourier matrix \cite{cybenko1982moment}. We can show that $\wt T$ in particular, which has rank $r = \poly(k,\log(d/\delta))$, can be written as  $FDF^*$ where $F \in \mathbb{C}^{d\times r}$ is a Fourier matrix, with $j^{th}$ column given by $[1, e^{2\pi i \cdot f_j}, e^{2\pi i \cdot 2f_j},\ldots, e^{2\pi i  \cdot (d-1) f_j}]$ for some frequency $f_j$, and $D \in \R^{r \times r}$ is diagonal. Our algorithm outputs the frequencies $f_1,\ldots,f_r$ and diagonal entries $D_{1,1},\ldots,D_{r,r}$, which fully determine $\wt T$.

As an immediate consequence of Theorem \ref{thm:sublinear_time_recovery} applied with $E=0$, we obtain the following sublinear time constant factor low-rank approximation algorithm for PSD Toeplitz matrices, which resolves the main open problem of \cite{kapralov2022toeplitz} for the case of constant factor approximation.
\begin{theorem}[Sublinear Time Toeplitz Low-Rank Approximation]\label{thm:sublinear_time_lowrankapprox}
Let $T\in \mathbb{R}^{d\times d}$ be a PSD Toeplitz matrix, $k$ be an integer rank parameter, and $\delta>0$ be an error parameter. There exists an algorithm that, given query access to entries of $T$, runs in  $\poly(k,\log(d/\delta))$ time and outputs a representation of symmetric Toeplitz matrix $\wt T$ with rank $\poly(k,\log(d/\delta))$ that satisfies with probability at least 0.9,
	\begin{equation*}
		\|T- \wt T\|_F \lsim \|T-T_k\|_F + \delta \|T\|_F.
	\end{equation*}
	%where $T_k = \argmin_{B: \rank(B)\leq k}\|T-B\|_F$ is the best rank-$k$ approximation to $T$.
\end{theorem}
We can further apply Theorem \ref{thm:sublinear_time_recovery} to the Toeplitz covariance estimation setting of \cite{eldar2020toeplitz}. Here, we set $E=\frac{1}{s} \cdot XX^T -T$ where $X \in \R^{d \times s}$ has $s$ columns sampled i.i.d. from the Gaussian distribution $\mathcal{N}(0,T)$. That is, $E$ is the error between the true covariance matrix $T$ and a sample covariance  matrix $\frac{1}{s} XX^T$ that our algorithm can access. Using similar ideas to \cite{eldar2020toeplitz}, we show that for $s = \tilde O(k^4/\epsilon^2)$, $\norm{E}_F \lsim \sqrt{\|T-T_k\|_2 \tr(T)+ \frac{\|T-T_k\|_F \tr(T)}{k}} + \epsilon \|T\|_2$. Combined with Theorem \ref{thm:sublinear_time_recovery} applied with $\delta = \epsilon/\sqrt{d}$ so that $\delta \norm{T}_F \le \epsilon \norm{T}_2$, this gives:
\begin{theorem}[Sublinear Time Toeplitz Covariance Matrix Estimation]\label{thm:sublinear_time_covariance_estimation}
	Let $T\in \mathbb{R}^{d\times d}$ be a PSD Toeplitz matrix, $k$ be an integer rank parameter, and $\epsilon > 0$ be an error parameter. There is an algorithm that given i.i.d. samples $x^1,\ldots, x^s \sim \mathcal{N}(0,T)$ for $s=\tilde O(k^4/\epsilon^2)$, runs in $\poly(k,\log (d/\epsilon),1/\epsilon)$ time and outputs a representation of symmetric Toeplitz $\wt T$with rank $\poly(k,\log(d/\epsilon))$ satisfying, with probability at least 0.9,
	\begin{equation*}
		\|T- \wt T\|_F \lsim \sqrt{\|T-T_k\|_2 \tr(T)+ \frac{\|T-T_k\|_F \tr(T)}{k}} + \epsilon \|T\|_2.
	\end{equation*}
	Further, the algorithm has entry sample complexity (ESC) $\poly(k,\log (d/\epsilon))$ --  it reads just  $\poly(k,\log (d/\epsilon))$ entries of each vector sample $x^i$.
\end{theorem}
Observe that the error guarantee of
Theorem \ref{thm:sublinear_time_covariance_estimation} is at least as strong, and can potentially be much stronger, than that of \eqref{eq:sampleBound}, since the bound is on $\norm{T-\wt T}_F$ rather than $\norm{T - \wt T}_2$. At the same time, our algorithm achieves the same vector and entry sample complexities as \cite{eldar2020toeplitz} up to $\poly(k,\log (d/\epsilon),1/\epsilon)$ factors, while running in sublinear time. Moreover, our output $\wt T$ is guaranteed to be low-rank, while the output of \cite{eldar2020toeplitz} is not.  %\Cam{Not sure this forward ref is nreally needed here? As we are just summarizing out contributions, it seems too specific and out of place? I would delete the full sentence. See my comment on a road map at the bottom of Sec 2.}

%At its heart, our algorithm solves the following problem.  Consider a function $x(t)$ where $x:\mathbb{Z}\to \mathbb{C}$, and $x$ is $k$-sparse in Fourier domain. Then, assuming noisy sample access to the function at any $t\in [d]=\{0,1,\ldots,d\}$, recover the $k$ potentially ``off-grid'' frequencies in the support of $\wh x(f) $ approximately in $\poly(k)$ time. The continuous time version of this problem, that is the $k$-Fourier sparse function $x(t)$ is from $x:\mathbb{R}\to \mathbb{C}$ and one has noisy sample access to it for any $t\in [0,d]$, was studied and solved in \cite{chen2016fourier}. The discrete time version we consider poses additional challenges when compared to \cite{chen2016fourier} due to the constraint that we can only access the function at integer points in $[0,d]$. The formal statement of this result is presented in section 4 as Lemma \ref{lem:heavy_freq_recovery}.

%\Cam{Only put periods after *subsections* since they are inlined. Not full section titles.}
\section{Technical overview}\label{sec:approach}
In this section, we sketch the ideas behind the proof of our main result, the sublinear time robust Toeplitz low-rank approximation algorithm of Theorem \ref{thm:sublinear_time_recovery}. 

\subsection{Recovering Low-Rank Toeplitz Matrices using Fourier Structure.}

Our starting point is the main result of \cite{kapralov2022toeplitz}, which shows that any PSD Toeplitz matrix has a near optimal low-rank approximation $\wt T$ which is itself Toeplitz. Armed with this existence result,  the key idea behind the algorithm of \cite{kapralov2022toeplitz} is to leverage the well-known sparse Fourier structure of low-rank Toeplitz matrices \cite{cybenko1982moment} to recover $\wt T$ using a sample efficient (but computationally inefficient) algorithm. %In particular, they leverage the well-known Vandermonde decomposition theorem, which shows that any Toeplitz matrix can be diagonalized by an {off-grid} Fourier matrix. Formally,
The  idea behind \cite{eldar2020toeplitz} is similar, except that they rely on a weaker existence statement, where the low-rank approximation $\wt T$ is guaranteed to have sparse Fourier structure, but not necessarily be Toeplitz.
Formally, after defining the notion of a Fourier matrix, we state the main result of \cite{kapralov2022toeplitz}.
\begin{definition}[Fourier matrix]\label{def:symmetric_fourier_matrix}
	For any set of frequencies $S=\{f_1,f_2,\ldots, f_s\}\subset [0,1]$, let the Fourier matrix $F_S\in \mathbb{C}^{d\times s}$ have $j^{th}$ column equal to $v(f_j) := [1, e^{2\pi i f}, e^{2\pi i (2f)}\ldots, e^{2\pi i  (d-1) f}]$.
	%For every $j\in \{1, 2,\ldots, s\}$ the $j$-th column of $F_{S}$ is the frequency vector $v(f_{j})$. 
\end{definition}

\begin{theorem}[Theorem 2 of \cite{kapralov2022toeplitz}]\label{frobeniusexistence}
        For any PSD Toeplitz matrix $T\in \mathbb{R}^{d\times d}$, $0<\eps,\delta<1$, and $k\leq d$, there exists symmetric Toeplitz $\wt{T}$ with rank $r = \tilde {O}(k \log(1/\delta)/\eps))$ such that the following holds,
		$$\|T-\wt{T}\|_F \leq (1+\eps)\|T-T_k\|_F + \delta\|T\|_F.$$ %where $T_k = \underset{B:rank(B)\leq k}{\argmin}\|T-B\|_F$ is the best rank-$k$ approximation to $T$ in the Frobenius norm.
		
		Further, $\wt T$ can be written as $\wt T = F_S D F_S^*$ where $F_S \in \mathbb{C}^{d \times r}$ is a Fourier matrix (Def. \ref{def:symmetric_fourier_matrix}) and $D \in \R^{r \times r}$ is diagonal.\footnote{Our algorithms will rely on several other properties of the frequency set $S$ and diagonal matrix $D$ guaranteed to exist by Theorem \ref{frobeniusexistence}. See Section \ref{sec:preliminaries} for a more complete statement of the  theorem, which details these properties.}
\end{theorem}

The first step in our proof of Theorem \ref{thm:sublinear_time_recovery} is to  apply Theorem \ref{frobeniusexistence} to the input PSD Toeplitz matrix $T$ with $\epsilon = \Theta(1)$. Our algorithm is given access to the entries of the noisy matrix $T+E$, which we will write as $\wt T + \wt E$ for $\wt E = E + T-\wt T$. Observe that by triangle inequality, since $\norm{T-\tilde T}_F \lsim \norm{T-T_k}_F + \delta \norm{T}_F$,
$$\norm{\wt E}_F \lsim \norm{E}_F + \norm{T-T_k}_F + \delta \norm{T}_F \lsim \max\{  \norm{E}_F, \norm{T-T_k}_F\} + \delta \norm{T}_F.$$

Thus, to prove Theorem \ref{thm:sublinear_time_recovery}, it suffices to show that we can recover $\wt T$ to within a constant factor of the noise level $\norm{\wt E}_F$. Our key contribution is to improve the runtime of this step from being exponential in the rank $r = \tilde O(k \log(1/\delta))$ to being polynomial, and hence sublinear in $d$ (recall that throughout $\tilde O(\cdot)$ hides poly logarithmic factors in $d$ and the argument).

To approximately recover $\wt T$, as in \cite{kapralov2022toeplitz} and \cite{eldar2020toeplitz}, we will leverage its Fourier structure. In particular, 
Theorem \ref{frobeniusexistence} guarantees that $\wt T$ can be written as $\wt T = F_S D F_S^*$, where $F_{S} \in \mathbb{C}^{d \times r}$ is a Fourier matrix (Def. \ref{def:symmetric_fourier_matrix}) and $D \in \R^{r \times r}$ is diagonal. To find $\wt T$, the algorithm of \cite{kapralov2022toeplitz} brute force searches for a good set of frequencies $S$ and corresponding diagonal matrix $D$. In particular, they consider a large pool of candidate frequency sets drawn from a net over $[0,1]^r$. For each set $S$, they solve a regression problem to find a diagonal matrix $D$ that (approximately) minimizes $\norm{T-F_S D F_S^*}_F$. These regression problems can be solved using a sublinear number of queries to $T$ using leverage score based random sampling \cite{sarlos2006improved,Woodruff:2014tg}. One can then return the best approximation given by any of the candidate frequency sets as the final output. 

Unfortunately, the above approach incurs exponential runtime, as the number of frequency sets considered grows exponentially in the rank $r = \tilde O(k \log(1/\delta))$. A similar issue arises in \cite{eldar2020toeplitz}, where a brute force search over frequency sets is also performed, but with a different regression step, where $D$ is relaxed to be any $r \times r$ matrix.

%\Cam{Don't have medskips or other additiona spacing before sections. (there is a medskip below this comment. check elsewhere.}
%\medskip 

\subsection{From Column Recovery to Matrix Recovery.}\label{subsec:intro_matrix_recovery}
To avoid the exponential runtime, of \cite{kapralov2022toeplitz} and \cite{eldar2020toeplitz}, we  apply tools from the extensive literature on efficient recovery of Fourier-sparse functions \cite{hassanieh2012nearly,indyk2014sample,indyk2014nearly,Boufounos:2015vj,kapralov2016sparse,chen2016fourier,Jin:2023uc}. Observe that if we expand out the decomposition $\wt T = F_S D F_S^*$ from Theorem \ref{frobeniusexistence},  each column of $\wt T$ is an $r$-Fourier sparse function from $\{0,\ldots,d-1\} \rightarrow \R$. In particular, letting $\wt T_j$ denote the $j^{th}$ column, for any $t \in \{0,\ldots,d-1\} $, letting $\{a_f\}_f \in S$ be the diagonal entries of $D$,
$$\wt T_j(t) = \sum_{f\in S}a_f e^{2\pi i f(t-j)}.$$ 

Our algorithm will recover a Fourier-sparse approximation to a single column of $\wt T$ using samples of $T + E = \wt T + \wt E$ in just $\poly(r,\log d)$ time using a sparse Fourier transform algorithm. This Fourier-sparse approximation will give us approximations to the frequencies in $F_S$, which in turn can be used to form an approximation to $\wt T$. This approach presents several challenges, which we discuss below.

%\Cam{Perhaps you can break this section into several paragraphs with bold heading to guide the reader a bit more. E.g. 'Column Signal-to-Noise-Ratio.' 'Heavy-Light Frequency Decomposition', 'Sublinear Time Regression onto the Recovered Frequencies'.}

\smallskip

\noindent \textbf{Column Signal-to-Noise ratio.}  First, we must ensure that the column that we apply sparse Fourier transform to does not have too much noise on it, compared to its norm. To do so, we sample a column $\wt T_j$ uniformly at random. By Markov's inequality, $\wt T_j$ is corrupted with noise whose $\ell_2$ norm is bounded as $\norm{\wt E_j}_2 \lsim \norm{\wt E}_F/\sqrt{d}$ with good probability. Further, using that $\wt T$ is symmetric Toeplitz and that it is close to $T$ and hence nearly PSD, we can apply a norm compression inequality of Audenaert \cite{audenaert2006norm} (see Claim \ref{lem:diagonal_dominance_psd}) to show that each column of $\wt T$ must have relatively large norm.  In particular, for any $j$ we show that  $\norm{\wt T_j}_2 \gsim \norm{\wt T}_F/\sqrt{d}$ (see Lemma \ref{lem:smallnoisefirstcol}). In combination, these facts ensure that a random column $\wt T_j$ has a similar signal-to-noise ratio as the full matrix $\wt T$ with good probability. Thus, we can expect to recover a good approximation to $\wt T$ from just this column.

\smallskip

\noindent \textbf{Heavy Frequency Recovery.}
Of course, given the noise $\wt E$, we cannot expect to recover approximations to all frequencies in $S$ given sample access to $\wt T + \wt E$. For example, if a frequency $f$ corresponds to a very small entry $a_f$ in $D$, we may not recover it. We must argue that omitting such  frequencies from our approximation of $\wt T$ does not introduce significant error. 
More formally, our sparse Fourier transform will recover a set of frequencies that well approximates some subset of the input frequencies $S_{heavy} \subset S$, but may not approximate the set frequencies $S_{light} := S \setminus S_{heavy}$. The algorithm guarantees that $S_{heavy}$ suffices to approximate our random column $\wt T_j$  up to the noise level $\norm{\wt E_j}_2$. Let $Z_{heavy}: \{0,\ldots,d-1\}  \rightarrow \R$ be given by restricting $\wt T_j$ to the frequencies in $S_{heavy}$, i.e., $Z_{heavy}(t) = \sum_{f \in S_{heavy}} a_f e^{2\pi i f t}$. Let $Z_{light} = T_j - Z_{heavy}$. Then we can show the following, $$\norm{Z_{light}}_2 = \|\wt T_j-Z\|_2 \lsim \|\wt E_j\|_2 +\delta \|\wt T_j\|_2.$$ 
We have ensured that $\|\wt E_j\|_2 \lsim \|\wt E\|_F/\sqrt{d}$ by choosing a random column. Further, by applying our recovery algorithm with error $\delta' = \delta/\sqrt{d}$, which introduces only additional $\log d$ dependences, we  have that  $\delta' \|\wt T_j\|_2 \le \delta \norm{\wt T}_F/\sqrt{d}$. Thus, we have overall that $\norm{Z_{light}}_2 \lsim \frac{\norm{\wt E}_F + \delta \norm{\wt T}_F}{\sqrt{d}}$. 

It remains to argue that this error being small on just column $\wt T_j$ ensures that it is small on the full matrix. In particular, if we let $\wt T_{light}$ denote our full Toeplitz matrix restricted to the unrecovered frequencies, we must show the following, 
$$\norm{\wt T_{light}}_F \lsim \sqrt{d} \cdot \norm{Z_{light}}_2 \lsim \norm{\wt E}_F + \delta \norm{\wt T}_F.$$
Again, we show this by arguing that $\wt T_{light}$ is itself near PSD and applying the norm compression inequality of \cite{audenaert2006norm}. We interpret the existence of $\wt T_{light}$ and $\wt T_{heavy} := \wt T - \wt T_{light}$ as a heavy-light decomposition of $\wt T$ into $\wt T = \wt T_{heavy}+\wt T_{light}$, see Section \ref{subsection:heavy_light_decomp} for the detailed proof. %This gives us a heavy-light decomposition $\wt T=\wt T^{heavy}+\wt T^{light}$, and it is formally presented and proved in Section \ref{subsection:heavy_light_decomp}. \Cam{I don't get this sentence. How does arguing that $\wt T_{light}$ is near PSD give us a heavy light decomposition? 

\smallskip

\noindent \textbf{Sublinear time Approximate Regression. }With the above bound in hand, the proof of Theorem \ref{thm:sublinear_time_recovery} is essentially complete, in the last step we use a leverage score based fast approximate regression primitive \cite{sarlos2006improved,Woodruff:2014tg} to regress onto the approximate frequencies that we recover to give an approximation to the full $\wt T$. We have argued that the noise incurred by not recovering all frequencies is bounded by  $\norm{\wt E}_F + \delta \norm{\wt T}_F$, and this bounds the optimum of the regression problem, which we solve approximately, giving the final guarantee of Theorem \ref{thm:sublinear_time_recovery}. To solve the regression problem efficiently, we crucially rely on the special structure of the Fourier spectrum $S$ of $\wt T$, stated in the full version of Theorem \ref{frobeniusexistence} in Section \ref{finalproof}.  Full details of this step  are presented in Section \ref{subsection:noisy_toeplitz_recovery}. %Finally in Section \ref{subsection:toeplitz_applications}, we discuss the formal details of instantiating the robust low-rank approximation result of Theorem \ref{thm:sublinear_time_recovery} to obtain our sublinear time Toeplitz low-rank approximation and covariance estimation results of Theorems \ref{thm:sublinear_time_lowrankapprox} and \ref{thm:sublinear_time_covariance_estimation}.\Cam{I am commenting this out since it is out of place. This section is not talking about Theorem 2 and 3. Pleaase odn't add this back in.}
%We detail our approaches to the above challenges further in Section \ref{sec:sublinear_toeplitz_recovery}.
%\Cam{If we are planning on adding Psuedocode for the overall algorithmic approach, this migth be a good place to do it.}

\subsection{Sublinear Time Discrete-Time Off-Grid Sparse Fourier Transform.}\label{subsec:intro_sparsefft}
We have argued  that applying a sparse Fourier transform to a random column of $\wt T$ suffices to recover a good approximation to the full matrix. It remains to show that we can implement such a sparse Fourier transform efficiently. We sketch the key ideas behind this efficient sparse Fourier transform here, with the complete proof appearing in Section \ref{sec:dtft_recovery}. 

%We can thus hope to recover the first column $\wt T_1$, and in turn $\wt T$ using noise-robust sparse Fourier recovery techniques. % on $\wt T_1(t)$, and then fit $\wt D$ using fast regression primitives. This approach was adopted by \cite{kapralov2022toeplitz} however their sparse Fourier recovery algorithm recovered the frequencies using a brute force search running in $\exp(k)$ time, which resulted in an inefficient running time overall.

If the frequencies in the set  $S$ guaranteed to exist by Theorem \ref{frobeniusexistence} were integer multiples of $1/d$, i.e., they were ``on-grid'', %then the function $\wt T_j(t)$ would be a periodic function with period $d$, and thus completely determined by its values at the integer points $t\in [d]$. 
then we could apply a discrete sparse Fourier transform algorithm to  recover a Fourier-sparse approximation to $\wt T_j$ given samples from the column $[T+E]_j = [\wt T + \wt E]_j$. Such algorithms  have been studied extensively \cite{hassanieh2012nearly,Hassanieh:2012tq,Gilbert:2014tp}. The main difficulty in our case is that the frequencies in $S$ may be ``off-grid'', i.e., arbitrary real numbers in $[0,1]$. % in general and thus $\wt T_1(t)$ may have a period $>>d$ making the problem hard. 
Recent work of \cite{chen2016fourier} solves the off-grid sparse Fourier transform problem given sample access to the function $[\wt T +\wt E]_j$ on the \emph{continuous range $[0,d]$}. However, in our setting, we can only access $[\wt T + \wt E]_j$ at the integer points $\{0,1,\ldots,d-1\}$ -- i.e., at the entries in the $j^{th}$ column of our matrix. Thus, we must  modify the approach of  \cite{chen2016fourier} to show that, nevertheless, we can well approximate $\wt T_j$ from samples  at these points. In doing so, we give to the best of our knowledge, the first efficient \emph{discrete time off-grid} sparse Fourier transform algorithm, which may be applicable in other settings that involve off-grid frequencies but on grid time samples.

 Informally,  we consider the following sparse Fourier recovery problem:
 \begin{problem}[Off-Grid Sparse DTFT -- Informal]\label{def:problem_definition}
Let $x^*(t) = \sum_{f \in S} a_f \cdot e^{2\pi i f t}$ for $t\in [d]:=\{0,1,\ldots,d-1\}$ and $S \subset [0,1]$ with $|S| = k$. Let $x(t) = x^*(t)+g(t)$, where $g$ is arbitrary noise. Given access to $x(t)$ for $t\in [d]$, approximately recover all $f \in S$ that contribute significantly to $x$. % in $\poly(k, \log d)$ time. %Remark: Note that since we only have time domain access at points $t\in [d]$, this is a harder problem than the one studied in \cite{chen2016fourier}\kstodo{add footnote formalizing the connection}.
\end{problem}
%For any function $x:\mathbb{Z}\to \mathbb{C}$ we let $\wh{x^*}:[0,1]\to \mathbb{C}$ denote the discrete time Fourier transform (DTFT) (see  Def. \ref{def:dtft} for a formal definition). 
%The work of \cite{chen2016fourier} solves the above problem in the more general case when $S \subset [-F,F]$ for some bandlimit $F$. However, their algorithm critically requires sample access to $x(t)$ for any real  $t\in [0,d]$. 
%\Cam{Throughout this and the previous section we are using $[d]$ to denote both the sets $\{0,1,\ldots,d-1\}$ and $\{1,\ldots,d\}$. We really should clean this up here and in the rest of the paper. Maybe we can definite a different notation for $\{0,1,\ldots,d-1\}$. Maybe $[d]_0$ or maybe we should just write out. $\{0,\ldots,d-1\}$. }  
In principal, without noise, one can recover $x^*$ from $2k$ samples at any time domain points -- including at the integers in $[d]$ -- via 
Prony's method \cite{Prony:1795wv}. The key challenge is to recover $x^*$ approximately in the presence of noise. Observe that with noise, it is generally not  possible to identify all frequencies in $S$ given sample access to $x(t)$ only on $[d]$. For example, if two frequencies $f$ and $f'$ are extremely close to each other, their contributions to the function $x$ on $[d]$ could nearly cancel out, making it impossible to identify them. This is true even if the coefficients $a_{f}$ and $a_{f'}$ are arbitrarily large. 

Thus, we will settle for approximate frequency recovery. In particular, we will output a list of frequencies $L$ such that it well approximates a subset of frequencies $S_{heavy}\subseteq S$. Denoting for any function $f:[d]\rightarrow \mathbb{C}$ let $\norm{f}_d^2 = \sum_{i\in [d]}|f(i)|^2$, this subset $S_{heavy}$ spans an approximation to the Fourier sparse function $x^*$ up to error $\lsim \norm{g}_d^2 + \delta \norm{x}_d^2$. To find this approximation algorithmically, we will have to regress onto the list of frequencies $L$, and the final approximation itself will use $\poly(k,\log(d/\delta))$ frequencies. %\Cam{no delta dependence in the sparsity bound? This seems weird since our additive error grows as $\delta$ gets smaller but then this suffices for a better error bound? Seems not right.} \Cam{looks like you updated here but not in the statement of Lemma 2.1 or in the later formal statements? E.g. where-ever we apply this statement it will have to be updated?}. 
In our Toeplitz setting (Theorem \ref{thm:sublinear_time_recovery}), this translates to $\poly(k,\log(d/\delta))$ rank of the final Toeplitz matrix that we output. We state our main approximate frequency recovery primitive below. See Section \ref{sec:dtft_recovery} for a more formal statement.
%\Cam{If $S(L)$ is just equal to $S_{heavy}$ we should switch the notation here to $S_{heavy}$ to be consistent with the previous section.}
%\Cam{The norm $\norm{g}_d$ is undefined. I think $x$ and $g$ are already just vectors in $\mathbb{R}^d$ so we can just use the $\ell_2$ norm right? (KS: yes but at some places we need norm of $d$ dim vectors and $d/2$ dim vectors). Also the sparsity looks incorrect. Also it seems that the statement of the lemma doesn't erally match our description of the result which comes before. Since the lemma doesn't acutlaly show that there is a good approximation to $x^*$ which is spanned by the $poly(k)$ frequencies in $L$. It just says that there is a good approximation in the span of $S_{heavy}$, which is a set of $\le k$ frequencies that we well approximatie. (KS: yes fixed)}
\begin{lemma}[Approximate Frequency Recovery -- Informal] \label{lem:dtft_freqrecovery} Consider the setting of Problem \ref{def:problem_definition}. Assume that $\|g\|_d^2 \leq c\|x^*\|_d^2$ for a small absolute constant $c$. Then for any given $\delta>0$, there exists an algorithm that in time and sample complexity $\poly(k,\log(d/\delta))$ outputs a list $L$ of $\poly(k,\log(d/\delta))$  frequencies in $[0,1]$ such that, with probability at least $0.99$, letting $$S_{heavy}=\left \{f\in S: \exists f'\in L \text{ s.t. } |f-f'|_{\circ}\leq \frac{\poly(k,\log(d/\delta)}{d}\right \},$$ and letting $x^*_{S_{heavy}}(t) := \sum_{f\in S_{heavy}}a_f e^{2\pi i f t}$, we have the following, $$\|x^*-x^*_{S_{heavy}}\|_d^2 \lsim \|g\|_d^2 +\delta \|x^*\|_d^2.$$ Here $|\cdot|_{\circ}$ is the wrap-around distance. For any $f_1,f_2 \in [0,1]$, $|f_1-f_2|_{\circ} := \min\{|f_1-f_2|,|1-(f_1-f_2)|\}$. %\Cam{You have to define it here. In paper is is never ok to forward ref to a definition or theorem.}
\end{lemma}
Lemma \ref{lem:dtft_freqrecovery} mirrors Lemma 7.23 in \cite{chen2016fourier}, but only requires accessing $x(t)$ at integer $t \in [d]$ rather than at real $t\in [0,d]$. 
We now discuss how we adapt the approach of \cite{chen2016fourier} to work in this more restricted sampling setting. 

\smallskip

\noindent \textbf{One cluster recovery.} The approach of \cite{chen2016fourier} first considers the \emph{one cluster} case (see Section \ref{subsection:one_cluster} for a formal definition), where most of the energy of $x$ is concentrated around a single frequency $f_0$. They show how to approximately recover this central frequency of the cluster up to $\poly(k,\log(d/\delta))/d$ error in $\poly(k,\log(d/\delta))$ time. Roughly, since the signal is close to a pure frequency, by considering the ratio $ x(\alpha+\beta)/x(\alpha)$ for carefully chosen sample points $\alpha,\beta$, they show that one can approximate $e^{2\pi i f_0 \beta}$, and in turn $f_0$. In their work, $\alpha$ and $\beta$ may be real valued. We need to modify the approach to restrict them to be integers. We formally do this in Section \ref{subsection:one_cluster}. The key idea follows similar ideas to our other modifications of their approach, discussed below. 

\smallskip

\noindent\textbf{Multi-cluster recovery with bounded support.}
After handling the single cluster case, \cite{chen2016fourier} reduces the general case to it via hashing techniques. In particular, by applying an efficient transformation in time domain, they access a signal whose frequencies are hashed versions of the frequencies in $S$. This approach is standard in the literature on sparse Fourier transform. The hash function spreads out the frequencies in $S$ so that they lie in different frequency ranges (called `buckets') with good probability, and can be recovered using the single cluster recovery primitive. The hash function used in \cite{chen2016fourier} is as follows: let  $\sigma\in \mathbb{R},b\in [0,1]$ and $B$ be the number of buckets. Then define:
\begin{align}
    \pi_{\sigma,b}(f) &= B \sigma (f-b) \mod B \nonumber,\\
    h_{\sigma,b}(f) &= \text{round}(\pi_{\sigma,b}(f))\label{eqn:hash_fn},
\end{align}
where the round$(.)$ function rounds real numbers to the nearest integer. The crucial claim that \cite{chen2016fourier} shows is that if we let $\sigma$ be a uniformly random real number in $\left [\frac{d}{\poly(k,\log(d/\delta))},\frac{2d}{\poly(k,\log(d/\delta))}\right ]$ then we have that for any $f_1,f_2$ such that $|f_1-f_2|\geq \frac{\poly(k,\log(d/\delta))}{d}$,
\begin{equation}\label{eqn:low_collision_prob}
    \Pr [h_{\sigma,b}(f_1)=h_{\sigma,b}(f_2)]=\Pr [B\sigma |f_1-f_2|\in (-1,1) \mod B]\lsim 1/B.
\end{equation}
This ensures that $f_1,f_2$ land in different hash buckets with good probability. I.e., after hashing, these frequencies can be separated in the Fourier domain. In particular,  \cite{chen2016fourier}  applies an efficient filtering approach (see Lemma \ref{lem:hashtobins}) to isolate a cluster of frequencies of width $\poly(k,\log(d/\delta))/d$ around each frequency. If the hash bucket containing this frequency has high SNR, then the frequency can be recovered via single-cluster recovery. %\cite{chen2016fourier} design a filtering procedure that allows time domain access to the function containing frequencies restricted to any bucket, which is then used by the one-cluster recovery procedure. See Lemma \ref{lem:hashtobins} for a formal treatment of this procedure.\Cam{Not sure why this is relevant to mention in such detial. Since we use the same approach. Just commenting out.}

To implement the above approach in our setting, where we can only access the input at integer time points, we need the random seed $\sigma$ to be a random \emph{integer}. In particular, we let $\sigma$ to be a uniformly random {integer} in $\left [\frac{d}{\poly(k,\log(d/\delta))},\frac{2 d}{\poly(k,\log(d/\delta))}\right ]$. However this restriction severely affects the collision behavior of the hash function $h_{\sigma,b}$.  For example, if $|f_1-f_2|_{\circ}=1/2$, then $B\sigma|f_1-f_2|_{\circ}\mod B=0$ for \emph{every even $\sigma$} and thus  \eqref{eqn:low_collision_prob} cannot hold with any probability less than $1/2$. %We can instead try to choose $\sigma$ to be a random odd number in $[d]$ as is the case in seminal works on discrete sparse Fourier transforms \cite{hassanieh2012nearly}, however this will lead to values of $\sigma \gg d/\poly(k)$ most of the times and this prevents us from ensuring that the entire cluster of width $\poly(k)/d$ around any frequency lands in the same hash bucket as that frequency.

To handle this issue, we  observe that if an instance only contained frequencies within an interval of width $1/B$ then we would have $|f_1-f_2|_{\circ}\leq 1/B$ for any $f_1,f_2$ in the function's support. Then, if $\sigma$ is a uniformly random {integer} in $\left [\frac{d}{\poly(k,\log(d/\delta))},\frac{2 d}{\poly(k,\log(d/\delta))}\right ]$, $B\sigma |f_1-f_2|_{\circ}$ is uniformly distributed on a grid of spacing at most $B|f_1-f_2|_{\circ}\leq 1$. Thus for every $\sigma$ such that $B\sigma |f_1-f_2|_{\circ} \in (-1,1) \mod B$, there are roughly at least $B$ other  values of $\sigma$ for which $B\sigma |f_1-f_2|_{\circ}  \notin (-1,1) \mod B$.  This is enough to show that \eqref{eqn:low_collision_prob} holds.

\smallskip

\noindent\textbf{General multi-cluster recovery.} Of course, a general input instance may have frequencies that do not lie in an interval of width $1/B$. To reduce to this setting, we apply an additional filtering step in the Fourier domain. Our filter splits the interval $[0,1]$ into $B$ intervals of width $1/B$, based on the same filtering approach of Lemma \ref{lem:hashtobins} as discussed before. %\Cam{IS what i wrote above correct?} 
For the complete proof, see Section \ref{subsection:multi_cluster_freq_recovery}. This completes our proof sketch for our discrete time off-grid sparse Fourier transform primitive (Lemma \ref{lem:dtft_freqrecovery}).
\subsection{Other related work.}\label{sec:prior} 
%\kstodo{Mostly taken from previous soda paper!}
A large body of work in numerical linear algebra, applied mathematics, theoretical computer science, and signal processing  has studied the problem of computing low-rank approximations of Toeplitz and Hankel matrices \cite{Luk:1996ur,Park:1999ta,Ishteva:2014wm,Cai:2016wh,Ongie:2017us,Krim:1996wm,Shi:2019ud}. In many applied signal processing settings, Toeplitz matrices arise as PSD covariance matrices, motivating our focus on the PSD case.
%Using the FFT, any $d \times d$ Toeplitz matrix can be multiplied by a vector in  $O(d \log d)$ time. \cite{Shi:2019ud} show how to use this  primitive to give a near-linear time, near-optimal low-rank approximation algorithm for Toeplitz matrices. In particular, their algorithm outputs $\wt T$ with rank $k$ satisfying $\norm{T-\wt T}_F \le (1+\epsilon) \norm{T-\tilde T}_F$ in $\tilde O(d + \poly(k/\epsilon))$ time. Note that $\wt T$ here may not be Toeplitz.
Significant prior work has also studied the problem of computing an optimal \emph{structure-preserving} Toeplitz low-rank approximation, where the low-rank approximation $\wt T$ is itself required to be Toeplitz. However, no simple characterization of the optimal solution to this problem is known \cite{chu2003structured} and polynomial time algorithms  are only known in the special cases of $k= 1$ and $k = d-1$ \cite{chu2003structured,Knirsch:2021ve}. On the practical side a range of heuristics are known, based on techniques such as convex relaxation \cite{Fazel:2013vw,Cai:2016wh,Ongie:2017us}, alternating minimization \cite{chu2003structured,Wen:2020ub}, and sparse Fourier transform \cite{Krim:1996wm}. 

%\cite{Park:1999ta,Ishteva:2014wm}. %Many works focus on computationally efficient, but not sublinear algorithms that output a good Toeplitz low-rank approximation to the input. 
%\cite{Wen:2020ub} matrix comlpetion alternatiing svd typo algorithm.

Many results in the signal processing literature study sublinear query algorithms for Toeplitz matrices, these are referred to as \emph{sparse array methods} which proceed by querying a small principal submatrix to obtain an approximation to the whole matrix \cite{Abramovich:1999vs,Chen:2015wz,Qiao:2017tp,Lawrence:2020ut}. The work of \cite{eldar2020toeplitz} and \cite{kapralov2022toeplitz} on sublinear query PSD Toeplitz low-rank approximation and Toeplitz covariance estimation is closely related to this literature. %give sublinear query algorithms for various problems on Toeplitz matrices as discussed in the introduction.
%
%
%Several of these works focus on algorithms with sublinear query complexity \cite{Lawrence:2020ut} \cite{Qiao:2017tp}
%
%%\cite{Ongie:2017us} convex relaxation
%
%%\cite{Park:1999ta} -- not really sure the method
%
%%\cite{Knirsch:2021ve} -- optimal rank-1 paper.
%
%
%%Structured low rank approximation \cite{chu2003structured}.
%%\begin{itemize}
%%%\item \url{https://www.cas.mcmaster.ca/~qiao/publications/LQV02.pdf}
%%%\item \url{http://www.cas.mcmaster.ca/~qiao/publications/LQ96b.pdf}
%%%\item \url{https://journalofinequalitiesandapplications.springeropen.com/track/pdf/10.1186/s13660-020-02340-w.pdf}
%%%\item \url{https://ieeexplore.ieee.org/document/8715594}
%%%\item \url{https://www.ncbi.nlm.nih.gov/pmc/articles/PMC5999344/}
%%%\item \url{https://conservancy.umn.edu/handle/11299/215327}
%%%\item \url{https://www.math.colostate.edu/~king/codex/slides/Plonka_2020_11_03.pdf} lots of good refs. Optimal rank-1 solution. Also optimal solution in infinite operator case.
%%\item \url{https://mtchu.math.ncsu.edu/Research/Lectures/Iep/chapter8.pdf} states that optimal Toeplitz low-rank approximation is open.
%%\end{itemize}
%

Beyond Toeplitz matrices, significant  work has focused on  sublinear time low-rank approximation algorithms for other structured matrix classes. This includes positive semidefinite matrices \cite{musco2017recursive,Bakshi:2020tl}, distance matrices \cite{Bakshi:2018ul,Indyk:2019vy}, and kernel matrices \cite{musco2017recursive,Yasuda:2019vf,Ahle:2020vj}.

%Our techniques also rely on sublinear time algorithms for sparse Fourier recovery problems. There is a large body of work on sparse discrete fourier transforms, that are concerned with computing $k$-Fourier sparse approximations to discrete time signals of period $d$ in $\wt O(k)$ time \cite{hassanieh2012nearly,indyk2014sample,indyk2014nearly,kapralov2016sparse,chen2016fourier}. Most relevant to our work is the work of \cite{chen2016fourier}, which obtained $\poly(k,\log d)$-Fourier sparse approximations on a bounded interval $[0,d]$ to $k$-Fourier sparse functions with arbitrary frequencies, thus not necessarily periodic, in $\poly(k,\log d)$ time given access to the input on the entire interval $[0,d]$.

%% file: preliminaries.tex
% !TEX root = ./main.tex

\section{Notation and preliminaries} 

In this section, we introduce notation and preliminary concepts that are used throughout this paper.

\label{sec:preliminaries}
%\Cam{missing periods after the subsections. Please propograte this fix throughout the paper, not just sec 1/2.}
\subsection{General and linear algebraic notation.}
Consider functions $f:\mathcal{X} \rightarrow \R$ and $g:\mathcal{X} \rightarrow \R$ for input domain $\mathcal{X}$. We write $f(\cdot )\lsim g(\cdot)$ if there exists a constant $C>0$ such that $f(x) \leq C g(x)$ for all $x\in \mathcal{X}$. For any integer $d > 0$, let $[d] = \{0,1,\ldots,d-1\}$. 
For any function $x:\Z \rightarrow \mathbb{C}$, and integer $d>0$, we let $\|x\|_{d}^2 := \sum_{j\in [d]} |x(j)|^2$. % \Cam{This is out of order. You need to define $[d]$ before you use it in the definition of the norm...} 
For any set $N$, let $N^n$ denote the set of all subsets of $N$ with $n$ elements.

For a matrix $A$, let $A^{T}$ and $A^*$ denote its transpose and Hermitian transpose, respectively. Let $A_{[i,j]}$ denote the $i,j$ entry of $A$ and for $i_1 < j_1, i_2 < j_2$, let $A_{[i_1:j_1,i_2:j_2]}$ denote the submatrix containing entries from rows $i_1$ to $j_1$ and columns $i_2$ to $j_2$. For any vector $x\in \mathbb{C}^d$, let $\|x\|_2 = \sqrt{x^*x}$ denote its $\ell_2$ norm. For a matrix $A\in \mathbb{C}^{d\times d}$, let $\|A\|_2 = \sup_{x\in \mathbb{C}^d} \|Ax\|_2/\|x\|_2$ denote its spectral norm and $\|A\|_F=\sqrt{\sum_{i,j \in [d]} |A_{[i, j]}|^2}$ denote its Frobenius norm.%\Cam{You are using a different notation here for indexing a matrix's entries than the one you defined immediately above. Since you are not using brackets here. Also do you mean to restrict $A$ to be square?} 

A Hermitian matrix $A\in \mathbb{C}^{d\times d}$ is positive semidefinite (PSD) if for all $x\in \mathbb{C}^d$, $x^*Ax \geq 0$. Let $\lambda_{1}(A)\geq \ldots \geq \lambda_{d}(A)\geq 0$ denote its eigenvalues. Let $\preceq$ denote the Loewner ordering, that is $A\preceq B$ if and only if $B-A$ is PSD. Let $A= U\Sigma V^*$ denote the compact singular value decomposition of $A$, and when $A$ is PSD note that $U \Sigma U^*$ is its eigenvalue decomposition. In this case, let $A^{1/2} = U\Sigma^{1/2}$ denote its matrix square root, where $\Sigma^{1/2}$ is obtained by taking the elementwise square root of $\Sigma$. Let $A_k = U_k \Sigma_k V_k^{*}$ denote the projection of $A$ onto its top $k$ singular vectors. Here, $\Sigma_k\in \mathbb{R}^{k\times k}$ is the diagonal matrix containing the $k$ largest singular values of $A$, and $U_k, V_k \in \mathbb{C}^{d\times k}$ denote the corresponding $k$ left and right singular vectors of $A$. Note that $A_k$ is the optimal rank $k$ approximation to $A$ in the spectral and Frobenius norms. That is, $A_{k} = \argmin_{B: \rank(B) \le k} \|A-B\|_2$ and $A_{k} =  \argmin_{B: \rank(B) \le k} \|A-B\|_F$. 

	%\Cam{Everyhting below should not be in this definition. It should be in a sepaparate fact.}
We let $\ast$ denote the convolution operator -- in both the discrete and continuous settings. 
For discrete functions $x,y:\mathbb{Z}\to \mathbb{C}$, we have $[x\ast y] (n)= \sum_{k\in \mathbb{Z}}x(k)y(n-k)$. For continuous functions ${x},{y}:[0,1]\to \mathbb{C}$, we have $[{x}\ast {y}](f) = \int_{0}^{1}{x}(f'){y}(f-f')df'$. %, thus $\ast$ is interpreted as the continuous convolution here. %\Cam{I don't get what you are saying here. Are you just trying to define the convolution operator? In that case it should be defined independently of the DTFT. And you should say that we overload notation a bit, using $\ast$ to denote both the discrete and continuous convolutions. This should go up in the earlier prelimns not buried in this definition.} 
For any function $f:\mathcal{D}\to \mathbb{C}$ defined over domain $\mathcal{D}$ (e.g., $\mathcal{D}$ can be $\mathbb{Z},\mathbb{R}$ etc.), we let $supp(f)$ denote its support. That is, $supp(f)=\{x\in \mathcal{D}:|f(x)|> 0\}$. %\Cam{This last sentence should not be here. You seem to be defining what the support of a function is? This should go in the notation earlier, or in a different definition.}
\subsection{Fourier analytic notation.}

Throughout, we  use the standard discrete-time Fourier transform.
\begin{definition}[Discrete-time Fourier transform (DTFT)]\label{def:dtft}
		For $x: \Z \to \mathbb{C}$, its DTFT $\wh{x}(f): [0, 1] \to \mathbb{C}$ is defined as follows,
		\[
		\wh{x}(f) =  \sum_{n \in \Z} x(n) e^{-2 \pi i f n}.
		\]
		The inverse DTFT is defined as follows,
		\[
		x(n)  = \int_{0}^{1} \wh{x}(f) e^{2 \pi i f n}  df.
		\]
  \end{definition}

The discrete time version of Parseval's theorem allows us to relate the energy in time and Fourier domains. %\Cam{Do you have a citation for this?}

	\begin{lemma}[Parseval's identity]\label{lem:parseval_dtft}
		For $x:\Z \rightarrow \mathbb{C}$ with DTFT $\wh x:[0,1]\rightarrow \mathbb{C}$, we have:
		\begin{equation*}
			\sum_{t=-\infty}^{\infty} |x(t)|^2 = \int_{0}^{1} |\wh x(f)|^2 df .
		\end{equation*}
	\end{lemma}

We  define the wrap around distance between two frequencies as follows.
%\Cam{I think we need to use $|\cdot |$ to denote absolute value since its standard. And then use a different notation for wrap around distance like $|\cdot|_{\circ}$ or something.}
\begin{definition}[Wrap around distance]\label{wraparounddefn}
     For any $f_1,f_2\in [0,1]$, let $|f_1-f_2|_{\circ}=\min\{|f_1-f_2|,|1-(f_1-f_2)|\}$. %\Cam{Intuitively what is this second distance? And why is $2 \pi$ coming into it? It seems that we are defining two completely different notions of distance in the same definition.}

\end{definition}

\subsection{Compact filter functions.}
Our algorithm will use a function $H:\mathbb{Z}\to \mathbb{R}$ which approximates the indicator function of an interval $[d]$ by preserving energies of $k$-Fourier sparse functions on this interval and almost killing off their energy outside this interval. At the same time, the support of this function in the Fourier domain is  compact. Formally we have the following,
%\Cam{You are implying that we are just restating Lemma 6.6 of \cite{chen2016fourier}. Maybe title this `Variant of Lemma 6.6 of \cite{chen2016fourier}' or sometihng like that. As we do have to reprove it}
\begin{lemma}[Discrete version of Lemma 6.6 of \cite{chen2016fourier}] \label{lem:H_discrete_properties}
	Given positive integers $d,k$ and real $\delta>0$, let $s_0=\poly(k,\log(d/\delta)),s_1 = \poly(k,\log(d/\delta)), s_3 = 1-1/\poly(k,\log(d/\delta))$, and  $l = \Theta(k\log(k/\delta))$, there is a function $H:\mathbb{Z}\rightarrow \mathbb{R}$ with DTFT $\wh{H}:[0,1]\rightarrow \mathbb{R}$ having the following properties,
	%\Cam{You need to be using left and right around brackets/parens to make them appropriatelty sized. I changed in the first equation here. Please propogate throughout the paper.}
	%\Cam{You should say in a few sentences intuitively what these properties mean. Maybe right before or after the lemma}
	\begin{align*}
		\textnormal{Property I:}& \quad H(t) \in [1 - \delta, 1] \quad \forall t \in \mathbb{Z}: |t-d/2| \leq d \left (\frac{1}{2} - \frac{2}{s_1} \right ) s_3,\\
		\textnormal{Property II:}& \quad H(t) \in [0, 1] \quad \forall t \in \mathbb{Z}: d\left(\frac{1}{2} - \frac{2}{s_1}\right) s_3 \leq |t-d/2| \leq \frac{d}{2} s_3,\\
		\textnormal{Property III:}& \quad H(t) \leq s_0 \cdot \left(s_1\left(\frac{|t-d/2|}{ds_3} - \frac{1}{2}\right) + 2\right)^{-l} \quad \forall t \in \mathbb{Z}: |t-d/2| \geq \frac{d}{2} s_3,\\
		\textnormal{Property IV:}& \quad supp(\wh{H}) \subseteq \left [-\frac{s_1 l}{2d s_3}, \frac{s_1 l}{2 ds_3} \right ] \text{ and so } \Delta_h := |supp(\wh{H})| = s_1 l/(ds_3)= \poly(k,\log(d/\delta))/d.
	\end{align*}
% \Cam{In the intro, you used $supp$ rather than $\mathrm{supp}$. Make consistent. Probability you should define an operator and use $\mathrm{supp}$ throughout.}
	 %\Cam{These should be defined before they are used.}
	For any exact $k$-Fourier-sparse signal $x^*(t)$ we additionally have the following,
	\begin{align*}
		\textnormal{Property V:}& \quad \sum_{t \in \mathbb{Z}\setminus [d]} |x^*(t) \cdot H(t) |^2 \leq \delta \|x^*\|_d^2,\\
		\textnormal{Property VI:}& \quad \|x^*\cdot H\|_d^2 \in [0.99\|x^*\|_d^2, \|x^*\|_d^2] .
	\end{align*}
	%\Cam{Why use these summation notations and not the norm notation which you have defined enough. I.e. $\norm{\cdot}_d^2$?}
	%Observe that we have $\Delta_h := |supp(\wh{H})| = s_1 l/(ds_3)= \poly(k,\log(d/\delta))/d$.
\end{lemma}
The first four properties guarantee that $H$ approximates the indicator function of the interval $[d]$ and has a compact support in Fourier domain. The final two properties capture the fact that multiplying $H$ by any $k$-Fourier sparse function almost kills of its energy outside $[d]$ and almost preserves its energy inside $[d]$. %\Cam{I don't follow this. How is this different from $H$ being an approximation indicator function which we say is established by the first four properties?} 
%\Cam{This is a bit confusing. Prop V and VI presumably don't need Prop IV to be shown right? Should we just say 'The final two properties, which follow from I-III, establish that application of $H$ on $k$-Fourier sparse functions almost kills of its energy outside $[d]$ and almost preserves its energy inside $[d]$'. Is this an accurate statement? (KS: also need leverage score bounds on Fourier sparse functions to prove 5 and 6 from 1,2,3.)}
Lemma \ref{lem:H_discrete_properties} follows by modifying the proof of Lemma 6.6 of \cite{chen2016fourier}, which gives an analogous filter function in the continuous domain. We state Lemma 6.6 below, followed by the proof of Lemma \ref{lem:H_discrete_properties}.
\begin{lemma}[Lemma 6.6 of \cite{chen2016fourier}]\label{lem:H_continuous}
    Given positive integers $d,k$ and real $\delta>0$, let $s_0=\poly(k,\log(d/\delta)),s_1 = \poly(k,\log(d/\delta)), s_3 = 1-1/\poly(k,\log(d/\delta))$, and  $l = \Theta(k\log(k/\delta))$, there is a function $H':\mathbb{R}\rightarrow \mathbb{R}$ with continuous Fourier transform $\wh{H'}:\mathbb{R}\rightarrow \mathbb{R}$ having the following properties,
	%\Cam{You need to be using left and right around brackets/parens to make them appropriatelty sized. I changed in the first equation here. Please propogate throughout the paper.}
	%\Cam{You should say in a few sentences intuitively what these properties mean. Maybe right before or after the lemma}
	\begin{align*}
		\textnormal{Property I:}& \quad H'(t) \in [1 - \delta, 1] \quad \forall t \in \mathbb{R}: |t-d/2| \leq d \left (\frac{1}{2} - \frac{2}{s_1} \right ) s_3,\\
		\textnormal{Property II:}& \quad H'(t) \in [0, 1] \quad \forall t \in \mathbb{R}: d\left(\frac{1}{2} - \frac{2}{s_1}\right) s_3 \leq |t-d/2| \leq \frac{d}{2} s_3,\\
		\textnormal{Property III:}& \quad H'(t) \leq s_0 \cdot \left(s_1\left(\frac{|t-d/2|}{ds_3} - \frac{1}{2}\right) + 2\right)^{-l} \quad \forall t \in \mathbb{R}: |t-d/2| \geq \frac{d}{2} s_3,\\
		\textnormal{Property IV:}& \quad supp(\wh{H'}) \subseteq \left [-\frac{s_1 l}{2d s_3}, \frac{s_1 l}{2 ds_3} \right ] \text{ and so } \Delta_h := |supp(\wh{H'})| = s_1 l/(ds_3)= \poly(k,\log(d/\delta))/d.
	\end{align*}
\end{lemma}
%We now present the proof of Lemma \ref{lem:H_discrete_properties}, which closely follows that of Lemma \ref{lem:H_continuous} in \cite{chen2016fourier}.
%\kstodo{add dependence on$\epsilon$ here}
\begin{proof}[Proof of Lemma \ref{lem:H_discrete_properties}]
%\Cam{I would be better for the reader didn't have to ref to [11] to understand our paper. This was the major complaint of the reviewer. Can we not just restate Lemma 6.6 here in full so that the reader can then refer to it? (KS: I am not sure because the calculation spans many pages and it is identical in our case as well.)}
Consider the function $H':\mathbb{R}\to \mathbb{R}$ given in  Lemma \ref{lem:H_continuous}. %\Cam{We need to say what this function is for completeness.} 
To obtain our function $H$, we discretize $H'$ by restricting it to the integers. The first three time domain properties of our function thus follow directly from the first three properties of Lemma \ref{lem:H_continuous}. %\Cam{Preferably the reader shouldn't have to ref to [11] to see these properties. If we stated Lem 6.6, they could just see them there.}.
Discretizing $H'$ in time domain to obtain $H$ results in aliasing of $\wh{H'}$ in Fourier domain to obtain $\wh{H}$. However, by Property IV of Lemma \ref{lem:H_continuous}, we know that the support of $\wh{H'}$  is contained in $[-1/2,1/2]$ assuming $d\gg \poly(k)$. %\Cam{I'm not really sure what this means. support of anything mod 1 lies in $[0,1]$ (KS:yes sorry,fixed. only want to say aliasing does not change \wh{H}' because support in [0,1]=[-1/2,1/2] mod 1). To not have aliasing the support would have to lie in $[-1,1]$ right (KS: no in [0,1])? Without any modulo? Also its not clear form the statement of Prof IV that this is even true. To make it true you basically need to assume $d$ is large compared to poly(k) right? } 
Thus $\wh{H}(f)=\wh{H'}(f)$ for all $f\in [-1/2,1/2]$, thus implying that property 4 of Lemma \ref{lem:H_discrete_properties} follows from property 4 of Lemma \ref{lem:H_continuous}.

Finally, to prove properties 5 and 6, we must use discrete versions of Lemmas 5.1 and 5.5 of \cite{chen2016fourier}, which are used to prove Properties of 5 and 6 of Lemma 6.6 in \cite{chen2016fourier}. %\Cam{To prove what?}
The discrete version of Lemma 5.1 was shown in Lemma C.1 of 
\cite{eldar2020toeplitz}, 
and establishes that for any $k$-Fourier sparse $x^*:\mathbb{Z}\to \mathbb{C}$, $\forall i\in [d]$, $|x^* (i)|^2 \lsim k^6 \log^3(k)(\|x^*\|_d^2/d) $. It is also fairly easy to obtain the following bound from inspecting the proof of Lemma 5.5 in \cite{chen2016fourier}: for any $k$-Fourier sparse $x^*:\mathbb{Z}\to \mathbb{C}$ $\forall i\in \mathbb{Z}\setminus [d]$, $|x^* (i)|^2 \lsim k^{13} (ki/d)^{2.2 k}(\|x^*\|_d^2/d)$. %\Cam{Well if its easy we should describe it.} 
Using these bounds in the proof of Properties 5 and 6 of Lemma 6.6 of \cite{chen2016fourier} and replacing integrals with sums, we obtain Properties 5 and 6 of Lemma \ref{lem:H_discrete_properties}.
%\Cam{The lemmas below don't make sense. Nothing in the lemma statement is defined. Also are we still inside the proof block?} 
%\Cam{I feel like this is exactly what the reviewers didn't want. We should just give the proof here -- rather than forcing the reader to read 11 right? It would be ok if it were given in an appendix but it should be given.}
\end{proof}
%\Cam{How come this function doesn't need to be discretized? (KS: because we need its properties before discretization/aliasing. We then use its aliased version wherever we use it (In HashToBins for eg.) and its properties derived from the continuous version}
We will also need the following filter function from \cite{chen2016fourier}, which when convolved with, allows us to access the input signal whose Fourier transform is restricted to a desired interval.

%\Cam{Is this lemma trying to say that such a function exists? Or is it trying to refer to some spefific function. if the later, the function needs to be defined.}
%\Cam{Again if we are proving this maybe we should say `Variant of Lemma 6.7 of \cite{chen2016fourier}'. This should be done throughout the paper. Titling it as 'Lemma 6.7 of \cite{chen2016fourier}' implies that we are just directly copying it from their paper.}
\begin{lemma}[Lemma 6.7 of \cite{chen2016fourier}] \label{lem:G_discrete_properties}
	Given $B>1$, $\delta,k,w >0$, let  $l=O(k\log(k/\delta))$. %\Cam{Not clear. Is $l$ given? Or is $l$ just derived from $k$ and $\delta$?}
 Then there exists a function  $G:\mathbb{R}\rightarrow \mathbb{C}$ with continous Fourier transform $\wh{G}:\mathbb{R}\rightarrow\mathbb{C}$ satisfying the following, %\Cam{I don't get this notation. What does a pair of functions followed by a bracket mean?},
	\begin{align*}
		\textnormal{Property I:}& \quad \wh{G}(f) \in [1 - \delta/k, 1] \quad if |f| \leq (1-w)/2B.
		\\
		\textnormal{Property II:}& \quad \wh{G}(f) \in [0,1] \quad if (1-w)/2B\leq|f| \leq 1/2B.\\
		\textnormal{Property III:}& \quad \wh{G}(f) \in [-\delta/k,\delta/k] \quad if |f| \geq 1/2B.\\
		\textnormal{Property IV:}& \quad\text{supp}(G(t))\subset \left[\frac{-lB}{w},\frac{lB}{w}\right].\\
		\textnormal{Property V:}& \quad\max(G(t))\lsim \poly(B,l).
	\end{align*}
\end{lemma}
%\Cam{What is this remark refering to? Is this meant to be a proof of Lemma 2.7?}
\begin{comment}
\begin{proof}
\kstodo{checked correctness here, CKPS also implicitly use the same filter function after discretization in HashToBins}
    Here $G(t):\mathbb{Z}\rightarrow \mathbb{C}$ is the discretized version of the function in Lemma 6.7 in \cite{chen2016fourier} %\Cam{Do you mean a discritized version of the function in the Lemma?}. 
    Note that discretizing $G$ of Lemma 6.7 of \cite{chen2016fourier} in Lemma \ref{lem:G_discrete_properties} will lead to aliasing in the Fourier domain, but it can be easily seen that the properties of $\wh{G}(f)$ will be identical with the same parameters asymptotically. This is because $|\wh{G}(f)|$ decays roughly as $f^{-l}$ for $|f|\geq 1/2B$ (see the Proofs of Lemma 6.7 and 6.6 in \cite{chen2016fourier} for the definition of this filter function and its growth properties), and thus after aliasing in the Fourier domain the value of $\wh{G}(f)$ at any $|f|\geq 1$ will have at most $\delta/k$ contribution from its values at frequencies larger than $1$ due to its decaying tail.%\Cam{I don't understand this. What do you mean 'would lead to aliasing'. Does it lead to aliasing or not? And what do you mean the properties 'would be identical'. Are they identical or not? Also how does this follow from the tail decaying exponentially. We haven't even said anything about a decaying tail in the Lemma statement. So what is this refering to?}
\end{proof}
\end{comment}

\subsection{Structure preserving Toeplitz low-rank approximation.}\label{finalproof}
We next formally state the main result of \cite{kapralov2022toeplitz} which shows that for any PSD Toeplitz matrix, there exists a near optimal low-rank approximation in the Frobenius norm which itself is Toeplitz. We will use this fact in the proof of Theorem \ref{thm:sublinear_time_recovery} by interpreting the input PSD Toeplitz matrix as a noisy version of the near optimal Toeplitz low-rank approximation, further corrupted by noise $E$. %\Cam{Say in one sentence where we will use this.} %The first claim is formalized in the following theorem.
\begin{reptheorem}{frobeniusexistence}
	Given PSD Toeplitz matrix $T\in \mathbb{R}^{d\times d}$, $\eps,\delta \in (0,1)$, and an integer rank $k \le d$, let $r_1=\wt O(k/\eps)$ and $r_2 = \wt O(\log(1/\delta))$. There exists a symmetric Toeplitz matrix $\wt{T} = F_{S}DF_{S}^{*}$ of rank $r=2r_1r_2=\wt{O}(k\log(1/\delta)/\epsilon)$ such that,
	\begin{enumerate}
		\item $\|T-\wt{T}\|_F \leq (1+\eps)\|T-T_k\|_F + \delta\|T\|_F$.%\Cam{We already defined $T_k$ in the notation part so I think you can stop redefining everywhere it is used, including here} %\Cam{I don't get this last statement. It can't geneally be true. As some $T$ don't satisfy this. Is this meant to refer to $\wt T$? If so I think it should be stated as a separate property in this list.}
  %\item $ \lambda_{d}(\wt T)\geq -\delta \|T\|_F $.
		%\Cam{The frobenius norm error and the bound on the smallest eigenvalue need to be stated as separate properies. I.e. there should be 3 properties in this list.}
		\item \label{thm:existence_freq_grid}$F_{S} \in \R^{d \times r}$ and $D\in \R^{r \times r}$ are Fourier (Def. \ref{def:symmetric_fourier_matrix}) and diagonal matrices respectively. The set of frequencies $S$ can be partitioned into $r_1$ sets $\wt{S}_{1},\ldots, \wt{S}_{r_1}$ where each $\wt{S}_i$ is as follows,
		\begin{equation*}
			\wt{S}_i = \underset{1\leq j\leq r_2}{\bigcup}\{f_i+\gamma j,f_i-\gamma j\},
		\end{equation*}
		where  $f_i\in \{1/2d,3/2d,\ldots, 1-1/2d\}$  for all $i\in [r_1]$, and $\gamma = \delta /(2^{C \log^7d})$ for fixed constant $C>0$.%\Cam{Just use $C$ not $C_2$ if there are not any other constants in the statement.}
  \item Let $-\wt S_i = \{1-f| \forall f\in \wt S_i\}$, then $-\wt S_i \subseteq S$ for all $i\in [r_1]$.%\Cam{I'm a bit confused. Is this a claim that follows from (2) (KS: Yes, from the single cluster case)?} Moreover for every $f$ and $1-f\in S$ the entries in $D$ are identical. 
  \item Let $S_i =-\wt S_i\cup \wt S_i$ and $D_i$ contain the corresponding entries of $D$ for every $i\in [r_1]$. We have that for any subset $S'$ of $\cup_{i\in [r_1]}\{S_i\}$ %\Cam{Isn't this union just $S$? (KS: No, $S'$ is a set of sets)}
  , if we let $\wt T' = \sum_{S_i\in S'}F_{S_i}D_iF_{S_i}^*$, then there exists a PSD Toeplitz matrix $T'$ such that $\|\wt T' - T'\|_F\leq \delta \|T'\|_F$.
	\end{enumerate}
	%\Cam{Its weird in the theorem to define $r_1$ with big-Oh notation,m but state $r_2$ with a fixed constant $C$. Make consistent.}
\end{reptheorem}
Points 2,3 and 4 of the previous Lemma provide essential structural properties of the near optimal Toeplitz low-rank approximation crucial for our approach. Point 2 will be used in our sublinear time approximate regression primitive by essentially providing a finite set inside $[0,1]$ where possible frequencies of the near optimal Toeplitz low-rank approximation could lie. Point 3,4 implies in particular that $\wt T$ is almost PSD, which is an important property used to show that a random column of $\wt T$ must have similar signal-to-noise ratio as the full matrix and thus a sparse Fourier transform of a random column of $\wt T$ yields good frequencies for approximating $\wt T$. Point 4 in fact implies a more fine grained notion, which is useful for our proof that the recovered heavy frequencies suffice to approximate $\wt T$. 
%\Cam{jsut like with Lemma 3.2, you need to add some discussion after this theorem statement explaining to the reader what it means. Most notatbly (2) and (3). And in particular, the very last part of (3).(KS: not able to crisply say why last part of 3 is needed, it is just a technicality needed in the heavy light decomposition.)}

%% file: single_cluster.tex
% !TEX root = ./main.tex

\section{Discrete Time Off-Grid Sparse Fourier Recovery}\label{sec:dtft_recovery}
In this section we prove our main discrete time off-grid Fourier recovery result of Lemma \ref{lem:dtft_freqrecovery}, which was outlined in Section \ref{subsec:intro_sparsefft}. We re-state the complete version of the lemma here.
\begin{replemma}{lem:dtft_freqrecovery}
%\Cam{Break this up. First define $x^*$. Then define $g$ and $x$. Don't list all three together at the beginning of the sentence. There need to be multiple sentences here.}
Consider $x,x^*,g:\mathbb{Z}\to \mathbb{R}$, where $x^*(t) = \sum_{f\in S} a_f e^{2\pi i f t}$ for $S\subseteq [0,1]$ with $|S|\leq k$, $g(t)$ is arbitrary noise, and $x(t) = x^*(t)+g(t)$. Assume that we can access $x(t)$ for $t\in [d]$ and that $\|g\|_d^2 \leq c\|x^*\|_d^2$ for a small absolute constant $c>0$. Let $\delta > 0$ be an error parameter, and let $\mathcal{N}^2 = \frac{1}{d}(\|g\|_d^2 +\delta \|x^*\|_d^2)$ be the noise threshold. %\Cam{$\mathcal{N}^2$ should not be defined in this lemma since its never used in any of the statements of the lemma.} 
Let $\Delta=\poly(k,\log(d/\delta))/d$ such that $\Delta \geq k \cdot |supp(\wh{H}(f))|$, where $H(t)$ (with Fourier transform $\wh H(f)$) is the function guaranteed to exist by Lemma \ref{lem:H_discrete_properties} for parameters $k,\delta$.
Then in time and sample complexity $\poly(k,\log(d/\delta))$ one can find a list $L$ of $\poly(k,\log(d/\delta))$ frequencies that satisfies the following with probability at least $0.99$: 
Let $$S_{heavy}=\{f\in S: \exists f'\in L \text{ s.t. } |f-f'|_{\circ}\lsim k\Delta \sqrt{k\Delta d}\},$$ and let $x^*_{S_{heavy}}(t)=\sum_{f\in S_{heavy}}a_fe^{2\pi i f t}$. Then the following holds, $$\|x^*-x^*_{S_{heavy}}\|_d^2 \lsim d\mathcal{N}^2.$$ %\Cam{Define $\Delta$ before not after you use it. Define all variables before not after you use them.} %Such an $f$ is also referred to a ``heavy'' frequency. 
%\Cam{Don't abbreviate params.}
\end{replemma}
In the subsequent subsections we state intermediary claims and their proofs building up to the proof of Lemma \ref{lem:dtft_freqrecovery}. In Section \ref{subsection:one_cluster} we first consider the case when all frequencies of the Fourier sparse function $x^*$ are very close to each other -- i.e., the one-cluster case.  In Section \ref{subsection:one_cluster} we describe how to approximately recover the central frequency of the one cluster. 
Then in Section \ref{subsection:multi_cluster_freq_recovery} we reduce the general case to the one-cluster case. In Section \ref{subsubsection:multi_cluster_bounded_freqs}, we first give the reduction in the setting when the frequencies lie in a relatively small interval. We call such instances `bounded'. We then give a reduction from fully general instances to bounded instances in Section \ref{subsubsection:multi_cluster_reduce}. %\Cam{the reader won't konw here what a bounded instance is. So state what that is informally.}

\subsection{One cluster case.}\label{subsection:one_cluster}

In this section, we consider discrete time signals that are \emph{clustered} in Fourier domain with approximately bounded support in time domain, as formalized in the  following definition. % \Cam{Do you want to say clustered in fourier domain and with approxiamtely bounded support in time domain? Since Prop II really says the later and doesn't say anything about being clustered in Fourier domain.}
\begin{definition}[$(\epsilon,\Delta)$-one clustered signal]\label{def:one_clustered_signal}
	$z:\mathbb{Z}\to \mathbb{C}$ is $(\epsilon,\Delta)$-one clustered around $f_0 \in [0,1]$ if the following holds,
	\begin{align*}
		\textnormal{ Property I: } & \quad \int_{f_0 - \Delta}^{f_0 + \Delta} |\wh{z}(f)|^2 df \geq  (1 - \eps) \int_{0}^{1} |\wh{z}(f)|^2  df, \\
		\textnormal{Property II: } & \quad \|z\|_d^2 \geq  (1 - \eps) \sum_{t \in \Z} |z(t)|^2. 
	\end{align*}
\end{definition}
The main result of this subsection is that we can approximately recover the central frequency $f_0$ of a clustered signal in sublinear time only using samples of the function $z$ at on-grid time domain points $\{0,\ldots,d-1\}$. 
%\Cam{need to change lemma title. Since this is not just a lemma form CKPS.} 
\begin{lemma}[Variant of Lemmas 7.3, 7.17 of \cite{chen2016fourier}]\label{lem:on_grid_lemma_7_17}                                                               
	Let $z$ be a $(\epsilon,\Delta')$-clustered signal around $f_0$ (Def. \ref{def:one_clustered_signal}) for any $\epsilon$ smaller than an absolute constant. Suppose $f_0\in \mathcal{I}$ for an interval $\mathcal{I}\subseteq [0,1]$ of size $|\mathcal{I}|$ smaller than an absolute constant %\Cam{I'm not really sure what this means? Like $|\mathcal{I}|$ is a decreasing function? but of what arguments?} 
 with $\mathcal{I}$ known to the algorithm a priori. Let $\Delta' = O(k\Delta)$% \Cam{What is $k$ here? Presumably you are assuming $z$ has $k$-frequencies, but this was never mentioned.} f
 for $k,\Delta$ as defined in Lemma \ref{lem:dtft_freqrecovery}.%\Cam{This refs back to the informal version of the Lemma in the intro, which has no $\Delta$ in it. Can you  get to ref to the formal version in Sec 4?} 
 Then procedure \hyperref[alg:freq_rec_1_cluster]{\freqrecovcluster{}} %(Algo \ref{BLAH}) \Cam{Should give also number with name so it can be more easily ref'ed. In this instance and others.} 
 returns an $\wt f_0$ such that, with probability at least $1 - 2^{-\Omega(k)}$,
	\[
	|\wt{f}_0 - f_0|_{\circ} \lsim \Delta' \sqrt{\Delta' d}.
	\]
	%\Cam{Define $\Delta'$ before using it. Also I'm not sure what this mean. 'We will consider'. That isn't the type of language that should appear in a lemma statement. Just state what $\Delta'$ is equal to.}
 Moreover, \hyperref[alg:freq_rec_1_cluster]{\freqrecovcluster{}} has time and sample complexity of $\poly (k,\log(d/\delta))$.
\end{lemma}
We now present the main subroutines and their proofs of correctness that lead  the proof of Lemma \ref{lem:on_grid_lemma_7_17}.
\subsubsection{Sampling the signal.}\label{subsubsec:sampling_the_signal}
In this section, our main goal is to present Algorithm \ref{alg:OneGoodSample} and Lemma \ref{lem:getlegsamp_corr} which proves its correctness. The algorithm outputs a weak estimate of the central frequency $f_0$ of a clustered signal. This algorithm is then repeatedly invoked in the the algorithms presented in the subsequent sections to find a sharper estimate of the central frequency $f_0$ as per the statement of Lemma \ref{lem:on_grid_lemma_7_17}
\begin{algorithm}[H]
\caption{\onegoodsamp$(z(t))  $}\label{alg:OneGoodSample}
\begin{algorithmic}[1]
    
    \STATE Let $m= \poly(k,\log(d))$.
    \STATE Obtain $m$ i.i.d uniformly at random samples $x_1,\ldots,x_m$ from $[d]$.
    \STATE Query $z(x_i)$ for all $i\in [m]$.
    
    \STATE Let distribution $D_m$ supported on $x_{1},\ldots,x_m$ be such that the probability mass on $x_i$ is $D_{m}(x_i)=|z(x_i)|^2/(\sum_{i=1}^m |z(x_i)|^2)$.
    \STATE \textbf{Return:} $\alpha \sim D_{m}$.
    \end{algorithmic}
\end{algorithm}
We now state the main statement of the lemma. This lemma and its proof is the discrete time version of the Lemma 11 in \cite{chen_et_al:LIPIcs.ICALP.2019.36}.
\begin{lemma}[Counterpart of Lemma 7.2 in \cite{chen2016fourier}] \label{lem:getlegsamp_corr}
Let $z(t)$ be a $(\epsilon,\Delta')$-one clustered signal as per the setup of Lemma \ref{lem:on_grid_lemma_7_17}. Then there exists a constant $C$ such that for any $\beta< C/\Delta'$, Algorithm \ref{alg:OneGoodSample} when run on $z(t)$ returns $\alpha \in [d]$ such that $|z(\alpha+\beta)-z(\alpha)e^{2\pi i f_0 \beta}|\leq 0.01(|z(\alpha)|+|z(\alpha+\beta)|)$ with probability $0.9$. Moreover the runtime and sample complexity of Algorithm \ref{alg:OneGoodSample} is $\poly(k,\log(d))$.
\end{lemma}

To prove this lemma, we will need the following helper claims. First, let $z(t) = z_C(t)+z_{\bar{C}}(t)$ such that $\wh{z_{C}}(f) = \wh{z}(f)\mathbbm{1}_{f\in[f_0-\Delta',f_0+\Delta']}$ and $z_{\bar{C}}(t) = z(t) -z_C(t)$. The first claim we would need is the following,
\begin{claim}\label{claim:z_c_energy_comparable}
 For $z_{C}$ and $z_{\bar{C}}$ we have that   $\|z_C\|_{d} \geq (1-\sqrt{2\epsilon})\|z\|_d$ and $\|z_{\bar{C}}\|_d\leq \sqrt{2\epsilon} \|z\|_d$.
\end{claim}
\begin{proof}
    From Definition \ref{def:one_clustered_signal} of one-clustered signals, we have the following,
    \begin{align*}
        \|z_{\bar{C}}\|_d^2 &\leq \sum_{t\in Z}|z_{\bar{C}}(t)|^2\\
        & = \int_{0}^{1} |\wh{z_{\bar{C}}}(f)|^2 df\\
        & = \int_{[0,1]\setminus [f_{0}-\Delta',f_{0}+\Delta']}|\wh{z}(f)|^2 df\\
        & \leq \epsilon \int_{0}^{1} |\wh{z}(f)|^2 df\\
        &=\epsilon\sum_{t\in \mathbb{Z}}|z(t)|^2 \leq 2\epsilon \|z\|_d^2\quad \text{(By property 2 in Definition \ref{def:one_clustered_signal})}.
    \end{align*}
    Now observe that by triangle inequality, we know the following,
    \begin{equation*}
        \|z\|_d\leq \|z_{C}\|_{d}+\|z_{\bar{C}}\|_d.
    \end{equation*}
    We already know that $\|z_{\bar{C}}\|_d\leq \sqrt{2\epsilon}\|z\|_d$, thus we get that $\|z_{C}\|_{d}\geq (1-\sqrt{2\epsilon})\|z\|_d$.
\end{proof}

Next we have the following claim which shows that we can estimate $\|z\|_d$ using uniform sampling with $\poly(k,\log(d))$ samples. 
\begin{claim}\label{claim:z_norm_estimation}
    For $m=\poly(k,\log(d))$ samples, $((\sum_{i=1}^m|z(x_i)|^2)/m)^{1/2} = (1\pm 0.1)\|z\|_d$ with probability $0.9$.
\end{claim}
\begin{proof}
    First we will obtain a bound on the maximum value of $z_{C}(t)$ for any $t\in [d]$ compared to its average value $\|z_C\|_d$. Consider any $t\in [d]$, then we have the following,
    \begin{align*}
        |z_C(t)| &= \left|\int_{0}^{1} \wh{z}_{C}(f)e^{2\pi i ft} df\right|\\
        &\leq \int_{0}^{1} |\wh{z}_C(f)|df \\
        & = \int_{f_0-\Delta'}^{f_0+\Delta'} |\wh{z}_C(f)|df\\
        &\leq \sqrt{2\Delta'}(\int_{f_0-\Delta'}^{f_0+\Delta'}|\wh{z}_C(f)|^2 df)^{1/2}\\
        &\lsim \sqrt{2\Delta' d} \|z_C\|_d.
    \end{align*}
    where the final inequality followed from the following fact, $$\sum_{t\in \mathbb{Z}\setminus [d]}|z_C(t)|^2 \leq 2(\sum_{t\in \mathbb{Z}\setminus [d]}|z(t)|^2+ \sum_{t\in \mathbb{Z}\setminus [d]}|z_{\bar{C}}(t)|^2)) \leq 6\epsilon \|z\|_d^2.$$ Thus $\|z_C\|_d\geq (1-\sqrt{6\epsilon})(\sum_{t\in \mathbb{Z}}|z_{C}(t)|^2)^{1/2}$.
Since $\Delta'=\poly(k,\log(d))/d$, we have that $|z_{C}(t)|\leq \poly(k,\log(d)) \|z_C\|_d$ for all $t\in [d]$. Thus by a Chernoff bound, we know that if we take $m=\poly(k,\log(d))$ i.i.d. samples $x_1,\ldots, x_m$ from $[d]$ then $((\sum_{i\in [m]}|z_C(x_i)|^2)/m)^{1/2}  = (1\pm 0.01)\|z_C\|_d = (1\pm 0.05)\|z\|_d$ with probability at least $0.99$. Also note that $\mathbb{E}_{t\sim [d]}[|z_{\bar{C}}(t)|^2] = \|z_{\bar{C}}\|_d^2$. Thus by a Markov's inequality, we know that $((\sum_{i\in [m]}|z_{\bar{C}}(x_i)|^2)/m)^{1/2}\leq 10\|z_{\bar{C}}\|_d\leq 20\sqrt{\epsilon}\|z\|_d\leq 0.05\|z\|_d$ (for a small enough constant $\epsilon$) holds with probability $0.99$. Thus by triangle inequality for $\ell_2$ norm we have the following,
\begin{align*}
   ((\sum_{i\in [m]}|z(x_i)|^2)/m)^{1/2}&= ((\sum_{i\in [m]}|z_{C}(x_i)|^2)/m)^{1/2}\pm((\sum_{i\in [m]}|z_{\bar{C}}(x_i)|^2)/m)^{1/2}\\
   & =(1\pm 0.05)\|z\|_d \pm 0.05\|z\|_d\\
   & = (1\pm 0.1)\|z\|_d.
\end{align*}
By a union bound, this happens with probability $0.9$.
\end{proof}
The next helper claim will be the final claim we need before proving Lemma \ref{lem:getlegsamp_corr}.
\begin{claim}\label{claim:good_sample_on_avg}
    For a small enough constant $c>0$, for any $\beta < c/\Delta'$ we have that $y(t) = z(t)e^{2\pi i f_0\beta }-z(t+\beta)$ satisfies $\|y\|_d^2\lsim 0.001 \|z\|_d^2$.
\end{claim}
\begin{proof}
    First note that $\|y\|_d^2\leq \sum_{t\in \mathbb{Z}}|y(t)|^2 = \int_{0}^1 |\wh{y}(f)|^2 df$.
    Now we first bound $\int_{[0,1]\setminus[f_{0}-\Delta',f_0+\Delta']}|\wh{y}(f)|^2 df$ as follows,
    \begin{align*}
        \int_{[0,1]\setminus[f_{0}-\Delta',f_0+\Delta']} |\wh{y}(f)|^2 df &=  \int_{[0,1]\setminus[f_{0}-\Delta',f_0+\Delta']}|\wh{z}(f)\cdot e^{2\pi i  \beta f_0} - \wh{z}(f)e^{2\pi i \beta f}|^2 df\\
        &\leq 4\int_{[0,1]\setminus[f_{0}-\Delta',f_0+\Delta']}|\wh{z}(f)|^2 df\quad \text{(by Cauchy-Schwartz)}\\
        & \leq 4\epsilon \int_{0}^1 |\wh{z}(f)|^2 df\\
        & \leq 5\epsilon \|z\|_d^2\leq 0.0001\|z\|_d^2.
    \end{align*}
    Next we bound $\int_{f_{0}-\Delta'}^{f_0+\Delta'}|\wh{y}(f)|^2 df$ as follows,
    \begin{align*}
        \int_{f_{0}-\Delta'}^{f_0+\Delta'}|\wh{y}(f)|^2 df &=\int_{f_{0}-\Delta'}^{f_0+\Delta'}|\wh{z}(f)\cdot e^{2\pi i  \beta f_0} - \wh{z}(f)e^{2\pi i \beta f}|^2 df\\
        &\leq \int_{f_{0}-\Delta'}^{f_0+\Delta'}|\wh{z}(f)|^2 |e^{2\pi i  \beta f_0}-e^{2\pi i  \beta f}|^2 df
    \end{align*}
    Now for every $f\in [f_0-\Delta',f_0+\Delta']$, $\beta < c/\Delta'$ implies that by a Taylor expansion $|e^{2\pi i  \beta f_0}-e^{2\pi i  \beta f}|\leq 4\pi c<0.0001$ for small enough $c$. Thus we get the following,
     \begin{align*}
        \int_{f_{0}-\Delta'}^{f_0+\Delta']}|\wh{y}(f)|^2 df &\leq 0.0001\int_{f_{0}-\Delta'}^{f_0+\Delta'}|\wh{z}(f)|^2 df\\
        &\leq 0.0001 \int_{0}^1 |\wh{z}(f)|^2 df\\
        &\leq 0.0005 \|z\|_d^2.
    \end{align*}
    Combining these two bounds, we prove the statement of the claim.
\end{proof}
We are now ready to state the proof of Lemma \ref{lem:getlegsamp_corr}.
\begin{proof}
    For a random sample $\alpha\sim D_m$ we now consider the following,
    \begin{equation*}
        \mathbb{E}_{\alpha\sim D_m}\left[\frac{|z(\alpha)e^{2\pi i \beta f_0}-z(\alpha+\beta)|}{|z(\alpha)|^2}\right] = \mathbb{E}_{\alpha\sim D_m}\left[\frac{|y(\alpha)|^2}{|z(\alpha)|^2}\right] = \sum_{i=1}^m \frac{|y(x_i)|^2}{|z(x_i)|^2}\cdot \frac{|z(x_i)|^2}{(\sum_{j=1}^{m}|z(x_j)|^2)} = \frac{\sum_{i=1}^m|y(x_i)|^2}{\sum_{i=1}^m|z(x_i)|^2}.
    \end{equation*}
   Condition on Claims \ref{claim:z_norm_estimation} and \ref{claim:good_sample_on_avg} to hold. Then by Markov's inequality applied to the statement of Claim \ref{claim:good_sample_on_avg}, we get that with probability $0.99$ we have that $(\sum_{i=1}^m |y(x_i)|^2)/m\leq 0.0001 \|z\|_d^2$. This implies the following holds with probability $0.98$,
   \begin{equation*}
       \mathbb{E}_{\alpha\sim D_m}\left[\frac{|z(\alpha)e^{2\pi i \beta f_0}-z(\alpha+\beta)|}{|z(\alpha)|^2}\right]\leq \frac{0.0001 \|z\|_d^2}{0.99\|z\|_d^2}\leq 0.0001.
   \end{equation*}

Applying a Markov bound again, we get that with probability $0.9$ over an $\alpha\sim D_m$, we have that $|z(\alpha)e^{2\pi i \beta f_0}-z(\alpha+\beta)|\leq 0.001|z(\alpha)| $. This also implies $|z(\alpha+\beta)|=(1\pm 0.001)|z(\alpha)|$. Thus we also get that $|z(\alpha)e^{2\pi i \beta f_0}-z(\alpha+\beta)|\leq 0.001(|z(\alpha)|+z(\alpha+\beta)|$.
\end{proof}

\subsubsection{Frequency recovery.}\label{sec:single_cluster_freq_recovery}
%\kstodo{this section also verified. basically need $\beta$ to be real in order to be able to get absolute error, otherwise wont be able to distinguish between freqs that are dist. $1$ apart.}
The main goal of this subsection is to present algorithms that can estimate the central frequency of a clustered signal approximately, leading to the proof of Lemma \ref{lem:on_grid_lemma_7_17}.%Cam{This is too vague. Are we doing to prove Lemma 4.1, our main one-cluster recover primitive in this section? Or do something else?} 
These algorithms use the primitives presented in the previous section.

\begin{algorithm}[H]
\caption{\textsc{Locate1Inner}$(z, \Delta', d, \wh{\beta}, z_{emp}, \wh{L},R_{loc}$)}\label{alg:locateinner}
    \begin{algorithmic}[1]
    %\KwIn{Query access to $z$ on $[d]$, $\Delta$, $d$, $\wh{\beta}$, $z_{emp}$, $\wh{L}$, variables from \textsc{LocateSignal}}
    \STATE $v_q \gets 0, \forall q \in [t]$.
    \WHILE{$r = 1 \to R_{loc}$}
        \STATE Choose $\beta \in [0.5 \wh{\beta}, \wh{\beta}] \cap \Z$ uniformly at random.
        \STATE $\gamma \gets \onegoodsamp(z)$.

         \STATE $\theta' \gets \text{phase}(z(\gamma) / z(\gamma + \beta))/2\pi$. %\Cam{Not clear what $\phi$ is here... ALso these two statements should be put on separate lines. They merge visually and it looks like one statement.}
        \FOR{$q \in [t]$}
        	\STATE $\theta_q = \wh L - \Delta l/2 + \frac{q-0.5}{t}\Delta l$.
           	\STATE If $\|2\pi\theta' -2\pi \beta \theta_q\|_{\circ}  \leq s\pi$ then add vote to $v_q,v_{q-1},v_{q+1}$\quad \quad ($\|x-y\|_{\circ} = \min_{z\in \mathbb{Z}}|x-y+2\pi z|$ for any $x,y\in \mathbb{R}$.) %\Cam{Is this a redefinition of wrap around distance?}).
        \ENDFOR
    \ENDWHILE   
    \STATE $q^* \gets \{q | v_q > \frac{R_{loc}}{2}\}$.
    \STATE \textbf{Return:} $L \gets \wh L - \Delta l/2 + \frac{q^*-0.5}{t}\Delta l $.
    
    \end{algorithmic}
\end{algorithm}
\begin{algorithm}[H]
\caption{\locatesig$(z, d, F, \Delta', z_{emp},\mathcal{I}$)}\label{alg:locatesig}
    \begin{algorithmic}[1]
    %\KwIn{Query access to $z$ on $[d]$, $d$, $\Delta$, $z_{emp}$}
    \STATE $t = \log(d)$, $t' = t / 4$, 
    $D_{max} = \log_{t}(d)$, $R_{loc} \eqsim \log_{1/c}(tc)$ ($c<1/2$ is some constant),
    $L^{(1)} \gets \text{midpoint of } \mathcal{I}$, $i_0 = \log_{t'}(1/|\mathcal{I}|)+1$.
    \FOR{$i = i_0 \to D_{max}$}
        \STATE $\Delta l = 1 / (t')^{i - 1}$,
        $s \eqsim c$, $\wh{\beta} \gets \frac{ts}{2 \Delta' l}$.
        \IF{$\wh{\beta}\gsim d / (d \Delta)^{3/2}$}
            \STATE \textbf{Break}.
        \ELSE
            \STATE $L^{(i)} \gets \textsc{LocateInner}(z, \Delta', d, \wh{\beta}, z_{emp}, L^{(i - 1)},R_{loc})$.
        \ENDIF
    \ENDFOR
    \STATE \textbf{Return:} $L^{(D_{max})}$.
    
    \end{algorithmic}
\end{algorithm}
\begin{algorithm}[H]
\caption{\freqrecovcluster$(z, d, \Delta',\mathcal{I}$)}\label{alg:freq_rec_1_cluster}
    \begin{algorithmic}[1]
    
    \STATE $z_{emp} \gets \getempen(z, d, \Delta')$\;
    \FOR{$i = 1 \to O(k)$}
        \STATE $L_r \gets \textsc{Locate1Signal}(z, d, \Delta', z_{emp},\mathcal{I})$.
    \ENDFOR
    \STATE \textbf{Return:} $L^* \gets \operatorname{median} \{L_r | r \in [O(k)] \}$.
    \end{algorithmic}
\end{algorithm}
The following lemma formalizes the guarantees of the \hyperref[alg:locateinner]{\textsc{Locate1Inner}} primitive which is used iteratively in \hyperref[alg:locatesig]{\locatesig{}} to refine and narrow-down the estimate for $f_0$, the frequency around which $z$ is clustered. The final algorithm leading to the proof of Lemma \ref{lem:on_grid_lemma_7_17} \hyperref[alg:freq_rec_1_cluster]{\freqrecovcluster{}} then runs \hyperref[alg:locatesig]{\locatesig{}} multiple times and returns the median of all the runs to get an approximation to $f_0$ with high-probability%\Cam{Should say which of these is the final algo that leads to the main result -- Lemma 4.1}. 
This lemma is the analogue of Lemma 7.14 of \cite{chen2016fourier}, however since $\beta$ is always restricted to be an integer in Algorithm \hyperref[alg:locateinner]{\textsc{Locate1Inner}}, an alternate proof of correctness is needed. Another major difference is that since we already know an interval $\mathcal{I}$ of size $o(1)$ such that $f_0\in \mathcal{I}$, the initialization of the frequency searching primitive \hyperref[alg:locatesig]{\locatesig{}} uses this information.%\Cam{There needs to be text that explains what's going on in this lemma. E.g. what is $q$. What is the voting scheme? What are you trying to establish? (KS: sorry ran out of time, will do for arxiv version)}
\begin{lemma}[Variant of Lemma 7.14 of \cite{chen2016fourier}]\label{lem:single_cluster_voting}
Consider an invocation of \hyperref[alg:locateinner]{\textsc{Locate1Inner}} on inputs (as per Alg. \ref{alg:locatesig}) %\Cam{I don't follow this. It implies that all the other inputs can be set however we want. Is that really true? E.g. you seem to be setting $\hat \beta$ in a specific way.} 
such that there is a $q'\in [t]$ with $f_0 \in [\wh L - \Delta l/2 + \frac{q'-1}{t}\Delta l,\wh L - \Delta l/2 + \frac{q'}{t}\Delta l]$. Let $\beta$ be sampled uniformly at random from $[\frac{st}{4 \Delta l}, \frac{st}{2 \Delta l}] \cap \Z$ and let $\gamma$ denote the output of procedure \hyperref[alg:OneGoodSample]{\onegoodsamp}($z,\Delta',d,\beta,z_{emp},\wh{L},R_{loc}$). Then the following holds,
    \begin{itemize}
        \item with probability at least $1 - s$, $v_{q'}$ will increase by one,
        \item for any $|q - q'| > 3$, with probability at least $1 - 15s$ $v_{q}$ will not increase.
    \end{itemize}
\end{lemma}
\begin{proof}
    The proof follows that of Lemma 7.14 %\Cam{any numbered lemma, theorem, etc. should be captialized. I.e., Lemma not lemma.}
    of \cite{chen2016fourier}. In their notation $\theta = f_0$ and $\theta_q = \wh L - \Delta l/2 + \frac{q-1/2}{t}\Delta l$. The major and only difference lies in analyzing the case when $|q - q'| > 3$  and $|\theta - \theta_q|_{\circ} \geq \frac{\Delta l}{s t}$ and showing that in this case $v_q$ will not increase with high constant probability. Here, we adopt the analysis of \cite{hassanieh2012nearly} and instead of using Lemma 6.5 of \cite{chen2016fourier}, we will use a corollary of Lemma 4.3 of \cite{hassanieh2012nearly} since our choice of $\beta$ is a random \emph{integer} rather than a random \emph{real number} in some range. 
    This lemma is as follows.
    %\Cam{its weird to call it a corollary but then label it a lemma.}
    %\Cam{The notation in this lemma is really weird. Why do you need the tildes? Can you just use different variable names? The tildes look odd.}
    %\Cam{Need to say at a high level why we need this type of statement. Where will be used?}
    \begin{lemma}[Corollary of Lemma 4.3 of \cite{hassanieh2012nearly}]\label{lem:beta_integer}
    For some integer number $m$, if we sample $\beta$ uniformly at random from a set $T \subseteq [m]$ of $t$ consecutive integers, for any $i \in [d]$ and a set $S \subseteq [d]$ of $l$ consecutive integers,
\[
    \Pr[\beta i \mod d \in S] \leq \frac{1}{t} + \frac{i m}{d t} + \frac{lm}{d t} + \frac{l}{it}.
\]
\end{lemma}
We use Lemma \ref{lem:beta_integer} with the following values ---
    we set $m = \lceil \frac{st}{2 \Delta l} \rceil$, 
    $T = [\frac{st}{4 \Delta l}, \frac{st}{2 \Delta l}] \cap \Z$, 
    $t \geq \frac{st}{4 \Delta l} - 1$, 
    $S = [0, \frac{3s}{4} d ] \cap \Z$, 
    $l \leq \frac{3}{4}sd + 1$, 
    $i = d |\theta - \theta_q|_{\circ}$.
Without loss of generality we can assume that $i$ is an integer by rounding $\theta =f_0$ to the nearest integer multiple of $1/d$. This is feasible as the signal will still be clustered around this rounded $f_0$ since the width of the cluster $\Delta'=\poly(k,\log(d/\delta))/d\gg 1/d$. Recall that $t = \log d$. By assuming that $d$ is large enough, we assume that $\frac{st}{4 \Delta l} \geq \max(4, 10/s + 1)$. Note that $d \frac{\Delta l}{s t} \leq i \leq d \Delta l$. Observe that $\frac{m}{t} \leq (\frac{st}{2 \Delta l} + 1) / (\frac{st}{4 \Delta l} - 1) \leq 3 $. Hence,
    \begin{align*}
        \Pr[\beta i \mod d \in S] & \leq \frac{s}{10} + \frac{3 d\Delta l}{d} + 3 \frac{(3/4) s d + 1}{d} + \frac{(3/4)sd + 1}{d \Delta l (st)^{-1} \cdot (st (4 \Delta l)^{-1} + 1) } \\
        & \leq \frac{s}{10} + o(1)  + 9/4 s + 3 s \leq 7.5 s,
    \end{align*}
    %\Cam{You should end the equation with a comon since it is the middle of a sentence. Else you are starting a sentene with `where'.}
	where we used the fact that $\Delta l \leq |\mathcal{I}|\leq  s/1.5$ since $s= \Theta(1)$ and $|\mathcal{I}|$ is a small enough constant, By the same bound for $S = [-\frac{3s}{4}d, 0]$ and a union bound, we conclude
    \[
        \Pr[\beta d |\theta - \theta_q|_{\circ} \mod d \in [-\frac{3s}{4}d, \frac{3s}{4}d]] \leq 15s,
    \]
    which is equivalent to the following, %\Cam{here you are starting a sentence with 'which'. The previous equation needs a comma. It seems maybve you just blanket inserted periods into equations, which will not work. You need to decide the punctuation based on the context.}
    \[
        \Pr[2\pi \beta |\theta - \theta_q|_{\circ} \mod 2 \pi \in [-\frac{3s}{4} 2\pi, \frac{3s}{4} 2\pi]] \leq 15s.
    \] Recall that we denote $\|x-y\|_{\circ} = \min_{z\in \mathbb{Z}}|x-y+2\pi z|$ as the circular distance between $x,y$ for any $x,y\in \mathbb{R}$. %\Cam{I actually recall that we had been using this to denote the wraparound distance with a different definition...}
    Thus, overall we get that with probability at least $15s$, $\|2\pi \beta (\theta_q - \theta)\|_{\circ}  > (3s/4) 2\pi $. By triangle inequality, this further implies the following, $$\|2\pi \beta (\theta' - \theta_q)\|_{\circ} > \|2\pi \beta (\theta -  \theta_q)\|_{\circ} - \|2\pi \beta (\theta'-\theta)\|_{\circ}  > (3s/4)2\pi - s\pi/2 = s\pi. $$
   This implies that $v_q$ will not increase as per Line 8 of Algorithm \hyperref[alg:locateinner]{\textsc{Locate1Inner}}.%\Cam{I don't follow what this last sentence is saying. Also the algorithm does not have a line 39.}
 \end{proof}
The next lemma essentially gives the final guarantee of the \hyperref[alg:locatesig]{\locatesig{}} algorithm which is almost our final result, but with constant success probability. We will then use the median trick to boost the success probability.%\Cam{Say at a high level what that guarantee is... you should be guiding the reader. E.g. the should understand that this is almost our final result, but with constant success probability that will be boosted using the median trick.} which almost gives the proof of Lemma \ref{lem:on_grid_lemma_7_17}, but we then will need the median trick to boost the success probability from a constant to $1-2^{-\Omega(k)}$.
\begin{lemma}[Variant of Lemmas 7.15 and 7.16 in \cite{chen2016fourier}]\label{lem:locate1signal_constantprob}
	Consider the parameter setting as described in Algorithm \hyperref[alg:locatesig]{\locatesig{}}.
	The procedure \hyperref[alg:locateinner]{\textsc{Locate1Inner}} uses $R_{loc}$ legal samples and then procedure \hyperref[alg:locatesig]{\locatesig{}} runs \hyperref[alg:locateinner]{\textsc{Locate1Inner}} $D_{max}$ times to output a frequency $\wt{f_0}$ such that $|\wt{f_0} - f_0|_{\circ}\lsim \Delta' \sqrt{\Delta' d}$ with probability at least $0.9$. Moreover, \hyperref[alg:locatesig]{\locatesig{}} has time and sample complexity $\poly(k,\log(d/\delta))$.
\end{lemma}
\begin{proof}
	Equipped with Lemma \ref{lem:single_cluster_voting} %\Cam{equipped with the proof of the Lemma? or the Lemma itself?}, 
 the proof is identical to the proofs of Lemmas 7.15 and 7.16 in \cite{chen2016fourier}.
\end{proof}
Finally, the success probability can be boosted by repeating the procedure $O(k)$ times and using the median trick.

%\Cam{I woudln't restate this lemma. You should just give the proof of it. And state that you are proving it.}
\begin{lemma}   
{\ref{lem:on_grid_lemma_7_17}}                                                               
	Let $z$ be a $(\epsilon,\Delta')$-clustered signal around $f_0$ as per Definition \ref{def:one_clustered_signal} for any $\epsilon$ smaller than an absolute constant. Suppose $f_0\in \mathcal{I}$ for an interval $\mathcal{I}\subseteq [0,1]$ of size $|\mathcal{I}|$ smaller than an absolute constant, with $\mathcal{I}$ known to the algorithm apriori. Then procedure \hyperref[alg:freq_rec_1_cluster]{\freqrecovcluster{}} returns an $\wt f_0$ such that with probability at least $1 - 2^{-\Omega(k)}$,
	\[
	|\wt{f}_0 - f_0|_{\circ} \lsim \Delta' \sqrt{\Delta' d}.
	\]
%where $\Delta=\poly(k,\log(1/\delta))/d$ as per the setting of Lemma \ref{lem:dtft_freqrecovery}.
 Moreover, \hyperref[alg:freq_rec_1_cluster]{\freqrecovcluster{}} has time and sample complexity of $\poly (k,\log(d/\delta))$.
\end{lemma}
\begin{proof}
    Follows the proof of original lemma, Lemma 7.17 in \cite{chen2016fourier}.
\end{proof}

%% file: multi_cluster.tex
\subsection{General case.}\label{subsection:multi_cluster_freq_recovery}
Consider the setup and parameters of Lemma \ref{lem:dtft_freqrecovery}, where we have a general signal $x^*(t) = \sum_{f\in S} a_f e^{2\pi i f t}$ where $S\subseteq [0,1]$, $|S|\leq k$ and $g(t)$ is noise satisfying $\|g\|_d^2\leq c \|x^*\|_d^2$ for a small enough constant $c>0$. Our framework to recover frequencies from such an instance will be to first reduce a general instance to a \emph{bounded} instance, these are instances where $supp(\wh{x^*})$ is only contained in some interval of length $1/B$. We will use $B=\Theta(k^2)$. Then using hashing techniques, we will show how to recover frequencies from bounded instances by reducing them to one-cluster instances and then running one-cluster recovery primitive of Lemma \ref{lem:on_grid_lemma_7_17}. 

First we introduce the basic hashing primitives that will be needed multiple times throughout this section. The following definition introduces the hash function which maps frequencies in the Fourier domain to $B$ buckets, and the filter function which when convolved with gives access to the function containing frequencies restricted to any desired bucket. This is the standard notation identical to Definitions 6.3 and 6.8 in \cite{chen2016fourier}.
\begin{definition}[Hashing  and filtering notation]\label{def:permutation}
	Let $h_{\sigma, b}:[0,1]\rightarrow \{0,1,\ldots,B-1\}$ defined as follows for any $\sigma\in \mathbb{R},b\in [0,1]$, $$h_{\sigma, b}(f^*) = \text{round}(\frac{B}{2\pi} \cdot (2\pi \sigma(f^*-b)\mod 2\pi ))\quad \forall f^*\in [0,1].$$
	Let $G^{(j)}_{\sigma, b}$ be the function which when convolved with allows us to access the function restricted to bin $j$ for all $j\in [B]$ (via $G$ as per Lemma \ref{lem:G_discrete_properties}): 
	\begin{align*}
	\wh{G}^{(j)}_{\sigma,b}(f) &= \wh{G}^{\text{dis}}(\frac{j}{B}-\sigma f-\sigma b):=\underset{i\in \mathbb{Z}}{\sum} \wh{G}(i+\frac{j}{B}-\sigma f-\sigma b)\quad \forall f\in [0,1],
  \\
  G^{(j)}_{\sigma, b}(t) &= DTFT(\wh{G}^{(j)}_{\sigma,b}(f)).
	\end{align*}
	When necessary, we will make explicit the parameters $B,\delta,w,k$ used in the construction of $G$ as per Lemma \ref{lem:G_discrete_properties} and therefore $G^{(j)}_{\sigma, b}$. 
	
\end{definition}
Next we present the notation for the functions obtained by convolving the input with the filter function $G^{(j)}_{\sigma,b}(t)$.
\begin{definition}\label{def:convolved_signal}
    Let $H(t)$ be as per Lemma \ref{lem:H_discrete_properties} for parameters $k,\delta$ and $G^{(j)}_{\sigma,b}$ as per Def. \ref{def:permutation}. Let $z^{(j)} = (x^* \cdot H)\ast G^{(j)}_{\sigma,b}$, $g^{(j)} = (g\cdot H)\ast G^{(j)}_{\sigma,b}$ and $ x^{(j)} = (x\cdot H)\ast G^{(j)}_{\sigma,b} = z^{(j)}+ g^{(j)}$, for all $j\in [B]$.
\end{definition}
Finally we present the algorithm \hyperref[alg:hashtobins]{\hashtobins{}} and its guarantees from \cite{chen2016fourier} which computes $z^{(j)}(t)$ for any integer $t$ of one's choice for all $j\in [B]$. The proof is slightly modified to work with the DTFT (see Definition \ref{def:dtft}).\\
\begin{algorithm}[H]
\caption{\hashtobins($x,H,G,B,\sigma,\alpha,b,\delta,w=0.001(\text{default})$))}\label{alg:hashtobins}
\begin{algorithmic}[1]
	\STATE Let $(G(t),\wh{G}(f))$ be the filter functions as per Definition \ref{lem:G_discrete_properties} with parameters $B,\delta,w$.
	\STATE Let $W(t) = x\cdot H(t), D =O(\log(k/\delta))$ such that $|supp(G(t))|=2BD$
	and $V$  as $V[j] = G[j]\cdot W(\sigma(j-\alpha))e^{2\pi i \sigma b}$ for $j\in [-BD,BD]$ (here $G[j] = G(j) \forall j\in [-BD,BD] $ is the discretization of $G(t)$).
	\STATE Let $v\in \mathbb{R}^{B}$ as $u[j] = \sum_{i\in [-D,D]}V[j+iB]$ $\forall j\in [B]$.
	
		\STATE \textbf{Return:} $\textsc{FFT}(v)$;
\end{algorithmic}	
\end{algorithm}
\begin{lemma}[Variant of Lemma 6.9 in \cite{chen2016fourier}]\label{lem:hashtobins}
Let $u\in \mathbb{C}^{B}$ be the output of \hyperref[alg:hashtobins]{\hashtobins{}}$(x,H,G,B,\sigma,\alpha,b,w)$, and assume $\sigma$ and $\sigma \alpha \in \mathbb{Z}$. Then $u[j] = x^{(j)}(\sigma \alpha)$ $\forall j\in [B]$ where $x^{(j)}$ is as per Definition \ref{def:convolved_signal}. Let $D=O(\log(k/\delta)/w)$ such that $|supp(G(t))|=2BD$ as per Lemma \ref{lem:G_discrete_properties}. Then \hyperref[alg:hashtobins]{\hashtobins{}} takes the following samples from $x$ - $\{x(\sigma(i-\alpha))\}_{i=-BD}^{BD}$, and runs in time $O(B\log(k/\delta)/w)$.
\end{lemma}
\begin{proof}
The proof is identical to the proof of Lemma 6.9 in \cite{chen2016fourier}, but we just need to observe that in the last line of the proof where they conclude that $$\wh{u}[j] = \int_{-\infty}^{\infty}\wh{W}(s)\cdot \wh{G^{dis}}(j/B - \sigma s -\sigma b)e^{-2\pi i \sigma a s} ds,$$ for all $j\in [B]$, $\wh{W}(f) = \wh{x\cdot H}(f)$ is the continuous Fourier transform of $(x\cdot H)(t)$ which is considered to be a continuous signal from $\mathbb{R}\rightarrow \mathbb{C}$ (Here the continuous version of $H$ as per Lemma 6.6 of \cite{chen2016fourier}, i.e. before discretization in time domain to obtain the $H$ function as per Lemma \ref{lem:H_discrete_properties}). We however need to consider the discretized version of this signal. Now observe that since $\sigma$ and $\sigma a$ are integers we have that,
$$\wh{G}^{dis}(j/B - \sigma (s+i) -\sigma b)e^{-2\pi i \sigma a (s+i)}=\wh{G}^{dis}(j/B - \sigma s -\sigma b)e^{-2\pi i \sigma a s},$$ for all $i \in \mathbb{Z}$, where $\wh{G}^{dis}$ is as per Definition \ref{def:permutation}. Thus we can rewrite the definition of $\wh{u}[j]$ as $$\wh{u}[j]= \int_{0}^{1}\left(\sum_{i\in \mathbb{Z}}\wh{W}(s+i)\right)\cdot \wh{G^{dis}}(j/B - \sigma s -\sigma b)e^{-2\pi i \sigma a s} ds,$$
where $\sum_{i\in \mathbb{Z}}\wh{W}(s+i)$ is the exactly the DTFT of $(x\cdot H)(t)$ when $t$ is restricted to $\mathbb{Z}$. Thus $\wh{u}(j) = (x\cdot H \ast G^{(j)}_{\sigma,b})(\sigma a)=x^{(j)}(\sigma a)$ where $G^{(j)}_{\sigma,b}(t)$ is as per definition \ref{def:permutation} and $(x\cdot H)(t),G^{(j)}_{\sigma,b}(t)$ are both discrete time signals as is in our case.
\end{proof}

Equipped with these tools we now describe a clustering based pre-processing step, also described in \cite{chen2016fourier}, applied to general instances before explaining what are bounded instances and how to reduce to these instances. 
\begin{definition}[Clustering]\label{def:clustering}
Consider $H$ as per Lemma \ref{lem:H_discrete_properties} for parameters $k,\delta$. For any two frequencies $f_1,f_2$ in the support of $\wh {x^*} (f)$ we say that $f_1 \sim f_2$ if their supports overlap after convolving with $\wh H$, i.e. $supp(\wh {H \cdot e^{2\pi i f_1 t}})\cap supp(\wh {H \cdot e^{2\pi i f_2 t}})\neq \emptyset$. We cluster frequencies by taking a transitive closure under this relation $\sim$ and define $C_1,C_2,\ldots, C_l$, where $0\leq l\leq k$, as the clusters. Thus, each $C_i \subseteq [0,1]$ is an interval, $C_1\cup \ldots \cup C_l = supp(\wh {x^* \cdot H }(f))$ and $C_i \cap C_j = \emptyset$ for any $i \neq j$. 
 \end{definition}
 \begin{remark}\label{remark:cluster_width}
  Let $\Delta_h= |supp(\wh{H}(f))|$ as per Lemma \ref{lem:H_discrete_properties}, and let the $\Delta$ parameter of Lemma \ref{lem:dtft_freqrecovery} be set such that it satisfies $\Delta\geq k\Delta_h$. Note that thus width of any cluster is at most $ k\Delta_h \leq \Delta = \poly(k,\log(d/\delta))/d$.
 \end{remark}
$\Delta$ will be fixed throughout this section as per Remark \ref{remark:cluster_width}. This is the same $\Delta$ as set in the statement of Lemma \ref{lem:dtft_freqrecovery}. Equipped with these tools and notations, we now define a bounded instance as follows.

\begin{definition}[$(\mathcal{I},\delta')$-bounded instance]\label{def:bounded_freq_case}
	Let $\mathcal{I}\subseteq [0,1]$ be an interval satisfying $|\mathcal{I}|\leq 1/B$, and $\epsilon>0$ be a small enough constant.
	Let $x^*(t) = \sum_{f\in S} a_f e^{2\pi i f t}$ where $S\subseteq [0,1]$, $|S|\leq k$. Let $H(t)$ be as per Lemma \ref{lem:H_discrete_properties} for parameters $k,\delta$, and $g(t)$ be noise. Then the instance $x(t) = (x^* (t)+ g(t))\cdot H(t)$ is an $(\mathcal{I},\delta')$-bounded instance if the following is satisfied - All clusters $C$ as per Definition \ref{def:clustering} applied to $x^*\cdot H$ satisfy $C\subseteq \mathcal{I}$ and $\int_{[0,1]\setminus \mathcal{I}}|\wh {g\cdot H} (f)|^2 df\leq \delta'$. The noise threshold $\mathcal{N}^2$ is defined as $$\mathcal{N}^2 = \frac{\delta}{d} \int_{[0,1]}|\wh{x^*\cdot H}(f)|^2 df + \frac{1}{d}\int_{\mathcal{I}}|\wh{g\cdot H}(f)|^2df+ k\delta'/d\epsilon.$$
\end{definition}
 We will only be able to recover ``heavy'' frequencies from a bounded instance, and thus next we present their definition.
% However, it will be necessary to alter $G$ via $\delta \rightarrow \epsilon\delta$ \hl{to check}, i.e. to make $\hat{G}$ closer to finitely supported.

\begin{definition}[Heavy frequency]\label{def:heavy_freq}
	Consider the setup of Definition \ref{def:bounded_freq_case}. Call a frequency $f^*\in [0,1]$ heavy frequency if it satisfies
	\begin{align*}
		\int_{f^*-\Delta}^{f^*+\Delta}|\widehat{x^*\cdot H}(f)|^2 df \gsim \frac{d \mathcal{N}^2}{k}, 
	\end{align*}
	where the interval $[f^*-\Delta,f^*+\Delta]$ is modulo 1.% (the interval is modulo 1).
\end{definition}

For any heavy frequency $f^*$ we would be interested in isolating the cluster of width $O(k\Delta)$ around it by making sure no other heavy frequency maps to the bucket $f^*$ maps to. This is to ensure that the signal restricted to that bucket becomes a one-cluster signal, thus allowing the application of Lemma \ref{lem:on_grid_lemma_7_17} to recover $f^*$ approximately. Thus next we define what a ``well-isolated'' frequency is.
\begin{definition}[Well-isolated frequency, from \cite{chen2016fourier}]\label{def:well_isolated_f}
	$f^*$ is well-isolated under the hashing $(\sigma, b)$ if, for j=$h_{\sigma, b}(f^*)$, the signal $z^{(j)}$ (as per Definition \ref{def:convolved_signal}) satisfies 
	$$ \int_{[0,1]\backslash (f^*-200k\Delta, f^*+200k\Delta)} |\wh{z}^{(j)}(f)|^2df \lesssim \epsilon d \mathcal{N}^2/k.$$
\end{definition}

\subsubsection{Reducing from a bounded multi cluster instance to $\Theta(k^2)$ one cluster instances.}\label{subsubsection:multi_cluster_bounded_freqs}
In this section, we show how to reduce a bounded instance to $B=\Theta(k^2)$ clustered instances.
%\hltodoil{We should include diagrams from CKPS of the $G$ function and maybe $G^j$ with $\sigma=1$}

This is presented as the main result of this section below, Lemma \ref{lem:on_grid_lemma_7_6}. It extends Lemma 7.6 of \cite{chen2016fourier} for the case when we can only take samples on integer points in time domain. It essentially states that if we hash the at most $k$ clusters in the Fourier spectrum of the input into $B=\Theta(k^2)$ buckets then all clusters are simultaneously isolated. Moreover the clusters which are heavy and the corresponding bins to which they hash have high SNR then they are essentially one-clustered as per Def. \ref{def:one_clustered_signal}. We first state this main lemma, then proceed with sub lemmas and their proofs that will build up to the proof of this main lemma.
\begin{lemma}[Variant of Lemma 7.6 of \cite{chen2016fourier} for on-grid samples]\label{lem:on_grid_lemma_7_6} Let $x(t) = (x^*(t)+g(t))\cdot H(t)$ be a $(\mathcal{I},\delta')$-bounded instance as per Definition \ref{def:bounded_freq_case}. Let $G^{(j)}_{\sigma,b}(t)$ for all $j\in [B]$ be as per Def. \ref{def:permutation} and Lemma \ref{lem:G_discrete_properties} for parameters $B=\Theta(k^2),\delta=\epsilon\delta/k$, $w=0.01$ and $k=k$. Apply the clustering procedure of Definition \ref{def:clustering} to $x^*\cdot H$ and let $k_C\leq k$ be the total number of clusters in $\wh{x^*\cdot H}$. Let $\sigma$ be a u.a.r \textbf{integer} in $[\frac{1}{200Bk\Delta}, \frac{1}{100Bk\Delta}]$ and $b$ be a u.a.r. real number from $[0, \frac{1}{\sigma}]$. With prob. at least $1-k_C^2/B-k_C/99k$ the following is true - 

Consider any cluster $C$ with $f^*$ the midpoint of $C$. Then $f^*$ is well-isolated (Def. \ref{def:well_isolated_f}). Moreover for $j=h_{\sigma,b}(f^*)$ if,
\begin{enumerate}
    \item $f^*$ is heavy as per Definition \ref{def:heavy_freq},
    \item $\int_{[0,1]} |\wh{g}^{(j)}(f)|^2df\lsim \epsilon \int_{f^*-\Delta}^{f^*+\Delta} |\wh{x^*\cdot H}(f)|^2df$ ($g^{(j)}$ as per Definition \ref{def:convolved_signal}),
\end{enumerate}
then $x^{(j)}$ (as per Definition \ref{def:convolved_signal}) is an $(2\sqrt{\epsilon}, 200k\Delta)$-one clustered signal around $f^*$ with interval $\mathcal{I}$. 
\end{lemma}
The particular choice of $\sigma$ is because we need to ensure that the time points at which \hyperref[alg:hashtobins]{\hashtobins{}} accesses the input as per Lemma \ref{lem:hashtobins} are integers. Thus $\sigma$ and $\alpha$ must be chosen such that $\sigma$ and $\sigma \alpha$ are both integers (see Lemma \ref{lem:hashtobins}). This is done by choosing $\sigma$ to be a random integer in a bounded interval and $\alpha$ such that $\sigma \alpha$ is an integer and it equals the time point at which we want to access the functions $x^{(j)}$. However this introduces complications because if $\sigma$ is a random integer as opposed to a random real number in a bounded interval then the collision behaviour under the random hashing changes. Thus we first start with the main lemma that bounds the collision probability of two frequencies under the hash function of Definition \ref{def:permutation}. This proof is different from Claim 6.4 in \cite{chen2016fourier} because $\sigma$ is not being random real number in a bounded interval, we can only ensure small collision probability for frequencies not more than $1/B$ apart.

\begin{lemma}[Hashing collision probabilities, analog of Claim 6.4 of \cite{chen2016fourier}]\label{lem:on_grid_claim_6_4}
Let $\sigma$ be a u.a.r. integer in $[\frac{1}{200Bk\Delta}, \frac{1}{100Bk\Delta}]$. Let $h_{\sigma,b}$ be as per Definition \ref{def:permutation} for $\sigma$ and arbitrary $b\in [0,1]$. Then,
\begin{enumerate}
    \item For any $f_1,f_2\in [0,1]$ s.t. $200k\Delta \leq |f_1-f_2| < 200(B/2-0.5)k\Delta$, then $\mathbb{P}_{\sigma}[h_{\sigma, b}(f_1) = h_{\sigma, b}(f_2)] = 0$.
    \item For any $f_1,f_2\in [0,1]$ s.t. $200(B/2-0.5)k\Delta <|f_1-f_2| < \frac{1}{B}$, then $\mathbb{P}_{\sigma}[h_{\sigma, b}(f_1) = h_{\sigma, b}(f_2)] \lesssim \frac{1}{B}$.
\end{enumerate}
\end{lemma}
\begin{proof}
For convenience, let $F = |f_1-f_2|$. For a hash collision to occur, $h_{\sigma, b}(f_1) = h_{\sigma, b}(f_2)$, it must be that $2\pi \sigma F \in (s\cdot 2\pi - \frac{2\pi}{2B}, s\cdot 2\pi + \frac{2\pi}{2B})$ for some integer $s$ as per the definition of $h_{\sigma,b}$ in Def. \ref{def:permutation}.
    \begin{enumerate}
        \item The proof of this case is unchanged when $\sigma$ is restricted to integers. For completeness: if $200k\Delta \leq F < 200(B-0.5)k\Delta$, then $\frac{2\pi}{B} \leq |2\pi\sigma F| < \frac{2\pi}{100Bk\Delta} 200(B/2-0.5)k\Delta = (1-\frac{1}{2B})2\pi$. Thus collision is impossible. 
        \item %\hl{this gives intuition but has a bug: it doesn't take into account the range of $\sigma$.} 
         The collision condition is equivalent to $\sigma F \in (s - \frac{1}{2B}, s + \frac{1}{2B})$ for some integer $s$.  Note that the range of possible values of $s$ is $\{\lfloor \frac{F}{200Bk\Delta }\rfloor,\ldots,\lceil \frac{F}{100Bk\Delta}\rceil \}$.The rest of the proof proceeds by a simple counting argument based on lengths of intervals. Assume $F < \frac{1}{B}$. 
        %\hl{add floor/ceiling and justify that bounds make sense [eg non-neg] because $F$ is small.}
         For any integer $s$, there are at most $\frac{1/B}{F}+1=\frac{1}{BF}+1$ integer values of $\sigma$ such that $\sigma F \in (s - \frac{1}{2B}, s + \frac{1}{2B})$. Taking into account the total number of possible values of $s$, we get that the total number of such integer values $\sigma$ is at most 
        \begin{align*}
        	(\frac{1}{BF}+1)\cdot ( \lceil \frac{F}{100Bk\Delta}\rceil - \lfloor \frac{F}{200Bk\Delta }\rfloor +1) &\leq  (\frac{1}{BF}+1) \cdot (\frac{F}{200Bk\Delta}+2)\\
        	& \lsim \frac{1}{200B^2k\Delta}.
        \end{align*}
        Moreover, there are at least $\frac{(B-1)/B}{F} = \frac{B-1}{BF}$ integer values of $\sigma$ such that $s + \frac{1}{2B} < \sigma F < (s+1) - \frac{1}{2B}$, i.e. such that $\sigma F$ lies in between this interval and the next one. Again, taking into account the total possible values of $s$ we get that the total number of such integer $\sigma$ is at least
        \begin{align*}
        	(\frac{B-1}{BF})\cdot (\max (  \lceil \frac{F}{100Bk\Delta}\rceil - \lfloor \frac{F}{200Bk\Delta }\rfloor -3,1))
        	&\gsim (\frac{B-1}{BF})\cdot(\frac{F}{200Bk\Delta})\\
        	&\gsim \frac{1}{200Bk\Delta}.
        \end{align*}
        Note that using the notation $\lsim$ and $\gsim$ is correct as $\frac{1}{BF}$ and $\frac{F}{200Bk\Delta}$ are at least $\Omega(1)$. Thus we can upper bound the collision probability up to constants by $\frac{1/200B^2k\Delta}{1/200Bk\Delta}=\frac{1}{B}$.
    \end{enumerate}
\end{proof}

 We need the following definition before proceeding.
\begin{definition}[Nice cluster]\label{def:nice_cluster}
Consider the setup of Lemma \ref{lem:on_grid_lemma_7_6}.    Let $C$ be a cluster with midpoint $f^*$ as per Definition \ref{def:clustering}. Let $j=h_{\sigma,b}(f^*)$ and $\wh{G}^{(j)}_{\sigma,b}(f)$ be the corresponding filter function as per Definition \ref{def:permutation}. Then $C$ is a nice cluster if $|\wh{G}^{(j)}_{\sigma,b}(f)|\geq 1-\epsilon\delta/k$ for every $f\in C$.
\end{definition}
Now we move on to state and prove the lemma that bounds the probability of isolating all frequencies (as per Definition \ref{def:well_isolated_f}) under the random hashing $h_{\sigma,b}$.
\begin{lemma}[Variant of Lemma 7.19 of \cite{chen2016fourier} for on-grid samples]\label{lem:on_grid_7_19_well_isolated} Consider setup of Lemma \ref{lem:on_grid_lemma_7_6}. Then with probability $1-k_C^2/B-k_C/99k$ over the randomness in $\sigma,b$ the following holds - 

For all clusters $C$ with $f^*$ being the midpoint of $C$, $f^*$ is well isolated and $C$ is nice (as per Definition \ref{def:nice_cluster}). Moreover, for $j=h_{\sigma,b}(f^*)$ and $z^{(j)}$ as per Definition \ref{def:convolved_signal},
 \begin{equation*}
 	\int_{f^* -200k\Delta}^{f^* +200k\Delta} |\wh z^{(j)}(f)|^2 df \in [1-\delta,1] \int_{f^*-\Delta}^{f^*+\Delta} |\wh{x^*\cdot H}(f)|^2 df.
 \end{equation*}
\end{lemma}
\begin{proof}
    Consider any cluster $C$ with $f^*$ being its central frequency. We divide by cases and draw on Lemma \ref{lem:on_grid_claim_6_4}.
    \begin{itemize}
        \item Due to the random shift $b$, the entire interval where $|\wh{G}^{(j)}_{\sigma,b}(f)|\in [1-\epsilon\delta/k,1]$ is randomly shifted by $\sigma b$ (see Def. \ref{def:permutation}). Furthermore since the width of the interval containing $f^*$ and where $|\wh{G}^{(j)}_{\sigma,b}(f)|\in [1-\epsilon\delta/k,1]$ is at least $99k\Delta$ and the width of $C$ is $\Delta$, the entire cluster $C$ lands in the same bucket as $f^*$ and the region where $|\wh{G}^{(j)}_{\sigma,b}(f)|\in [1-\epsilon\delta/k,1]$ with probability $1-1/99k$. 
        \item If $200k\Delta \leq |f'-f^*| < \frac{1}{\sigma}-\frac{1}{\sigma B} $, then $\frac{2\pi}{B} \leq |2\pi\sigma (f^*-f')|\leq 2\pi-\frac{2\pi}{B} $  implies that the two frequencies are always mapped to different buckets as per Definition \ref{def:permutation}. As a result, they fall in the region where $|\wh{G}_{\sigma,b}^{(j)}(f)|\leq\frac{\epsilon \delta}{k}$. Any such $f'$ contributes at most $\frac{\epsilon \delta}{k} \int_{f'-\Delta}^{f'+\Delta} |\wh{x^* \cdot H}(f)|^2 df$ to the energy in Fourier domain, and thus the total contribution of all such $f'$ is at most $\frac{\epsilon \delta}{k} \int|\wh{x^* \cdot H}(f)|^2 df \leq \epsilon d \mathcal{N}^2/k$ to the energy in Fourier domain.
        \item If $\frac{1}{\sigma} -\frac{1}{\sigma B}\leq |f'-f^*| \leq \frac{1}{B}$, then by Lemma \ref{lem:on_grid_claim_6_4}, the probability of a hash collision is at most $\frac{1}{B}$, and by taking a union bound over the at most $k_C$ such cluster midpoints $f'$ in the spectrum of $\wh{x^*}$, we have that no such $f'$ lands in the same bucket as $f^*$ with probability at least $1-1/B$. Thus, for any such $f'$, the entire interval $f'\pm \Delta$ again falls in the $\frac{\epsilon \delta}{k}$-valued tail of $\wh{G}$. This implies the total contribution from all such clusters with midpoint $f'$ to the energy in the Fourier domain is at most $\frac{\epsilon \delta}{k} \int|\widehat{x^* \cdot H}(f)|^2 df \leq \epsilon d \mathcal{N}^2/k$.
    \end{itemize}
    Points 2 and 3 imply that $f^*$ is well-isolated. Point 1,2 and 3 implies that 
 \begin{equation*}
	\int_{f^* -200k\Delta}^{f^* +200k\Delta} |\wh z^{(j)}(f)|^2 df \in [1-\delta,1] \int_{f^*-\Delta}^{f^*+\Delta} |\wh{x^*\cdot H}(f)|^2 df.
\end{equation*}
with probability $1-k_C/B-1/99k$. Taking a union bound over at most $k_C$ such clusters $C$, the proof of the Lemma is finished. Point $1$ after the union bound over all $k_C$ clusters implies all such clusters are nice (as per Definition \ref{def:nice_cluster}).
\end{proof}
Then the next lemma we need shows that even after multiplying with $\wh{G}^{(j)}_{\sigma,b}$, the signal's energy in the time domain remains concentrated on $[0,d]$.
\begin{lemma}\label{lem:concentrated_in_time}
Condition on the guarantee of Lemma \ref{lem:on_grid_7_19_well_isolated} holding. Let $C$ be any cluster with midpoint $f^*$ such that $f^*$ is heavy as per Definition \ref{def:heavy_freq}. Then $z^{(j)} (t)$ for $j=h_{\sigma,b}(f^*)$ satisfies,
	\begin{equation*}
		\sum_{t\in [-\infty,0]\cup [d,\infty]}|z^{(j)}(t)|^2   \lsim  \epsilon \sum_{t=-\infty}^{\infty}|z^{(j)}(t)|^2.
	\end{equation*}
\end{lemma}

\begin{proof}
From Lemma \ref{lem:on_grid_7_19_well_isolated}, we know that $C$ is nice and $f^*$ is well isolated. Let $x^*_C(t) = \underset{f\in supp(\wh{x^*}):f\in C}{\sum}v_fe^{2\pi i ft}$ Thus we have the following, $$\int_{[0,1]}|\wh{z}^{(j)}(f)-\wh{x^*_C\cdot H}(f)|^2df\lsim \epsilon\delta/k \int_{[0,1]}|\wh{x^*\cdot H}(f)|^2df\leq \epsilon d\mathcal{N}^2/k.$$ This implies we have the following by Cauchy-Schwarz,
\begin{align}\label{eqn:z_j_time}
    \int_{[0,1]}|\wh{z}^{(j)}(f)|^2 &\gsim \int_{[0,1]}|\wh{x^*_C\cdot H}(f)|^2 df - \epsilon d\mathcal{N}^2/k\\
    &\gsim \int_{[0,1]}|\wh{x^*_C\cdot H}(f)|^2 df,
\end{align}
where we used the fact that $f^*$ is heavy and thus $\epsilon d\mathcal{N}^2/k \leq \epsilon \int_{[0,1]}|\wh{x^*_C\cdot H}(f)|^2 df$. 

Now we have that for $z^{(j)} =(x^*_C\cdot H)(t)+ g_C(t) $, $g_C$ satisfies, $$\sum_{t=-\infty}^{\infty}|g_C(t)|^2 = \int_{[0,1]}|\wh{g_C}(f)|^2df \lsim (\epsilon\delta/k )\int_{[0,1]}|\wh{x^*\cdot H}(f)|^2\leq \epsilon d\mathcal{N}^2/k.$$ Thus we have that, $$\sum_{t\in [-\infty,0]\cup [d,\infty]}|z^{(j)}(t)|^2  \lsim \sum_{t\in [-\infty,0]\cup [d,\infty]}|(x^*_C\cdot H)(t)|^2 + \epsilon d\mathcal{N}^2/k.$$
Now since $x^*_C$ is an at most $k$-Fourier sparse function, by the properties of the $H$ function as per Lemma \ref{lem:H_discrete_properties}, we have the following,
\begin{align*}
    \sum_{t\in [-\infty,0]\cup [d,\infty]}|(x^*_C\cdot H)(t)|^2 &\leq \epsilon \sum_{t\in [-\infty,\infty]}|(x^*_C\cdot H)(t)|^2\\
    &= \epsilon \int_{[0,1]}|\wh{x^*_C\cdot H}(f)|^2 df,
\end{align*}
and since $f^*$ is heavy, $\epsilon \int_{[0,1]}|\wh{x^*_C\cdot H}(f)|^2 df + \epsilon d \mathcal{N}^2/k\lsim \epsilon\int_{[0,1]}|\wh{x^*_C\cdot H}(f)|^2 df$. Thus we finally get the following,
$$\sum_{t\in [-\infty,0]\cup [d,\infty]}|z^{(j)}(t)|^2  \lsim \epsilon\int_{[0,1]}|\wh{x^*_C\cdot H}(f)|^2 df.$$ Combining this with equation \ref{eqn:z_j_time} and applying Parseval's theorem completes the proof of the lemma.
\end{proof}

With these tools in place, we are ready to prove Lemma \ref{lem:on_grid_lemma_7_6}.
\begin{proof}[Proof of Lemma \ref{lem:on_grid_lemma_7_6}]
Condition on Lemma \ref{lem:on_grid_7_19_well_isolated} being true. Consider any cluster $C$ with midpoint $f^*$ as per the Lemma satisfying the assumptions of the Lemma \ref{lem:on_grid_lemma_7_6}.	Let $j = h_{\sigma,b}(f^*)$ and recall $z^{(j)} = (x^*\cdot H) \ast G^{(j)}_{\sigma,b}$ and $x^{(j)} = ((x^*+g) \cdot H ) \ast G^{(j)}_{\sigma,b}$. Let $I_{f^*} = [f^*-200k\Delta,f^* + 200k\Delta]$ modulo 1. Then combining the fact that $f^*$ is heavy and it is well-isolated, we know that the following holds from guarantees of Lemma \ref{lem:on_grid_7_19_well_isolated},
	\begin{equation}\label{eqn:hashing_1}
		\int_{I_{f^*}} |\wh{z}^{(j)}(f)|^2 df \gsim \int_{f^*-\Delta}^{f^*+\Delta} |\wh{x^*\cdot H}(f)|^2 df \geq d\mathcal{N}^2/k,
	\end{equation}
and the following as well
\begin{equation}\label{eqn:hashing_2}
	\int_{[0,1]\setminus I_{f^*}} |\wh{z}^{(j)}(f)|^2 df \lsim \epsilon d\mathcal{N}^2/k.
\end{equation}
	Now recall $\wh {g}^{(j)} = \wh{g\cdot H} \cdot \wh{G}^{(j)}_{\sigma,b}$.  
Then the previous discussion implies the following,
\begin{align}
	\int_{I_{f^*}}|\wh{x}^{(j)}(f)|^2 df &\gsim \int_{I_{f^*}}|\wh{z}^{(j)}(f)|^2df -\int_{I_{f^*}}|\wh{g}^{(j)}(f)|^2df\quad \text{(Cauchy-Schwarz)}\\
        &\gsim \int_{I_{f^*}}|\wh{z}^{(j)}(f)|^2df \quad \text{(assumption 2 of Lemma \ref{lem:on_grid_lemma_7_6} plus eqn. \ref{eqn:hashing_1})}\label{eqn:helper_eqn}\\
        &\gsim \int_{f^*-\Delta}^{f^*+\Delta}|\wh{x^*\cdot H}(f)|^2 df \geq d\mathcal{N}^2/k\quad \text{(equation \ref{eqn:hashing_1})}.
\end{align}
Furthermore,
\begin{align*}
	\int_{[0,1]\setminus I_{f^*}}|\wh{x}^{(j)}(f)|^2 df &\leq 2\left(\int_{[0,1]\setminus I_{f^*}}|\wh{z}^{(j)}(f)|^2 + |\wh{g}^{(j)}(f)|^2 df\right) \quad \text{(Cauchy-Schwarz)}\\
	& \lsim \epsilon d\mathcal{N}^2/k + \epsilon\int_{ I_{f^*}} |\wh{z}^{(j)}(f)|^2 df\quad \text{(assumption 2 of Lemma \ref{lem:on_grid_lemma_7_6} and eqns. \ref{eqn:hashing_1} and \ref{eqn:hashing_2})}\\
	&\lsim \epsilon \int_{I_{f^*}}|\wh{x}^{(j)}(f)|^2 df \quad \text{(from eqn. \ref{eqn:helper_eqn} and $f^*$ being heavy)}.
\end{align*}
Combining equation \ref{eqn:hashing_1} with the bound on $\int_{[0,1]\setminus I_{f^*}}|\wh{x}^{(j)}(f)|^2 df$ above, we get the following,
\begin{equation}\label{freq_concentrated}
	\int_{I_{f^*}} |\wh{x}^{(j)}(f)|^2 df \geq (1-\epsilon) \int_{[0,1]} |\wh{x}^{(j)}(f)|^2 df.
\end{equation}

Next, our goal is to show that the $x^{(j)}$'s energy in time domain is concentrated in $[0,d]$. By Plancherel's theorem, we get the following,
\begin{equation}\label{eqn:low_noise_time}
\sum_{t=0}^{d}|g^{(j)}(t)|^2 \leq \sum_{t=-\infty}^{\infty}|g^{(j)}(t)|^2  = \int_{[0,1]}|\wh{g}^{(j)}(f)|^2 df\lsim \epsilon  \int_{[0,1]}|\wh{z}^{(j)}(f)|^2 df  =  \epsilon \sum_{t=-\infty}^{\infty}|z^{(j)}(t)|^2,
\end{equation}
where the second last inequality used assumption 2 of Lemma \ref{lem:on_grid_lemma_7_6}. Combining this with Lemma \ref{lem:concentrated_in_time}, we have the following, $$\sum_{t=0}^d|g^{(j)}(t)|^2 \leq \sum_{t=-\infty}^{\infty}|g^{(j)}(t)|^2 \leq \epsilon \sum_{t=0}^{d}|z^{(j)}(t)|^2.$$ 
Equipped with this inequality, we can bound the energy of $x^{(j)}$ in time domain over $[d]$ as follows,
\begin{align*}
	\sum_{t=0}^{d}|x^{(j)}(t)|^2 &\geq \sum_{t=0}^{d} |z^{(j)}(t)|^2 - \sum_{t=0}^{d} |g^{(j)}(t)|^2 - 2\sqrt{\sum_{t=0}^{d}|z^{(j)}(t)||g^{(j)}(t)|}\\
	&\geq \sum_{t=0}^{d} |z^{(j)}(t)|^2 - \sum_{t=0}^{d} |g^{(j)}(t)|^2 - 2\sqrt{(\sum_{t=0}^{d}|z^{(j)}(t)|^2)(\sum_{t=0}^d|g^{(j)}(t)|^2)}\\
	& \geq (1-2\sqrt{\epsilon})\sum_{t=0}^d |z^{(j)}(t)|^2.
\end{align*}
Furthermore combining equation \ref{eqn:low_noise_time} and using Lemma \ref{lem:concentrated_in_time}, we can bound the energy of $x^{(j)}$ outside $[d]$ as follows,
\begin{align*}
	\sum_{t\in [-\infty,0]\cup [d,\infty]}|x^{(j)}(t)|^2 &\lsim \sum_{t=-[\infty,\infty]}|g^{(j)}(t)|^2+	\sum_{t\in [-\infty,0]\cup [d,\infty]}|z^{(j)}(t)|^2 \\
	&\lsim \epsilon \sum_{t=0}^d |z^{(j)}(t)|^2.
\end{align*}
The above two inequalities imply the following,
\begin{equation}\label{time_concentrated}
	\sum_{t=0}^d |x^{(j)}(t)|^2 \geq (1-2\sqrt{\epsilon})\sum_{t=-\infty}^{\infty} |x^{(j)}(t)|^2.
\end{equation}
Finally we know that $f^*\in \mathcal{I}$. Thus, equations \ref{time_concentrated} and \ref{freq_concentrated} imply $x^{(j)}$ is an $(2\sqrt{\epsilon},200k\Delta)$-clustered signal around $f^*$ and thus completing the proof of Lemma \ref{lem:on_grid_lemma_7_6}. Since we conditioned on Lemma \ref{lem:on_grid_7_19_well_isolated}, this holds with probability $1-k_C^2/B-k_C/99k$ simultaneously for all clusters $C$ of $\wh{x^*\cdot H}$.
\end{proof}
\subsubsection{Frequency recovery of bounded instances.}\label{subsection:recovering_all_freqs_by_boosting}
Lemmas \ref{lem:on_grid_lemma_7_6} and \ref{lem:on_grid_lemma_7_17} imply the following lemma which shows how to recover all ``recoverable'' frequencies of a bounded instance with high probability. %\hltodo{Should we write the union bound out explicitly in the final version? Reader has to backtrack a lot to the one-cluster lemma. (Addressed)}
\begin{lemma}\label{lem:const_prob_bounded_recovery}
Let $x(t) = (x^* (t) +  g(t))\cdot H(t)$ be an $(\mathcal{I},\delta')$-bounded instance with noise threshold $\mathcal{N}^2$ as per the setup of Definition \ref{def:bounded_freq_case}. Let $G(t)$ for all $j\in [B]$ be as per Def. \ref{def:permutation} and Lemma \ref{lem:G_discrete_properties} for parameters $B=\Theta(k^2),\epsilon\delta/k$ and $w=0.01,k$. Let $\Delta$ be as per Lemma \ref{lem:dtft_freqrecovery}. Apply the clustering procedure of Definition \ref{def:clustering} to $x^*\cdot H$ and let $k_C\leq k$ be the total number of clusters. Let $\sigma$ be a u.a.r. integer in $[\frac{1}{200Bk\Delta},\frac{1}{100Bk\Delta}]$ and $b$ be a u.a.r. real number in $[0,\frac{1}{\sigma}]$. Then one can find a list $L$ of $B$ frequencies in $[0,1]$ in time $\poly(k,\log(d/\delta))$ such that with probability at least $1-k_C^2/B-k_C/99k-o(1/k)$ the following holds - 
For all clusters $C$ with $f^*$ being the midpoint of $C$ and $j=h_{\sigma,b}(f)$, $f^*$ is well-isolated (Def. \ref{def:well_isolated_f}). Moreover if the following is satisfied,
\begin{enumerate}
    \item $f^*$ is heavy as per Definition \ref{def:heavy_freq} and,
    \item $\int_{[0,1]} |\wh{g}^{(j)}(f)|^2df\lsim \epsilon \int_{f^*-\Delta}^{f^*+\Delta} |\wh{x^*\cdot H}(f)|^2df$ ($g^{(j)}$ as per Def. \ref{def:convolved_signal}),
\end{enumerate}
 then there exists a $f\in L$ such that $|f^* - f|_{\circ} \lsim k\Delta \sqrt{k\Delta d}$ .
\end{lemma}
\begin{proof}
Condition on Lemma \ref{lem:on_grid_lemma_7_6} holding true. Our procedure to find $L$ is just to run \hyperref[alg:freq_rec_1_cluster]{\freqrecovcluster{}} on each of the $B$ instances corresponding to each hash bucket and thus we get a list of $B$ frequencies, one frequency per each hash bucket. Consider any cluster $C$ of $x^*\cdot H$ with midpoint $f^*$ satisfying the conditions of the Lemma. Let $j=h_{\sigma,b}(f^*)$, then Lemma \ref{lem:on_grid_lemma_7_6} implies that $x^{(j)} = (x^*+g)\cdot H \ast G^{(j)}_{\sigma,b}$ is $(\sqrt{\epsilon},200k\Delta)$-one clustered around $f^*$. We then apply Lemma \ref{lem:on_grid_lemma_7_17} to the $j^{th}$ bucket, that is on $x^{(j)}$. We do this by returning the $j^{th}$ element of the output of \hyperref[alg:hashtobins]{\hashtobins}{$(x,H,B,\sigma,\alpha/\sigma,b,\delta,w)$} to implement time domain access $x^{(j)}(\alpha)$ for any $\alpha\in [d]$ as needed by the algorithm of Lemma \ref{lem:on_grid_lemma_7_17}. This implies that we can recover a $f$ such that $|f-f^*|_{\circ}\lsim k\Delta  \sqrt{k\Delta d}$ with probability $1-2^{-\Omega(k)}$. We take a union bound for Lemma \ref{lem:on_grid_lemma_7_17} to succeed for all $j\in [B]$. This happens with probability $1-B\cdot2^{-\Omega(k)}=1-o(1/k)$. Then we union bound this with Lemma \ref{lem:on_grid_lemma_7_6} succeeding. This completes the proof of the lemma.
\end{proof}
\subsubsection{Reducing a general instance to a bounded instance.}\label{subsubsection:multi_cluster_reduce}
Consider the setup and parameters of Lemma \ref{lem:dtft_freqrecovery}, that is a general signal $x(t) = x^*(t) + g(t)$, where $x^*(t) = \sum_{f\in S} a_f e^{2\pi i f t}$ where $S\subseteq  [0,1]$, $|S|\leq k$ and $g(t)$ is arbitrary. 
%\begin{remark}
%We further assume that $supp(\wh{x^*\cdot H})\cap supp(\wh{g\cdot H})=\emptyset$, if not then we can absorb the noise into the signal. This will just lead to an additional loss of $\|g\|_d^2$ in the final estimation error.
%\end{remark}
%, and $\|g\|_d^2 \leq c \|x^*\|_d^2$ for some small enough \hltodo{small enough with respect to what problem parameters?} constant $c>0$. 
In this section we explain the reduction from such a general instance to $B=\Theta(k^2)$ bounded instances.  To achieve this, we convolve with the filter function $G^{(j)}_{\sigma,b}(t)$ as per Definition \ref{def:permutation} with $\sigma=1$ and $b$ is real number chosen uniformly at random between $[0,1]$. The parameters used in the $G$ filter function as per Definition \ref{lem:G_discrete_properties} to construct $G^{(j)}_{1,b}(t)$ is $B$ and $w  = 1/\poly(k)$ and $\delta$.

Before stating and proving the main result of this section, We first make a few important observations about this filtering operation. The claim below formalizes the behavior of the hash function $h_{1,b}$, and its proof trivially follows from Definition \ref{def:permutation}.

\begin{claim}
The hash function $h_{\sigma,b}$ as per Definition \ref{def:permutation} for $\sigma=1$ and $b\in [0,1]$ partitions $[0,1]$ into $B$ intervals/buckets $[b-1/2B,b+1/2B],[b+1/2B,b+3/2B],\ldots,[b-3/2B,b-1/2B]$ where each interval is modulo 1. Convolving the input with $G^{(j)}_{1,b}$ gives access to the input only containing frequencies in the $j^{th}$ such interval/bucket.
\end{claim}

%\begin{claim}[Remark C.15 of \cite{chen2016fourier}]
	%Computing $(x\cdot H) \ast G^{(j)}_{1,b}(t)$ for some $t\in [0,d]$ using \textsc{HashToBins} needs the following  samples from $x$ - $x(1-t),x(2-t),\ldots, x(Bl/\alpha-t)$. 
%\end{claim}
We will use the \hyperref[alg:hashtobins]{\textsc{HashToBins}} primitive as per Lemma \ref{lem:hashtobins} to access the function containing frequencies restricted to a such a bucket/interval.
We now define the notion of a badly cut cluster needed to argue that all clusters land in the region where the filter function $G$ has value almost $1$.
\begin{definition}\label{def:badly_cut_cluster}
	Let $C$ be any cluster as per Definition \ref{def:clustering}. Suppose there exists a $j$ such that $C\subseteq I_j:=[b+j/2B,b+(j+2)/2B]$. Let $I_j^{inner}:=[b+j/2B + w/2B,b+(j+2)/2B-w/2B]$. We say that $C$ is ``\emph{badly cut}'' if $C\not\subseteq I_j^{inner}$.
\end{definition}

 Since the width of any cluster is $\poly(k,\log(d/\delta))/d=o(1/k^2)=o(1/B)$, we can show the following Lemma.
\begin{lemma}\label{lem:nothing_badly_cut}
With probability at least $0.99$ over the randomness in $b$, no cluster in $x^*\cdot H$ as per Definition \ref{def:clustering} is badly cut.
\end{lemma}
\begin{proof}
Since the width of any cluster is at most $\poly(k,\log(d/\delta))/d$ and the width of the good region $I^{inner}_j$ in any interval $I_j$ is at least $(1-w)/B = \Theta(1/k^2)$, the probability that a fixed cluster is badly cut is at most $\poly(k,\log(d/\delta))/d$. Since there are at most $k$ clusters, taking a union bound finishes the proof since $k^3\cdot \poly(k,\log(d/\delta))/d << 0.01$. 
\end{proof}
We choose a $b$ uniformly at random and condition on the event guaranteed to hold with probability $0.99$ as per the previous Lemma \ref{lem:nothing_badly_cut}.
We now state the main claim of this section below which shows that this event is enough to guarantee that the instance corresponding to each bin $j\in [B]$ is bounded. 
\begin{lemma}\label{lem:gen_to_bounded_reduction}
 Choose $b$ u.a.r. in $[0,1]$. Let $z^{(j)}(t) = ((x^*+g)\cdot H) \ast G^{(j)}_{1,b}(t)$ for $ G^{(j)}_{1,b}$ as per Def. \ref{def:permutation} and Lemma \ref{lem:G_discrete_properties} for parameters $B,w=1/\poly(k),\delta$. Let $C_j$ be the union of all clusters of $x^* \cdot H$(see Def. \ref{def:clustering} for clusters) in interval $I_j$ (see Def. \ref{def:badly_cut_cluster} for $I_j$) and $S_j=supp(\wh{x^*})\cap C_j$. Then with probability at least 0.9 over choice of $b$, for all $j\in [B]$ $z^{(j)}(t) = (x^{(j)}\cdot H)(t) + (g^{(j)}\cdot H)(t)$, where $x^{(j)}(t) = \sum_{f\in S_j} a_f e^{2\pi i f t}$, is a $(I_j, \delta \|x^*\|_d^2/k )$-bounded instance (as per Def. \ref{def:bounded_freq_case}). Furthermore for any interval $I\subseteq I_j$,
\begin{equation*}
    \int_{I}|\wh{g^{(j)}\cdot H}(f)|^2 df\lsim \int_{I}|\wh{g\cdot H}(f)|^2 df + (\delta/k) \|x^*\|_d^2.
\end{equation*}
\end{lemma}
Now we state the proof this Lemma.

\begin{proof}[Proof of Lemma \ref{lem:gen_to_bounded_reduction}]
Lemma \ref{lem:nothing_badly_cut}, implying $C_j\subseteq I^{inner}_{j}$, combined with the fact that  $|\wh{G}^{(j)}_{1,b}(f)|\in [1,1- \delta/k]$ for all $f\in I_j^{inner}$, $|\wh{G}^{(j)}_{1,b}(f)|\leq 1$ for all $f\in I_j$ from Lemma \ref{lem:G_discrete_properties} implies the following,
\begin{align*}
\int_{I}|\wh{g^{(j)}\cdot H}|^2 df &= \int_{I}|\wh{x^{(j)}\cdot H}(f) - \wh{x^*\cdot H}\cdot \wh{G}^{(j)}_{1,b}(f) - \wh{g\cdot H}\cdot \wh{G}^{(j)}_{1,b}(f)|^2 df\\ 
&\leq 2\int_{I}|\wh{g\cdot H}(f)|^2 df +  \frac{2 \delta}{k}\int_{I}|\wh{x^* \cdot H}(f)|^2 df\\
&\lsim \int_{I}|\wh{g\cdot H}(f)|^2 df +  \frac{\delta}{k}\int_{[0,1]}|\wh{x^* \cdot H}(f)|^2 df \lsim \int_{I}|\wh{g\cdot H}(f)|^2 df + \frac{\delta}{k} \|x^*\|_d^2.
\end{align*}
Furthermore $|\wh{G}^{(j)}_{1,b}(f)|\leq \delta/k$ for all $f\in [0,1]\setminus I_j$ from Lemma \ref{lem:G_discrete_properties} implies the following,
\begin{align*}
	\int_{[0,1]\setminus I_j}|\wh{g^{(j)}\cdot H}(f)|^2 df &=\int_{[0,1]\setminus I_{j}}|(\wh{x^*\cdot H}+\wh{g})\cdot\wh{G}^{(j)}_{1,b}(f)-0|^2 df\quad (\wh{x^{(j)}\cdot H}(f)=0 \forall f\in [0,1]\setminus I_j)\\
 &\leq \frac{\delta }{k} \left(\int_{0}^{1} |\wh{x^* \cdot H}(f)|^2 df+ \int_{0}^{1} |\wh{g \cdot H}(f)|^2 df\right)\\
	&\lsim \frac{\delta }{k} \|x^*\|_d^2,
\end{align*}
where in the last line we used the assumption of Lemma \ref{lem:dtft_freqrecovery} that $\|g\|_d^2\leq c\|x^*\|_d^2$ for some small enough constant $c>0$. This implies that $z^{(j)}(t) = (x^{(j)}\cdot H)(t) + g^{(j)}(t)$ is a $(I_j, \delta \|x^*\|_d^2/k)$-bounded instance (recall Def. \ref{def:bounded_freq_case}) because the previous equation implies that $\int_{[0,1]\setminus I_j}|\wh{g^{(j)}\cdot H}(f)|^2 df\lsim \frac{\delta }{k}\|x^*\|_d^2$ and $x^{(j)}$, an at most $k$-Fourier sparse function, has all its clusters in $I_j$ whose width is $1/B$. 
\end{proof}
Equipped with these reductions, we are now ready to finish the proof of Lemma \ref{lem:dtft_freqrecovery}. The proof essentially reduces a general instance to $B$ bounded instances and then applies the algorithm of Lemma \ref{lem:const_prob_bounded_recovery} to recover frequencies from each of these bounded instances. 
\begin{proof}[Proof of Lemma \ref{lem:dtft_freqrecovery}]
Consider the setup as described in Lemma \ref{lem:dtft_freqrecovery}. We apply Lemma \ref{lem:gen_to_bounded_reduction} to $x$. Then we know that with probability at least 0.9, $z^{(j)}(t) = (x^{(j)}\cdot H)(t) + g^{(j)}(t)$, where $x^{(j)}(t) = \sum_{f\in S_j} a_f e^{2\pi i f t}$ is an $(I_j, \delta \|x^{*}\|_d^2/k)$ bounded instance for all $j\in [B]$. Now consider the $j^{th}$ such bounded instance. Then the noise threshold in the $j^{th}$ bounded instance $x^{(j)}$ $\mathcal{N}_j^2$ as per Def. \ref{def:bounded_freq_case} satisfies,
	\begin{align}
		\mathcal{N}_j^2 &= \frac{1}{d}\int_{ I_{j}} |\wh{g}^{(j)}(f)|^2 df + \frac{\delta}{d\epsilon}\|x^{*}\|_d^2 + \frac{\delta}{d} \int_{[0,1]}|\wh{x^{(j)}\cdot H}(f)|^2 df\\
		&\lsim \frac{1}{d}\int_{ I_{j}} |\wh{g\cdot H}(f)|^2 df + \frac{\delta}{d}\|x^*\|_d^2 + \frac{\delta}{d} \int_{C_j}|\wh{x^*\cdot H}(f)|^2 df\\
		& \lsim \frac{1}{d}(\int_{[0,1]} |\wh{g}(f)|^2 df + \delta\|x^{*}\|_d^2)=\frac{1}{d}(\|g\|_d^2 + \delta \|x^*\|_d^2) = \mathcal{N}^2,\label{eqn:n_lessthan_nj}
	\end{align}
where in the second inequality we used the guarantee of Lemma \ref{lem:gen_to_bounded_reduction} to upper bound $\int_{I_j}|\wh{g}^{(j)}(f)|^2df$ and in third inequality we used the fact that $\int_{I_j}|\wh{g\cdot H}(f)|^2 df \leq \|g\cdot H\|_d^2 \leq \|g\|_d^2= \int_{[0,1]}|\wh{g}(f)|^2 df$ (noise $g(t)$ outside $[d]$ is $0$). Now we apply Lemma \ref{lem:const_prob_bounded_recovery} to $z^{(j)}$. Then we know that all clusters in $C_j$ (set of all clusters of $x^{(j)}$) with midpoint $f^*$ that do not satisfy both the conditions of Lemma \ref{lem:const_prob_bounded_recovery} have no guarantee of being recovered. Call $C_{unrec}$ the union of all such clusters. First the amount of energy corresponding to clusters $C_j$ with midpoints that are not heavy (as per Def. \ref{def:heavy_freq}) is at most the following,
\begin{equation}\label{eqn:light_clusters}
|C_j|d \mathcal{N}_j^2/k\leq d\mathcal{N}^2,
\end{equation}
since all clusters are disjoint as per \ref{def:clustering}. Now consider the random hashing $(\sigma,b)$ and corresponding filters $G^{(\cdot)}_{\sigma,b}$as per Lemma \ref{lem:const_prob_bounded_recovery}, then we know that all clusters $C\in C_j$ with midpoint $f^*$ and $i=h_{\sigma,b}(f^*)$ satisfying the following,
\begin{equation*}
    \int_{[f^*-\Delta,f^*+\Delta]}|\wh{x^{(j)}\cdot H}(f)|^2 df \lsim \int_{[0,1]}|\wh{g^{(j)}\cdot H}\cdot \wh{G}^{(i)}_{\sigma,b}(f)|^2 df,
\end{equation*}
also have no guarantee of being recovered. Since the clusters are disjoint and also well-isolated simultaneously from Lemma \ref{lem:const_prob_bounded_recovery}, thus mapping to different bins, the total amount of energy lost due to such low SNR clusters is at most the following,
\begin{align*}
    \sum_{i\in [B]}\int_{[0,1]}|\wh{g^{(j)}\cdot H}\cdot \wh{G}^{(i)}_{\sigma,b}(f)|^2 df &\lsim \int_{[0,1]}|\wh{g^{(j)}\cdot H}(f)|^2 (\sum_{i\in [B]}|\wh{G}^{(i)}_{\sigma,b}(f)|^2|) df \\
    &\lsim \int_{[0,1]}|\wh{g^{(j)}\cdot H}(f)|^2 df =  \int_{I_j}|\wh{g^{(j)}\cdot H}(f)|^2 df+ \int_{[0,1]\setminus I_j}|\wh{g^{(j)}\cdot H}(f)|^2 df\\
    &\lsim \int_{I_j}|\wh{g^{(j)}\cdot H}(f)|^2 df + \delta\|x^*\|_d^2/k \quad \text{(since $z^{(j)}$ is a bounded instance)}\\
    &\lsim \int_{I_j}|\wh{g\cdot H}(f)|^2 + \delta\|x^*\|_d^2/k \quad \text{(from guarantee of Lemma \ref{lem:gen_to_bounded_reduction})},
\end{align*}
where the first two inequalities follow from Def. \ref{def:clustering} and Properties I-III of Lemma \ref{lem:G_discrete_properties}. Summing this up over all $j\in [B]$, we get that the total energy of such low SNR clusters is at most the following,
\begin{equation}\label{eqn:low_snr_clusters}
\int_{[0,1]}|\wh{g\cdot H}(f)|^2 df + \delta \|x^*\|_d^2 = \|g\cdot H\|_d^2 +\delta \|x^*\|_d^2\leq \|g\|_d^2 + \delta \|x^*\|_d^2=d\mathcal{N}^2.   
\end{equation}
Combining equations \ref{eqn:low_snr_clusters} and \ref{eqn:light_clusters} we get that $\int_{C_{unrec}}|\wh{x^*\cdot H}(f)|^2 df\lsim  d\mathcal{N}^2$. Morever taking a union bound over all $j\in [B]$ for Lemma \ref{lem:const_prob_bounded_recovery} to succeed for all $z^{(j)}$, we get that this event happens with probability $1-(\sum_{j\in [B]}|C_j|^2/B+|C_j|/99k) \geq 1-O(k^2)/B-k/99k\geq 0.9$ for $B=O(k^2)$ (since $\sum_{j\in [B]}|C_j|^2\leq O((\sum_{j\in [B]}|C_j|)^2)=O(k^2)$). If we let $L$ to be the union of the outputs of Lemma \ref{lem:const_prob_bounded_recovery} on all $j\in [B]$, then from the claim of Lemma \ref{lem:const_prob_bounded_recovery} we get that $S_{heavy}$ as defined as per Lemma \ref{lem:dtft_freqrecovery} satisfies the following,
\begin{align*}
    \|x^*-x^*_{S_{heavy}}\|_d^2 &\lsim \int_{[0,1]}|\wh{x^*\cdot H}(f)-\wh{x^*_{S_{heavy}}\cdot H}(f)|^2 df\\
    &= \int_{C_{unrec}}|\wh{x^*\cdot H}(f)|^2 df\\
    &\lsim d\mathcal{N}^2.
\end{align*}
We can implement sample access to $z^{(j)}(t)$ for any integer $t$ that the algorithm of Lemma \ref{lem:const_prob_bounded_recovery} demands by returning $u[j]$ for $u=$ \hyperref[alg:hashtobins]{\hashtobins}$(x,H,G,1,t,b,w)$. A remark here is that whenever \hyperref[alg:hashtobins]{\hashtobins} requires to access the input outside $[d]$ we output $0$, from property 5 of Lemma \ref{lem:H_discrete_properties} this only increases the noise by a additive $\delta \|x^*\|_d^2$ factor. The correctness of this follows from Lemma \ref{lem:hashtobins} and thus the total running time and sample complexity suffers a multiplicative overhead of $O(B\log(k/\delta)/w)$ on top of running the algorithm of Lemma \ref{lem:const_prob_bounded_recovery}. Thus the overall running time and sample complexity is still $\poly(k,\log(d/\delta))$.
\end{proof}

\section{Sublinear Time Algorithms for Toeplitz Matrices}\label{sec:sublinear_toeplitz_main_sec}
In this section, we prove our main sublinear time robust Toeplitz matrix approximation result (Theorem \ref{thm:sublinear_time_recovery}) and describe its applications to sublinear time Toeplitz low-rank approximation and covariance estimation.
In Section \ref{subsection:heavy_light_decomp} we present the proof of a heavy-light decomposition result %\Cam{'the heavy-light decomposition result' sounds like this is something the the reader should be familiar with. But they won't be.} 
using the off-grid frequency recovery algorithm of Lemma \ref{lem:dtft_freqrecovery} and the existence result of Theorem \ref{frobeniusexistence}, as briefly discussed in Section \ref{subsec:intro_matrix_recovery}. Next, in Section \ref{subsection:noisy_toeplitz_recovery} we use this heavy-light decomposition result and sublinear time approximate regression techniques for Fourier sparse functions to  prove Theorem \ref{thm:sublinear_time_recovery}. Finally, in  Section \ref{subsection:toeplitz_applications} we use Theorem \ref{thm:sublinear_time_recovery} to prove our sublinear time low-rank approximation (Theorem \ref{thm:sublinear_time_lowrankapprox}) and covariance estimation (Theorem \ref{thm:sublinear_time_covariance_estimation}) results.

\subsection{Heavy Light Decomposition.}\label{subsection:heavy_light_decomp}
In this subsection, we present the proof of our heavy-light decomposition as discussed in Section \ref{subsec:intro_matrix_recovery}. The formal statement is as follows.
\begin{lemma}\label{lem:heavy_freq_recovery}
Consider the input setup of Theorem \ref{thm:sublinear_time_recovery}. Let $\wt T = F_{S}  D F_{S}^*$ be as guaranteed to exist by Theorem \ref{frobeniusexistence} for $T$, $\epsilon=0.1$, $\delta$ and $k$.  Let $E^k= E+T-\wt T$. Assume $\|E^k\|_F \leq c\|\wt T\|_F$ for some small enough constant $c>0$. Then in $\poly(k,\log(d/\delta))$ time we can find a list of frequencies $L \subseteq [0,1]$ of size $|L| = \poly(k,\log (d/\delta))$ satisfying the following with probability 0.9 - Let $S_{heavy}\subseteq \wt S$ defined as follows, $$S_{heavy}= \{f\in  S: \exists f'\in L \text{ s.t. } |f-f'|_{\circ}\leq \poly(k,\log(d/\delta))/d \}.$$ Furthermore for every $f\in S_{heavy}$, if $(1-f) \mod 1\notin S_{heavy}$ then add it to $S_{heavy}$. Then there exists a diagonal $D^{heavy}$ and $\wt T^{heavy} = F_{S^{heavy}} D^{heavy}F_{S^{heavy}}^{*}$ such that $\wt T^{light} = \wt T - \wt T^{heavy}$ satisfies  $$\|\wt T^{light}\|_F \lsim \|E^k\|_F + \delta \|T\|_F.$$
  \end{lemma}

To prove this, we first need the following helper claim for PSD matrices which will also be useful at a later stage in the paper. 
\begin{claim}[Equation (5) in \cite{audenaert2006norm}]\label{lem:diagonal_dominance_psd}
The following holds for any PSD matrix $A\in \mathbb{R}^{d\times d}$ - 
\begin{equation*}
	\|A\|_F \leq \|A_{[0:d/2,0:d/2]}\|_F+\|A_{[d/2:d,d/2:d]}\|_F.
\end{equation*}
\end{claim}
Using the above claim, we now state our main helper lemma about the structural properties of symmetric Toeplitz matrices that are nearly PSD. The following lemma essentially says that if a nearly PSD Toeplitz matrix has a large Frobenius norm, then the first half of the first column must have a large $\ell_2$ norm as well.
\begin{lemma}\label{lem:smallnoisefirstcol}
	Let $\wt T$ be a $d\times d$ symmetric Toeplitz matrix, and suppose that there exists a PSD Toeplitz matrix $T$ satisfying $\|\wt T- T\|_F \leq 0.001 \|T\|_F$. 
    Then the following holds,
	\begin{equation*}
		\|\wt T_{[0:d/2,0]}\|_2\geq \frac{0.49}{\sqrt{d}} \|\wt T\|_F.
	\end{equation*}

% \Cam{The notation $O(c)$ where $c$ is a constant doesn't make sense. Also I don't think this Lemma makes sense as written. Also the Lemma I think is confusing as written. I think the Lemma should just prove that  $T_{[0:d/2,0]}\|_2$ is large compared to $\norm{tilde T}$. The part about picking a rnadom column of $E$ is orthogonal sine we make no assumptions about $E$ except a bound on its norm. Also the reader needs context. It is very mysterious why we even want to prove something about $E_{[i:i+d/2,i]}$. Or why we are looking at that specific chunk of the column.}

\end{lemma}
\begin{proof}
	Since $T$ is PSD, using Lemma \ref{lem:diagonal_dominance_psd} we know that 
	\begin{equation}
		\|T\|_F \leq \|T_{[0:d/2,0:d/2]}\|_F+ \|T_{[d/2:d,d/2:d]}\|_F.
	\end{equation}
	Since $T$ is also Toeplitz, $T_{[0:d/2,0:d/2]} = T_{[d/2:d,d/2:d]}$. Thus we have $\|T_{[0:d/2,0:d/2]}\|_F \geq 0.5 \|T\|_F$. Since $\|\wt T_{[0:d/2,0:d/2]}-T_{[0:d/2,0:d/2]}\|_F\leq \|\wt T - T\|_F\leq 0.001 \|T\|_F\leq 0.01 \|\wt T\|_F$, we get that $\|\wt T_{[0:d/2,0:d/2]}\|_F\geq 0.49 \|\wt T\|_F$. Finally we have the following,
 \begin{equation}\label{eqn:toeplitz_norm}
 \begin{aligned}
     \|\wt T_{[0:d/2,0:d/2]}\|_F^2 &=(d/2) T_{[0,0]}^2+\sum_{i\in [1,d/2]} (d-2i)\wt T_{[i,0]}^2 \quad \text{(since $\wt T_{[0:d/2,0:d/2]}$ is symmetric Toeplitz)}\\
     &\leq \sum_{i\in [d/2]} d \wt T_{[i,0]}^2= d \|\wt T_{[0:d/2,0]}\|_2^2.
 \end{aligned}
 \end{equation}
 Thus we get that $\sqrt{d}\|\wt T_{[0:d/2,0]}\|_2 \geq \|\wt T_{[0:d/2,0:d/2]}\|\geq 0.49 \|\wt T\|_F$.  
\end{proof}
\begin{comment}
The second helper lemma below essentially shows the converse, that is if the $\ell_2$ norm of the first half of the first column of an (almost) PSD Toeplitz matrix is small, then 
\begin{lemma}\label{lem:bounding_t_light}
	Let $\wt T$ be a $d\times d$ symmetric Toeplitz matrix, and suppose that there exists a PSD Toeplitz matrix $T$ satisfying $\|\wt T- T\|_F \leq \delta \|T\|_F$ for some $\delta>0$. If $\|\wt T_{[0:d/2,0]}\|_2 \leq M$ for some $M$ then $\|\wt T\|_F \leq 3(\sqrt{d}M + \delta\|T\|_F)$.
	
\end{lemma}
\begin{proof}
	Since $\wt T$ is symmetric Toeplitz, we have the following from equation \ref{eqn:toeplitz_norm}.
 \begin{align*}
     \|\wt T_{[0:d/2,0:d/2]}\|_F\leq \sqrt{d}\|\wt T_{[0:d/2,0]}\|_2
 \end{align*}
 Thus we have that $\|\wt T_{[0:d/2,0:d/2]}\|_F \leq \sqrt{d}M$.
Now we have that $$\|\wt T_{[0:d/2,0:d/2]}-  T_{[0:d/2,0:d/2]}\|_F \leq \|\wt T-T\|_F\leq \delta\|T\|_F.$$ Thus by triangle inequality $$\|\wt T_{[0:d/2,0:d/2]}\|_F\geq \|T_{[0:d/2,0:d/2]}\|_F-\delta \|T\|_F.$$
	Finally, since $T$ is PSD and Toeplitz, we have the following from applying Claim \ref{lem:diagonal_dominance_psd} again,
	\begin{align*}
		\|T\|_F &\leq \|T_{[0:d/2,0:d/2]}\|_F+\|T_{[d/2:d,d/2:d]}\|_F\\ &= 2\|T_{[0:d/2,0:d/2]}\|_F.
	\end{align*}
Thus we have the following chain of inequalities.
\begin{align*}
	\|\wt T\|_F &\leq \|T\|_F + \delta \|T\|_F\\
	&\leq 2\|T_{[0:d/2,0:d/2]}\|_F+ \delta \|T\|_F\\
	&\leq 2\|\wt T_{[0:d/2,0:d/2]}\|_F+ 3\delta \|T\|_F\\
	&\leq 3(\sqrt{dM} + \delta\|T\|_F).
\end{align*}
\end{proof}
\end{comment}
Equipped with this helper lemma and our sublinear time off-grid recovery result of Lemma \ref{lem:dtft_freqrecovery}, we are ready to present the proof of Lemma \ref{lem:heavy_freq_recovery}. 
%\Cam{I confused about the organization of this section. Are the previous results going to be used to proof Lemma 5.1? This should be made clear explicitly}
\begin{proof}[Proof of Lemma \ref{lem:heavy_freq_recovery}]
For $\wt T = F_{ S}  D F_{ S}^*$ suppose that $ S = \{f_1,f_2,\ldots,f_{\wt O(k)}\}$ and $ D = diag([v_1,\ldots , v_{\wt O(k)}])$. Let $\wt T_1(t)=\sum_{j=1}^{\wt O(k)}v_j e^{2\pi i f_j t}$, a $\wt O(k)$ Fourier sparse function. Then it is easy to see from expanding out $\wt T = F_{ S} D F_{ S}^*$ that the first column of $\wt T$ is defined by $\wt T_1(t)$ for $t\in \{0,\ldots, d-1\}$. A minor technicality compared to the description in the tech-overview is that rather than working with a random column $\wt T_{[0:d,j]}$ for $j\sim \{0,\ldots,d/2\}$, we will work with the $d/2$ sized chunk of it $\wt T_{[j:j+d/2,j]}$ which by the virtue of $\wt T$ being Toeplitz is equal to the first half of the first column $\wt T_{[0:d/2,0]}$. Thus this $d/2$ sized chunk of the $j^{th}$ column has identical Fourier spectrum compared to the first column. 
 
  First observe that $\mathbb{E}_{i\sim [d/2]}[\|E^{k}_{[i,i+d/2,i]}\|_2^2]\leq \|E^k\|_F^2$, thus applying Markov's inequality the following holds with probability $0.99$ for an $i\sim [d/2]$,
 \begin{equation}\label{eqn:small_noise_2}
 	\|E^{k}_{[i:i+d/2,i]}\|_2^2\leq (100/d)\|E^k\|_F^2.
 \end{equation}
We apply Lemma \ref{lem:smallnoisefirstcol} to $\wt T$, this is possible because point 3 of Theorem \ref{frobeniusexistence} in Section \ref{sec:preliminaries} implies that there exists some PSD Toeplitz matrix $T'$ such that $\|\wt T -T'\|_F\leq \delta\|\wt T\|_F$.
We thus get the following by combining equation \ref{eqn:small_noise_2}$, \|E^k\|_F\leq c\|\wt T\|_F$ with Lemma \ref{lem:smallnoisefirstcol},
\begin{equation}\label{eqn:small_noise_1}
\begin{aligned}
\|E^{k}_{[i,i+d/2,i]}\|_{2} &\leq (10/\sqrt{d})\|E^k\|_F \\
&\leq (10c/\sqrt{d})\|\wt T\|_F\\
&\lsim c\|\wt T_{[0:d/2,0]}\|_{2}.
\end{aligned}
 \end{equation}
Fix this $i$ and condition on this event that equations \ref{eqn:small_noise_1} and \ref{eqn:small_noise_2} hold. Let $x^*(t) = \wt T_1 (t)$, $g(t) = E^{k}_{i}(t+i):=E^{k}_{[i+t,i]}$ and $x(t)=x^*(t)+g(t)$ for $t\in [d/2]$. We can access $x(t)$ by querying the input $T+E$ at index $[i+t,i]$ for any $t\in [d/2]$. Let $H(t)$ be the function as per Lemma \ref{lem:H_discrete_properties} for parameters $\wt{O}(k),\delta$.
Apply Lemma \ref{lem:dtft_freqrecovery} to $x=x^*+g$ to obtain a list of frequencies $L$ of size $\poly(k)$ and let $S_{heavy} \subseteq  S$ be the set of all $f\in  S = supp(\wh{x^*})$ such that there exists some $f'\in L$ satisfying $|f-f'|\lsim  \Delta k \sqrt{\Delta k d} = \poly(k,\log(d/\delta))/d$. Let $S^{light} =  S \setminus S^{heavy}$. Observe that since the width of each set $\wt S_i $ as per point 2 in Theorem \ref{frobeniusexistence} in Section \ref{sec:preliminaries} is $\wt O(\gamma) = \wt O(1/2^{\poly \log d)}) = o(\Delta)$ ($\Delta$ is as per Lemma \ref{lem:dtft_freqrecovery}), each $\wt S_i$ is either completely in $S^{heavy}$ or completely in $\wt S \setminus S^{heavy}$. Let $D^{heavy}$ contain the diagonal entries of $ D$ corresponding to $S^{heavy}$, $\wt T^{heavy} = F_{S^{heavy}} D^{heavy} F_{S^{heavy}}^*$ and $\wt T^{heavy}_1$ be its first column.  Then we have the following,
\begin{align}
\|\wt T_1 (t) - \wt T^{heavy}_1(t)\|_{d/2}^2 &\lsim  \|g\|_{d/2}^2 + \delta \|\wt T_1\|_{d/2}^2 \\
	& = \|E^{k}_{[i:i+d/2,i]}\|_{2}^2 + \delta \|\wt T_1\|_{d/2}^2\quad \text{(Definition of $g=E^k_{[i:i+d/2,i]}$)}\\
	& \leq \|E^{k}_{[i:i+d/2,i]}\|_{2}^2 + \delta \|T\|_F^2 \label{heavylightbound1},
\end{align}
 where the first inequality follows from the guarantee of Lemma \ref{lem:dtft_freqrecovery}. We now state an important caveat below.
 \begin{remark}\label{rem:make_real}
   It may happen that $\wt T^{heavy}$ has complex entries, this can happen when there is some $f\in S^{heavy}$ such that $1-f\notin S^{heavy}$. However, discarding the imaginary part of entries in $\wt T^{heavy}$ can only lead to reducing $\|\wt T_1 (t) - \wt T^{heavy}_1(t)\|_{d/2}$, thus the bound of equation \ref{heavylightbound1} still holds. The removal of imaginary parts can be achieved by adding $1-f$ to $S^{heavy}$ for all such $f\in S^{heavy}$, and the coefficients of $f,1-f$ in $D^{heavy}$ will be equal to half of the corresponding coefficients in $D$, thus they will still be equal by point 3 of Theorem \ref{frobeniusexistence} in Section \ref{sec:preliminaries}.  
 \end{remark}
 Now define the symmetric Toeplitz matrix $\wt T^{light} = \wt T - \wt T^{heavy}$, and its first column $\wt T^{light}_1(t)=\wt T_1(t)-\wt T^{heavy}(t)$. Recall we know that every $\wt S_i$ (defined as per point 2 of Theorem \ref{frobeniusexistence} in Section \ref{sec:preliminaries}) is either completely in $S^{heavy}$ or completely out of it, and for every $\wt S_i\in S^{heavy}$, $-\wt S_i\in S^{heavy}$ ($-\wt S_i$ defined in point 3 of Thm. \ref{frobeniusexistence}). Let $S'$ be set of all $S_i=-\wt S_i\cup \wt S_i$ for $\wt S_i\in S^{heavy}$ such that $-\wt S_i$ was not in $S^{heavy}$, but we added it to make $\wt T^{heavy}$ real in Remark \ref{rem:make_real}. This implies $\wt T^{light}=\sum_{S_i\in S\setminus S^{heavy}}F_{S_i}D_i F_{S_i}^* + (1/2)\sum_{S_i\in S'}F_{S_i}D_iF_{S_i}^*$ where each $S_i,D_i$ as per point 3 of Theorem \ref{frobeniusexistence} in Section \ref{sec:preliminaries}. Let $\wt T^a = \sum_{S_i\in S\setminus S^{heavy}}F_{S_i}D_i F_{S_i}^*$ and $\wt T^b = \sum_{S_i\in S'}F_{S_i}D_iF_{S_i}^*$. Then by point 3 of Theorem \ref{frobeniusexistence} in Section \ref{sec:preliminaries} we can say that there exists PSD Toeplitz matrices $T^a,T^b$ such that,
 \begin{align*}
     \|\wt T^a - T^a\|_F&\leq \delta \|T^a\|_F \text{ and,}\\
      \|\wt T^b - T^b\|_F&\leq \delta \|T^b\|_F.
 \end{align*}
 Let $T^{light}=T^a+(1/2)T^b$, thus $T^{light}$ is also PSD Toeplitz. Thus $T^{light}$ satisfies the following,
\begin{equation}\label{heavylightbound2}
	\|\wt T^{light} - T^{light}\|_F \leq \delta (\|T^a\|_F+\|T^b\|_F)\leq O(\delta) \|T^{light}\|_F\ll 0.001 \|T^{light}\|_F.
\end{equation}
 On the other hand equation \ref{heavylightbound1} implies the following for $\wt T^{light}$,
\begin{equation}\label{heavylightbound_matrixversion}
    \|\wt T^{light}_{[0:d/2,0]}\|_2^2 = \|\wt T_1(t) - \wt T^{heavy}(t)\|_{d/2}^2 \lsim \|E^{k}_{[i:i+d/2,i]}\|_{2}^2 + \delta \|T\|_F^2.
\end{equation}
 
Thus equations (\ref{heavylightbound_matrixversion}) and (\ref{heavylightbound2}) allow us to apply Lemma \ref{lem:smallnoisefirstcol} to upper bound $\|\wt T^{light}\|_F\lsim \sqrt{d}\|\wt T^{light}_{[0:d/2,0]}\|_2$ to get the following,
\begin{align*}
	\|\wt{T}^{light}\|_F^2 &\lsim d\|E^{k}_{[i:i+d/2,i]}\|_F^2 + \delta d \|T\|_F^2\\
	\implies \|\wt T^{light}\|_F &\lsim \|E^{k}\|_F + \delta d\|T\|_F,
\end{align*}
where in the last line we used the equation \ref{eqn:small_noise_2}. Adjusting $\delta$ by $1/d$ factor (this is feasible by losing log factors as the dependence on $\delta$ is $\log(1/\delta)$), we finish the proof of point 2 of the Lemma.
\end{proof}
\subsection{Noisy Toeplitz Recovery.}\label{subsection:noisy_toeplitz_recovery}
Equipped with Lemma \ref{lem:heavy_freq_recovery} and Theorem \ref{frobeniusexistence}, in this subsection, we present the proof of Theorem \ref{thm:sublinear_time_recovery} which is our main sublinear time Toeplitz matrix approximation result.
\begin{proof}[Proof of Theorem \ref{thm:sublinear_time_recovery}]
Consider $\wt T^{heavy}$,$L$ and $\wt T^{light}$ as per the statement of Lemma \ref{lem:heavy_freq_recovery}. Lemmas \ref{lem:heavy_freq_recovery} and \ref{frobeniusexistence} imply the following for $\wt T^{heavy} = F_{S^{heavy}}D^{heavy}F_{S^{heavy}}^*$,
\begin{align*}
	\|T+E-\wt T^{heavy}\|_F &\leq \|E\|_F+\|T-\wt T\|_F + \|\wt T- \wt T^{heavy}\|_F\\
	& \lsim \|E\|_F + \|T-T_k\|_F + \delta \|T\|_F + \|\wt T^{light}\|_F\\
	& \lsim \|E\|_F + \|T-T_k\|_F + \delta \|T\|_F.
\end{align*}
Let $N = \{1/2d,3/2d,\ldots,1-1/2d\}$. Apply Lemma \ref{lem:heavy_freq_recovery} and let $L$ be the set of $\poly(k,\log d)$ frequencies returned by Lemma \ref{lem:heavy_freq_recovery} and $L' =\{ f\in N:  \exists f'\in L \text{ s.t. } |f'-f|\leq \poly(k,\log(d/\delta))/d\}$. This implies $|L'|\leq \poly(k,\log(d/\delta))$. Let $S(L') = \bigcup_{f \in L'} \bigcup_{1 \le j \le r_2} \{ f + \gamma j, f - \gamma j\}$ where $\gamma,r_2$ are as per Theorem \ref{frobeniusexistence}. Thus $|S(L')|\leq \poly(k,\log(d/\delta))$ This implies $ S^{heavy}$ as per Theorem \ref{frobeniusexistence} satisfies $S^{heavy}\subseteq S(L')$. Now we will solve the following regression problem approximately,
\begin{equation}\label{eqn:approximate_regression}
	\min_{D:D \text{ is diagonal }}\|T+E- F_{S(L')}DF_{S(L')}^*\|_F.
\end{equation}
We will first show that the optimal solution of the above optimization problem satisfies the guarantees of Theorem \ref{thm:sublinear_time_recovery}, then we will show how to obtain a constant factor approximate solution to the above problem in $\poly(k,\log(d/\delta))$ time.

Let $E^k = E+ T-\wt T$ as defined in Lemma \ref{lem:heavy_freq_recovery}, and first consider the case when $\|E^k\|_F\geq c \|\wt T\|_F$ where $c$ is the constant as per Lemma \ref{lem:heavy_freq_recovery}. This is the case when the noise is noticeably larger compared to the true input, and thus the guarantees of Lemma \ref{lem:heavy_freq_recovery} are not guaranteed to hold. In this case, returning $0$ as the solution of the regression problem won't be too bad. Formally we have the following,
\begin{align*}
	\min_{D:D \text{ is diagonal }}\|T+E- F_{S(L')}DF_{S(L')}^*\|_F &\leq \|T+E\|_F \quad \text{(for $D=0$)}\\
	&=\|T-\wt T + \wt T + E\|_F\\
	&\leq \|T-\wt T\|_F + \|\wt T\|_F + \|E\|_F\\
	&\lsim \|T-T_k\|_F + \delta \|T\|_F + \|\wt T\|_F + \|E\|_F\\
	&\lsim \|T-T_k\|_F + \delta \|T\|_F + \|E^k\|_F + \|E\|_F \\
	&\lsim \|E\|_F+\|T-T_k\|_F + \delta \|T\|_F.
\end{align*}
where we used Theorem \ref{frobeniusexistence} in the third last inequality, $\|E^k\|_F\geq c \|\wt T\|_F$ in the second last inequality and the definition of $E^k$ in the last inequality.

Now consider the other case when $\|E^k\|_F\leq c \|\wt T\|_F$. Then the requirements needed by Lemma \ref{lem:heavy_freq_recovery} hold. Now, from the structure of  $S$ as per points 2 and 3 of Lemma \ref{frobeniusexistence} we know that $S^{heavy} \subseteq S(L')$.  This implies the following,
\begin{align*}
	\min_{D:D \text{ is diagonal }}\|T+E- F_{S(L')}DF_{S(L')}^*\|_F &\leq \|T+E-\wt T^{heavy}\|_F \\
	&\leq \|E\|_F + \|T-\wt T\|_F + \|\wt T - \wt T^{heavy}\|_F\\
	&= \|E\|_F + \|T-\wt T\|_F + \|\wt T^{light}\|_F\\
	&\lsim \|E\|_F + \|T-T_k\|_F + \delta \|T\|_F,\\
\end{align*}
where in the last inequality we used point 2. of Lemma \ref{lem:heavy_freq_recovery}  to bound $\|\wt T^{light}\|_F$ and Theorem \ref{frobeniusexistence} to bound $\|T-\wt T\|_F$. Therefore in all cases we can conclude the following,
\begin{equation}\label{eqn:cost_of_opt}
	\min_{D:D \text{ is diagonal }}\|T+E- F_{S(L')}DF_{S(L')}^*\|_F \lsim  \|T-T_k\|_F + \|E\|_F+ \delta \|T\|_F.
\end{equation}
We will now apply the following lemma that is a corollary of Lemma 5.7 of \cite{eldar2020toeplitz} that allows us to find a $D$ that is a constant factor approximate solution to the regression problem of equation (\ref{eqn:approximate_regression}) in time only depending polynomially on $|S(L')|=\poly(k,\log(d/\delta))$.
\begin{lemma}[Corollary of Lemma 5.7 of \cite{eldar2020toeplitz}]\label{lem:weighted_regression}
There is an algorithm such that given any matrix $B\in \mathbb{R}^{d\times d}$ and set $M\subset [0,1]$ of size $|M|=m$, it runs in time at most $\poly(m)$ and returns a diagonal $D'\in \mathbb{R}^{m\times m}$ that satisfies the following with probability $0.99$,
\begin{equation*}
	\|B-F_{M}D' F_{M}^*\|_F \lsim \min_{\substack{D\in \mathbb{R}^{m\times m}:\\D \text{ is diagonal }}} \|B-F_{M}DF_{M}^*\|_F,
\end{equation*}
where $F_M$ is a Fourier matrix as per Def. \ref{def:symmetric_fourier_matrix}.
\end{lemma}
Applying the previous lemma to $B = T+E$ and $M= S(L')$, we can find a $D'$ in time $\poly(k,\log(1/\delta))$ such that the following holds with probability at least 0.99,
\begin{align*}
	\|T-F_{S(L')}D'F_{S(L')}^*\|_F&\leq \|T+E-F_{S(L')}D'F_{S(L')}^*\|_F + \|E\|_F\\
	&\lsim \|E\|_F + \|T-T_k\|_F + \delta \|T\|_F\\
	& \lsim \max \{\|E\|_F,\|T-T_k\|_F\} + \delta \|T\|_F,
\end{align*}
where in the second last inequality we used equation (\ref{eqn:cost_of_opt}). This finishes the proof of the main theorem.
\end{proof}
We end this section by presenting the proof of Lemma \ref{lem:weighted_regression}. This follows the proof of Lemma 5.7 in \cite{eldar2020toeplitz} with the modification that $D$ is restricted to be a diagonal, however we restate the proof here for completeness. 
\begin{proof}
Let $\wh{D} = \argmin_{\text{ Diagonal $D$}} \|B-F_M D F_M^*\|_F$. Let $S_1,S_2^T\in \mathbb{R}^{s\times d}$ be independent sampling matrices as per Claim A.1 of \cite{eldar2020toeplitz} for $\epsilon=\delta=0.1$ and the leverage score distribution of Corollary C.2 of \cite{eldar2020toeplitz}. Then $s= O(m\log^2(m))$. Let $D' = \argmin_{\text{ Diagonal $D$}} \|S_1BS_2-S_1F_M D F_M^*S_2\|_F$. $D'$ can be found in $\poly(s)=\poly(m)$ time. Our strategy to prove the lemma will have two steps. First we will state an inequality and show how that implies the lemma, and then we will prove the inequality. We will show that the following holds with probability at least $0.97$, for all diagonal $D\in \mathbb{R}^{m\times m}$,
	\begin{equation}\label{eqn:regression_cost_preservation}
		\|S_1F_MDF_M^*S_2 - S_1 B S_2\|_F  = (1\pm 0.1)\|F_MDF_M^* - B\|_F \pm 100 \|F_M \wh{D} F_M^* - B\|_F.
	\end{equation}
	Equipped with the previous inequality, we can use it to prove the lemma as follows,
	\begin{equation}
		\begin{aligned}
			\|F_M D' F_M^* - B\|_F &\leq 1.1\|S_1F_M D' F_M^* S_2 - S_1BS_2\|_F + 100\|F_M \wh{D}F_M^* - B\|_F\\
		   &\leq 1.1\|S_1F_M \wh{D} F_M^* S_2 - S_1BS_2\|_F + 100\|F_M \wh{D}F_M^* - B\|_F\\
			&\leq (1.1)^2 \|F_M \wh{D} F_M^*-B\|_F + (100\cdot 1.1+100)\|F_M \wh{D}F_M^* - B\|_F \\
			& \leq 212 \|F_M \wh{D}F_M^* - B\|_F,
		\end{aligned}
	\end{equation}
where the first and second last inequality follow by applying equation (\ref{eqn:regression_cost_preservation}) and the second inequality follows from the optimality of $D'$ for the subsampled regression problem. 

Now we focus on proving equation (\ref{eqn:regression_cost_preservation}). First using the triangle inequality we write $	\|S_1F_MDF_M^*S_2 - S_1 B S_2\|_F$ as follows,
\begin{equation}\label{eqn:triangle_inequality}
		\|S_1F_MDF_M^*S_2 - S_1 B S_2\|_F = \|S_1F_M DF_M^*S_2 - S_1 F_M \wh{D} F_M^*S_2\|_F \pm \|S_1 F_M \wh{D} F_M^*S_2 - S_1 BS_2\|_F.
\end{equation}
Now observe that since $S_1,S_2$ are independent leverage score sampling matrices, they are unbiased estimators of the norm of any vector in $\mathbb{R}^d$. That is, for any $X\in \mathbb{R}^{d\times r}$ for any $r$, $\mathbb{E}[\|S_1 X\|_F^2] = \mathbb{E}[\| X^TS_2\|_F^2] =\|X\|_F^2$ for both $i=1,2$. Thus applying Markov's inequality we get that with probability at least $0.99$,
\begin{equation*}
	\|S_1 F_M \wh{D}F_M^* - S_1 B\|_F^2 \leq 100\|F_M \wh{D}F_M^* - B\|_F^2.
\end{equation*}
Applying Markov's inequality again over the randomness of $S_2$, we get the following with probability at least $0.99$,
\begin{align*}
		\|S_1 F_M \wh{D}F_M^* S_2- S_1 BS_2\|_F^2&\leq 	100\|S_1 F_M \wh{D}F_M^* - S_1 B\|_F^2\\
		&\leq 	100^2\| F_M \wh{D}F_M^* - B\|_F^2.
\end{align*}
Thus taking a union bound over the randomness in $S_1,S_2$, we get that the following holds with probability at least $0.98$,
\begin{equation}\label{eqn:markov_bound}
	\|S_1 F_M \wh{D}F_M^* S_2- S_1 BS_2\|_F\leq 100\| F_M \wh{D}F_M^* - B\|_F.
\end{equation}
Finally, due to Claim A.1 and Corollary C.2 of \cite{eldar2020toeplitz}, since $S_1,S_2$ are independent leverage score sampling matrices taking $s=O(m\log^2(m))$ samples the following subspace embedding property holds with probability at least $0.99$,
\begin{align*}
	\|S_1 [F_M;F_M]y\|_2^2 &= (1\pm 0.01) \|[F_M;F_M]y\|_2^2 \quad \forall y\in \mathbb{C}^{2m} \text{ and } \\
    \| y^*[F_M;F_M]^*S_2\|_2^2&=(1\pm 0.01) \|[F_M;F_M]y\|_2^2 \quad \forall y\in \mathbb{C}^{2m}.
\end{align*}
This guarantee applied to $S_1$ implies the following,
\begin{equation*}
	\|S_1F_M DF_M^* - S_1F_M \wh{D}F_M^*\|_F = (1\pm 0.01)\|F_MD F_M^* - F_M\wh{D}F_M^*\|_F.
\end{equation*}
And then applying the subspace embedding guarantee for $S_2$ finally gives us the following,
\begin{align*}
	\|S_1F_M DF_M^*S_2- S_1F_M \wh{D}F_M^*S_2\|_F &= (1\pm 0.01)\|S_1F_M DF_M^* - S_1F_M \wh{D}F_M^*\|_F\\
	&= (1\pm 0.01)^2\|F_MD F_M^* - F_M\wh{D}F_M^*\|_F.
\end{align*}
Plugging the previous equation and \ref{eqn:markov_bound} back into \ref{eqn:triangle_inequality} we get \ref{eqn:regression_cost_preservation}. This completes the overall proof of Lemma \ref{lem:weighted_regression}.
\end{proof}
\subsection{Sublinear Time Low-Rank Approximation and Covariance Estimation.}\label{subsection:toeplitz_applications}
Our goal in this section is to present the proofs of the sublinear time low-rank approximation and covariance estimation results of Theorems \ref{thm:sublinear_time_lowrankapprox} and \ref{thm:sublinear_time_covariance_estimation} respectively. The proof of Theorem \ref{thm:sublinear_time_lowrankapprox} easily follows by applying Theorem \ref{thm:sublinear_time_recovery} for $E=0$. The main Lemma we will need to apply our framework to prove Theorem \ref{thm:sublinear_time_covariance_estimation} is the following that allows us to bound the magnitude of the noise.
\begin{lemma}\label{lem:low_rank_concentration}
Consider $d\times d$ PSD Toeplitz $T$. Let $k$ be an integer and $\epsilon>0$. Given samples $x_1,\ldots, x_s \sim \mathcal{N}(0,T)$,  let $X \in \mathbb{R}^{d\times s}$ be a matrix whose $i^{th}$ column is $x_i/\sqrt{s}$ for all $i\in [s]$. If $s=\wt O(k^4/\epsilon^2)$, then the following holds with probability at least $0.98$,

\begin{equation*}
	\|XX^T - T\|_F \lsim \sqrt{\|T-T_k\|_2 \tr(T)+ \frac{\|T-T_k\|_F \tr(T)}{k}}+ \epsilon \|T\|_2.
\end{equation*}
\end{lemma}
Assuming Lemma \ref{lem:low_rank_concentration}, we now present the proof of Theorem \ref{thm:sublinear_time_covariance_estimation}.

%\Cam{I don't think this looks correct. You can't just decide `ignore $\delta$' in the Big Oh notation since $\epsilon$ has a worse dependence. I think that we need toset $\delta = \epsilon/\poly(d)$ so that $\delta \norm{T}_F <= \epsilon \norm{T}_2$. And thus $\log(1/\delta) = O(\log (d/\epsilon))$ which should show up in the rank bound.}
\begin{proof}[Proof of Theorem \ref{thm:sublinear_time_covariance_estimation}]
	The proof follows easily from applying Theorem \ref{thm:sublinear_time_recovery} with $E = XX^T - T$, $\delta = \epsilon/\poly(d)$ and bounding the Frobenius norm of $E$ using Lemma \ref{lem:low_rank_concentration} with $s = \wt O(k^4/\epsilon^2)$. Note that $\delta \|T\|_F \leq \epsilon \|T\|_2$ since $\delta = \epsilon/\poly(d)$ thus $\log(1/\delta)=O(\log(d/\epsilon))$. This bounds the Vector Sample Complexity (VSC). Note that Theorem \ref{thm:sublinear_time_recovery} only accesses $\poly(k,\log(d/\epsilon))$ entries of $XX^T$ and any of its $(i,j)^{th}$ entry is equal to $\sum_{k=1}^s x_{k,i}x_{k,j}$. Thus each entry access to $XX^T$ requires reading 2 entries, the $i$ and $j^{th}$ entries, from each sample. Thus we get that the Entry Sample Complexity (ESC) is $\poly(k,\log(d/\epsilon))$.
\end{proof}
Now we proceed by presenting the proof of Lemma \ref{lem:low_rank_concentration}.
\begin{proof}[Proof of Lemma \ref{lem:low_rank_concentration}]
	Let $T = U\Sigma U^T$ be the eigenvalue decomposition of $T$ where $\Sigma \succeq 0$ is diagonal. Let $P_k = U_kU_k^T$ be the projection matrix onto the subspace spanned by the top-$k$ eigenvectors. Using the rotational invariance of the Gaussian distribution, we have that $X$ is distributed as $X \sim U\Sigma^{1/2} G$, where $G\in \mathbb{R}^{d\times s}$ is a matrix with each entry distributed independently as $\frac{1}{\sqrt{s}}\mathcal{N}(0,1)$. Then we upper bound $\|XX^T - T\|_F$ by the following three terms using the triangle inequality,
	\begin{equation*}
	\|XX^T - T\|_F \leq \|XX^T - P_k XX^T P_k\|_F + \|P_k XX^T P_k  - T_k\|_F+\|T_k - T\|_F.
	\end{equation*}
We finish the proof by taking a union bound over the following two claims, which bound the first and second terms of the expression above.
\begin{claim}\label{clm:pcp_lowrank}
	If $s=\wt O(k)$, then with probability at least $0.99$,
	\begin{equation*}
		\|XX^T - P_k XX^T P_k\|_F \lsim \sqrt{\|T-T_k\|_2 \tr(T)+ \frac{\|T-T_k\|_F \tr(T)}{k}}.
	\end{equation*}
\end{claim}
%Theond claim bounds the second term as follows.
\begin{claim}\label{clm:lowrank_covest}
	For any $\epsilon>0$. If $s = \wt O(k^4/\epsilon^2)$, then with probability at least $0.99$,
	\begin{equation*}
		\|P_k XX^T P_k - T_k \|_F \leq \epsilon \|T\|_2.
	\end{equation*}
\end{claim}
 Note that since $T$ is PSD, $\|T-T_k\|_F =\sqrt{\sum_{j=k+1}^d \lambda_j^2(T)}\leq \sqrt{\lambda_1(T)(\sum_{j=1}^d \lambda_j)}= \sqrt{\|T-T_k\|_2 tr(T)}$. Thus we ignore the third term $\|T-T_k\|_F$ in the upper bound on $\|XX^T-T\|_F$ in the big-Oh term of Claim \ref{clm:pcp_lowrank}. This completes the proof of Lemma \ref{lem:low_rank_concentration}.
\end{proof}
We now proceed by proving Claims \ref{clm:pcp_lowrank} and \ref{clm:lowrank_covest}.
\begin{proof}[Proof of Claim \ref{clm:pcp_lowrank}]
	We know that $X= U\Sigma^{1/2} G$ and $P_k$ is a rank-$k$ projection matrix. Thus, we can apply the projection cost-preserving sketch property of the Gaussian distribution from Lemma 12 and Theorem 27 of \cite{cohen2015dimensionality}. 

	As a result, if $s=\Omega(k/\gamma^2)$, then the following holds with probability at least 0.99,
	\begin{equation}\label{eqn:pcp_bound}
		\|P_k X - X\|_2^2 = (1\pm \gamma)\|P_k U\Sigma^{1/2}- U\Sigma^{1/2}\|_2^2 \pm \frac{\gamma}{k}\|P_kU\Sigma^{1/2} - U\Sigma^{1/2}\|_F^2.
	\end{equation}
Fix $\gamma = O(1)$. Now we use the fact that for any matrices $A,B$, $\|AA^*-BB^*\|_F \leq \|AB^*-BB^*\|_F+\|AA^*-AB^*\|_F \leq \|A-B\|_2(\|A\|_F+\|B\|_F)$. Applying this fact to $A= P_kX$ and $B = X$ we get the following,
\begin{align*}
	\|XX^T - P_kXX^TP_k\|_F &\leq \|P_kX- X\|_2 (\|X\|_F+\|P_kX\|_F)\\
	&\leq O \left( \sqrt{\|P_k U\Sigma^{1/2}- U\Sigma^{1/2}\|_2^2 + \frac{\|P_kU\Sigma^{1/2} - U\Sigma^{1/2}\|_F^2}{k}}\|X\|_F \right)\text{(from equation \ref{eqn:pcp_bound})} \\
	& = O \left( \sqrt{\|T^{1/2}-T_k^{1/2}\|_2^2 + \frac{\|T^{1/2}-T_k^{1/2}\|_F^2}{k}}\|X\|_F \right)\\
	& =O \left( \sqrt{\|T-T_k\|_2\|X\|_F^2 + \frac{\|T-T_k\|_F\|X\|_F^2}{k}} \right).\\
\end{align*}
In the second line, we used that $\|P_kX\|_F\leq \|X\|_F$, and the last line follows by observing that $U\Sigma^{1/2} = T^{1/2}, P_k U\Sigma^{1/2} = T_k^{1/2}$. Now observe that $\|X\|_F^2 = \|T^{1/2}G\|_F^2$. Thus by applying the Johnson-Lindenstrauss lemma to each row of $T^{1/2}$ and taking a union bound over the $d$ rows of $T^{1/2}$, we have that as long as $s=\Omega((1/\epsilon^2)\log(d/\delta))$, $\|X\|_F^2 \leq (1+\epsilon)\|T^{1/2}\|_F^2=(1+\epsilon)\tr(T)$ with probability $1-\delta$. Fixing $\epsilon=\delta = 0.001$, taking a union bound over this event and the projection cost-preserving sketch property, and plugging the bound $\|X\|_F^2 \leq O(\tr(T))$ into the last line of the derivation above, we finish the proof of the claim.
\end{proof}
Finally, we prove Claim \ref{clm:lowrank_covest}. This proof follows the proof strategy of Theorem 7.1. in \cite{lecturenotes}%\Cam{Don't do this with a footnote. Look up how to put a url in the citation itself. The footnote isn't proper.} \footnote{The lecture notes can be found \href{https://www.stat.cmu.edu/~arinaldo/Teaching/36755/F17/Scribed_Lectures/F17_0925.pdf}{here}.}, %\Cam{This needs to be a citation not a hyperlink!} 
 but adapted for the case when the covariance matrix is exactly rank $k$.
\begin{proof}[Proof of Claim \ref{clm:lowrank_covest}]
	For any $d\times d$ rank-$k$ matrix $A$, we have the following,
	\begin{equation*}
		\|A\|_2 = \max_{x\in \mathcal{S}^{k-1}}|x^TU_k^T A U_k x| \leq \frac{1}{1-2\epsilon}\max_{z\in \mathcal{N}_{\epsilon}} |z^TU_k^T A U_k z|,
	\end{equation*}
where $U_k$ is the matrix containing the top-$k$ eigenvectors of $A$, $\mathcal{S}^{k-1}$ is the unit sphere in $k$ dimensions, and $\mathcal{N}_{\epsilon}$ is an $\epsilon$-net of $\mathcal{S}^{k-1}$. Fix $\epsilon = 1/4$, then using Lemma 5 of \cite{Woodruff:2014tg} which uses volumetric arguments to bound the size of $\mathcal{N}_{\epsilon}$, we have that $|\mathcal{N}_{\epsilon}| \leq 17^k$. Now let $E = P_kXX^TP_k - T_k = T_k^{1/2}GG^T T_k^{1/2} - T_k$. Observe that $E$ has rank at most $k$. Letting $U_k$ be the matrix of the top-$k$ eigenvectors of $E$, we have the following,
\begin{equation*}
	\|E\|_2 \leq 2\max_{x\in \mathcal{N}_{\epsilon}} |x^TU_k^T EU_kx|.
\end{equation*}
This implies the following,
\begin{align*}
	\mathbb{P}(\|E\|_2> t) &\leq \mathbb{P}(\max_{x\in \mathcal{N}_{\epsilon}} |x^TU_k^T EU_kx| >t/2)\\
	& \leq \sum_{x\in \mathcal{N}_{\epsilon}}\mathbb{P}(|x^TU_k^T EU_kx| >t/2).
\end{align*}
Thus we need to focus on upper bounding $\mathbb{P}(|x^TU_k^T EU_kx| >t/2)$ for a fixed $x\in \mathcal{S}^{k-1}$. Let $y = U_k x$, which implies that $y\in \mathcal{S}^{d-1}$. Thus we need to bound $\mathbb{P}(|y^T Ey| >t/2)$ for some $y\in \mathcal{S}^{d-1}$. We have the following,
\begin{align*}
 y^T Ey &= \frac{1}{s}\sum_{i\in [s]} \left(y^T T_k^{1/2}g_i g_i^T (T_{k}^{1/2})^T y - y^T T_k y\right)\\
  &= \frac{1}{s}\sum_{i\in[s]}\left(Z_i^2 - \mathbb{E}[Z_i^2]\right),
\end{align*}
where each $Z_i = y^T T_k^{1/2}g_i$ and $g_i \sim \mathcal{N}(0,I_{d\times d})$. Thus each $Z_{i}\sim \mathcal{N}(0,\sigma^2)$ where $\sigma^2 = \|y^TT_k^{1/2}\|_2^2$. As a result, $\frac{1}{s}\sum_{i\in[s]}\left(Z_i^2 - \mathbb{E}[Z_i^2]\right)$ is a chi-squared random variable. Using standard chi-squared concentration bound of Laurent-Massart \cite{laurent2000adaptive}, we have the following,
\begin{align*}
	\mathbb{P}(|y^T E y|\geq t/2) \leq \exp \biggl \{ -\Omega \left(s \min \left( \frac{t^2}{\sigma^4},\frac{t}{\sigma^2}\right)\right) \biggr \}.
\end{align*}
Thus, setting $t = \epsilon \sigma^2$ and $s = O\left(\frac{1}{\epsilon^2}\log\left(\frac{|\mathcal{N}_{1/4}|}{\delta}\right)\right)$ yields that the above upper bound on the failure probability is at most $\delta/|\mathcal{N}_{1/4}|$. Therefore,
\begin{equation*}
	\mathbb{P}(\|E\|_2 >\epsilon \sigma^2) \leq \delta.
\end{equation*}
Thus we get that if $s = \wt O(k/\epsilon^2)$, then with probability at least $0.99$, 
\begin{equation*}
	\|E\|_2 \leq \epsilon \|y^T U_k \Sigma^{1/2}\|_2^2 \leq \epsilon tr(T_k) \leq \epsilon k \|T\|_2.
\end{equation*}

This further implies that $\|E\|_F \leq \epsilon k^{1.5}\|T\|_2$. Setting $\epsilon = \epsilon/k^{1.5}$, we conclude the proof of the claim.
\end{proof}